\documentclass[10pt,twocolumn]{IEEEtran}

\usepackage[utf8]{inputenc}

\usepackage{multicol}
\usepackage{pdfpages}
\usepackage[bookmarks=false]{hyperref}
\usepackage{url}
\usepackage[usenames,dvipsnames]{pstricks}
\usepackage{epsfig}
\usepackage{pst-grad} 
\usepackage{pst-plot} 

\usepackage{amsmath}
\usepackage{amssymb}

\usepackage{color}
\usepackage{pstricks}
\usepackage{cite}
\usepackage{algorithm}

\usepackage{amsthm}
\usepackage{amssymb}
\usepackage{psfrag}
\usepackage{float}
\usepackage{stfloats}
\usepackage{wrapfig}
\usepackage{graphicx}

\usepackage{amsmath}
\usepackage{amsfonts}
\usepackage{amssymb}
\usepackage{blindtext}
\usepackage{array}
\usepackage{arydshln}
\usepackage{pstricks,pst-plot}
\usepackage{pst-bezier}
\usepackage{pst-math}
\usepackage{algorithm,algcompatible,amsmath}
\usepackage{float}
\usepackage[mathscr]{eucal}
\usepackage{enumerate}

\usepackage{float}
\usepackage[caption = false]{subfig}

\DeclareFontFamily{OT1}{pzc}{}
\DeclareFontShape{OT1}{pzc}{m}{it}{<-> s * [1.35] pzcmi7t}{}
\DeclareMathAlphabet{\mathpzc}{OT1}{pzc}{m}{it}

\algnewcommand\REQUIRED{\item[\textbf{Required:}]}%
\algnewcommand\INPUT{\item[\textbf{Input:}]}%
\algnewcommand\OUTPUT{\item[\textbf{Output:}]}%
\newtheorem{definition}{Definition}

\def\bX{{\bf X}}

\def\bg{{\bf g}}

\def\bX{{\bf X}}

\def\bn{{\bf n}}

\def\bz{{\bf z}}

\def\by{{\bf y}}
\def\bY{{\bf Y}}

\def\bX{{\bf X}}

\def\bh{{\bf h}}
\def\bb{{\bf b}}

\def\bS{{\bf S}}

\def\bv{{\bf v}}

\def\bB{{\bf B}}

\def\bI{{\bf I}}

\def\bv{{\bf v}}

\def\bB{{\bf B}}
\def\ba{{\bf a}}

\def\bs{{\bf s}}

\def\bx{{\bf x}}

\def\balpha{{\boldsymbol \alpha}}

\def\btheta{{\boldsymbol \theta}}

\def\brho{{\boldsymbol \rho}}

\def\b0{{\bf 0}}

\DeclareMathAlphabet\mathbfcal{OMS}{cmsy}{b}{n}
\DeclareMathAlphabet\mathbfbb{OMS}{cmsy}{b}{n}

\DeclareMathAlphabet\mathbfcal{OMS}{cmsy}{b}{n}

\newcommand{\refappendix}[1]{\hyperref[#1]{Appendix~\ref*{#1}}}

\newtheorem{lemma}{Lemma}

\title{\LARGE Deep Generative Models for Downlink Channel Estimation in FDD Massive MIMO Systems}
\author{ Javad Mirzaei$^{\star}$, Shahram ShahbazPanahi$^{\star\dagger}$, Raviraj Adve$^{\star}$, and Navaneetha Gopal$^{\star}$\\

 {\small  $^{\star }$Department of Electrical and Computer Engineering, University of Toronto, Canada\\

 $^{\star\dagger}$Department of Electrical and Computer Engineering, Ontario Tech University, Canada}
}

\begin{document}

\maketitle
\begin{abstract}\label{Abs}
It is well accepted that acquiring downlink channel state information in frequency division duplexing (FDD) massive multiple-input multiple-output (MIMO) systems is challenging because of the large overhead in training and feedback. In this paper, we propose a deep generative model (DGM)-based technique to address this challenge. Exploiting the partial reciprocity of uplink and downlink channels, we first estimate the frequency-independent underlying channel parameters, i.e., the magnitudes of path gains, delays, angles-of-arrivals (AoAs) and angles-of-departures (AoDs), via uplink training, since these parameters are common in both uplink and downlink. Then, the frequency-specific underlying channel parameters, specifically, the phase of each propagation path, are estimated via downlink training using a very short training signal. In the first step, we incorporate the underlying distribution of the channel parameters as a prior into our channel estimation algorithm. We use DGMs to learn this distribution. Simulation results indicate that our proposed DGM-based channel estimation technique outperforms, by a large gap, the conventional channel estimation techniques in practical ranges of signal-to-noise ratio (SNR). In addition, a near-optimal performance is achieved using only few downlink pilot measurements.

\end{abstract}

\section{Introduction}
\subsection{Motivation}
Massive multiple-input multiple-output (MIMO) is a key technology in helping to meet the demands to be made of the next (fifth) generation of wireless technologies \cite{Marzetta_noncoop, Marzetta_scaleUp}. This technology can be deployed in time-division duplex (TDD) mode, where the uplink and downlink communication occur in the same frequency band but at different time slots, and frequency-division duplex (FDD), where the uplink and downlink operate simultaneously on different frequency bands. To unlock the full potential of massive MIMO in both FDD and TDD, we require an accurate channel state information (CSI) between the base station (BS) and the user equipment (UE).

In TDD massive MIMO systems, CSI acquisition relies on the assumption of channel reciprocity between the uplink and downlink. However, due to calibration errors between the uplink and downlink radio frequency (RF) chains, such a channel reciprocity may not hold \cite{1404883}. Additionally, pilot contamination, caused by the use of non-orthogonal pilot signals in neighboring cells, is another performance-limiting factor in TDD-based systems  \cite{pilot_contam_Jose}. Importantly, realizing the backward compatibility of FDD massive MIMO communication (e.g., with Long-Term Evolution (LTE)) has drawn considerable attention toward FDD from both academia and industry \cite{Larsson6736761}.

In FDD, due to different band of frequencies in the uplink and downlink, the downlink channel is neither the same, nor can it be inferred from the uplink channel without any downlink training. Traditionally, a UE estimates its own downlink channel from the received pilot symbols transmitted from the BS and feeds the estimated downlink channel information back to the BS for the subsequent signal transmission and resource allocation. This approach is practical in current (prior to 5G) generation of networks, where only a few antennas are used at the BS, allowing for orthogonal pilots and a small feedback overhead. However, in massive MIMO systems, due to the large overhead with the use of orthogonal pilots in massive arrays, as well as the huge required feedback, this approach may not be applicable. Therefore, it is an urgent requirement in FDD massive MIMO systems to reduce the amount of pilot symbols needed for downlink channel estimation.

\subsection{Related Work} \label{related_works}

The previous attempts in FDD channel estimating mainly focused on adapting the downlink arrays using the second-order statistics of the uplink measurements\cite{905894, 609909}. In recent years and in the context of massive MIMO, there are attempts aiming at reducing the feedback overhead by exploiting the sparsity in the underlying channel parameters. In particular, considering a parametric channel model, where the  channel is characterized by the parameters of only a few dominant paths, such as gains, direction-of-arrival (DoA), and direction-of-departure (DoD)\cite{7400949}. Given such a model, one may consider array processing techniques (such as SAGE \cite{753729}) in estimating the channel parameters. However, these techniques are mainly based on the EM algorithm which requires the user to know the likelihood function of pilot observations and the channel parameters, which may not be available. Also, none of the array processing techniques \cite{17564,32276} has been adopted in the context of FDD downlink channel estimation. The reason is that these methods usually work under certain assumptions, (for example the number of multipaths to be smaller than the number of transmit and the number of receive antennas) which may not be satisfied in FDD systems~\cite{8481590}. Even employing DoA estimation techniques is challenging. Most
DoA estimation techniques rely on the covariance matrix of the received signal. Constructing (and then
inverting) such covariance matrices takes a lot of computation effort which limits the applicability of these techniques in fast fading scenarios. These techniques also usually require the number of paths be lower than the number of antennas.

Considering millimeter wave (mmWave) frequencies, the channel tends to exhibit sparsity in the angular domain \cite{5454399, AlkhateebHBest}. Leveraging this sparsity, the channel can be reformulated using the channel gains, DoAs, and DoDs. One can recover these underlying channel parameters using compressive sensing (CS) techniques.

In a CS framework, it is assumed that the UE compresses and feeds back the received pilot signal to the BS, where different CS-based techniques can be employed to recover the sparse channel parameters \cite{5454399, 5621984, 7458188, 6816089, 8437163, 8697125, 8648511}. In a multi-user massive MIMO setting, the authors of  \cite{6816089} develop a joint orthogonal matching pursuit (OMP) algorithm to recover the channel parameters at the BS. To further reduce the pilot overhead, block OMP is developed in \cite{7564736} based on a user grouping technique, which relies  on the assumption that the users in  each group have the same channel correlation matrix. The authors of \cite{8298537} proposed a sparse Bayesian learning technique for sparse channel recovery followed by an off-grid refinement. Note that in a practical setting, and due to the limited scattering in propagation environment, different user links tend to share some common scatterers. While each UE channel matrix is sparse in angular domain, they may exhibit a common sparsity pattern. This observation motivated the authors of  \cite{8038934} to propose a variational Bayesian inference based approach for channel estimation, where Gaussian mixture model is exploited to capture the individual sparsity in each channel matrix.  There are works that attempt to improve the performance of CS-based techniques by incorporating a priori information into the recovery algorithm \cite{7390062, 7893736, 7564429}. It is worth mentioning that the CS-based techniques still suffer from some limitations. In particular, they require strong channel sparsity in DFT-basis which is not strictly held in some cases. In fact, we do not even know the basis that yields the most sparse representation. While they also require a large number of pilots, they are often iterative and computationally intensive during decoding, which may lead to long delays.

As an alternative to CS-based techniques, there are techniques relying on the spatial reciprocity between the uplink and downlink channels \cite{8648511, 8697125, 7524027}, operating on close-by carrier frequencies have been proposed. Given the fact that the uplink and downlink communication occur in the same propagation environment, the uplink channel estimates can be used in the estimation of downlink channel. In this way, difficulties in downlink channel estimation is shifted to the BS.

Uplink-downlink reciprocity can be considered in two ways. In one way, one can assume full-reciprocity where the multi-path components of channel (including phase, amplitude, delay, angle of arrival, departure, etc.) are the same in both uplink and downlink \cite{8778685, 9013411, DinaR2F2}. Based on this assumption, the authors of  \cite{DinaR2F2}, proposed to completely eliminate downlink training and feedback in LTE systems. However, there is not enough evidence to confirm such a full-reciprocity assumption in real world data.
Indeed, it is shown via measurements and theoretical investigations that not all channel parameters are reciprocal in uplink and downlink. In particular, there is no reciprocity between the phases of different multi-path components for FDD uplink and downlink channels  \cite{9322570}. Intuitively, this is expected given the sensitivity of the phase to operating frequency.  Alternatively, uplink and downlink channels are only \emph{partially} reciprocal. This implies that downlink training and feedback is inevitable. Partial reciprocity has been leveraged differently in the literature. In \cite{9014542}, the authors consider a spatial domain representation of channel\footnote{The channel is expressed by a small number of multi-path component such as AoA/AoDs and path gains}, and exploit the uplink-downlink angle reciprocity \cite{7727995}, where the AoA/AoD in the uplink and downlink are assumed to be the same. Then, the downlink channel estimation boils down to path gain estimation which can be done within a much lower pilot overhead.

As an alternative approach, deep learning (DL)-based techniques have been considered in recent studies published on channel estimation for FDD-based massive MIMO. Deep learning is a powerful tool to explore the underlying complex structure of data. This type of learning has been widely applied in various wireless communication problems, such as data detection \cite{8052521, 8680715}, channel estimation \cite{8672767}, beamforming \cite{9129762}, and hybrid precoding \cite{9048966, 9121328}, and have shown promising performance compared to conventional techniques. In the context of downlink channel estimation in FDD systems, in \cite{8647328}, the whole downlink CSI is estimated using only the CSI obtained over a small set of BS antennas via linear regression and support vector machines (SVM). This is based on the assumption that the channel among different BS antennas is correlated. In the same line of the work, the authors of \cite{9048929}, proposed a deep neural network (DNN)-based technique to map the channel both in space and frequency. This allows for the prediction of the channel of a set of BS antennas and a frequency band from the observation of different set of antennas and different frequency band. Inspired by the position-to-channel mapping investigated in \cite{9048929}, the authors of \cite{8795533} proposed a complex-valued neural network (named SCNet) to directly map the uplink channel to its downlink counterpart, without requiring any downlink pilot transmission. A similar technique with reduced complexity has been proposed in \cite{9057648}.

\subsection{Contribution and Methodology}
In this paper, we consider a single cell where a BS with massive number of antennas communicates with a UE in FDD mode. Using the fact that the channel matrix over each subcarrier is a function of a smaller set of parameters, namely, the number of propagation paths, the path gains, phases, delays, as well as AoAs and AoDs \cite{7400949}, we estimate these parameters instead of the channel matrix directly. These parameters depend on the physical properties of the propagation environment and on the operating frequencies, and importantly, they are independent of the number of antennas at the BS as well as the number of subcarriers \cite{6515173, 6834753, 8798669, 9034179}. Unlike the conventional techniques, where a long training sequence is transmitted over all antennas and over all subcarriers, we estimate these underlying channel parameters using a short training signal over a much smaller set of antennas and subcarriers.

Motivated by the partial reciprocity of uplink and downlink channels \cite{9322570, 9014542, 7136154}, we use the following steps to estimate the downlink channel: I) we estimate the frequency-independent underlying channel parameters, namely, the magnitudes of path gains, delays, AoAs and AoDs during the uplink training; and II) the frequency-specific underlying channel parameters, i.e., the phase of each propagation path, are estimated via downlink training. Using this strategy, we shift the burden in FDD downlink channel estimation to the BS, which is, anyway, responsible for uplink channel estimation. In the first step, we use the least squares (LS) estimation approach to estimate the frequency-independent parameters. The optimization problem in this step is difficult to solve analytically, mainly due to non-linear and non-convex structure of its objective function. To address this problem, we use \emph{deep generative models (DGMs)}\footnote{In particular, we use the generative adversarial network (GAN)}, to capture the distribution of the underlying parameters, and then use it as a prior to simplify the optimization problem. In the second step, we use the frequency-independent parameters, estimated in the first step, to estimate the frequency-specific parameters via an LS technique. In both steps, the optimization problem is carried out numerically using the gradient descent algorithm. The contributions of this paper are itemized below:
\begin{itemize}
  \item \textbf{Learning the distribution of channel parameters}: The unknown underlying distribution of the channel parameters is some function of the propagation environment, and therefore, it is complex and difficult to obtain analytically. We explain how to use DGMs to learn this distribution. In particular, we use the GAN structure to find a deterministic mapping function (i.e., a generator) that is capable of drawing samples from the underlying distribution of channel parameters, by feeding it with samples from a low-dimensional standard Gaussian distribution.

  \item \textbf{FDD downlink channel estimation}: We present a technique to show  how to exploit the uplink-downlink partial reciprocity, thereby reducing the pilot and feedback overhead in downlink. In addition, the sparsity assumption in channel parameters is relaxed. Instead, using a generator, obtained from a GAN, we incorporate the learned structure of channel parameters as a \emph{prior} into our channel estimation procedure. By doing so, the optimization problem operates in a low-dimensional subspace, whose dimensionality is defined by the generator and, importantly, it is independent of number of received pilots. Therefore, we achieve a significant reduction in computational complexity as well as CSI feedback overhead.
  \item \textbf{Convergence analysis}: We analytically prove the convergence of our estimation by showing that the gradient of the estimation objective function is Lipschitz-continuous, and therefore, the convergence of steepest descent algorithm is guaranteed.

\end{itemize}

The simulation results indicate that our proposed DGM-based channel estimation outperforms the conventional channel estimation technique in practical ranges of signal-to-noise ratio (SNR). This is mainly due to the capabilities of the generator in representing the underlying distribution of the channel parameters. Incorporating this prior knowledge into our channel estimation significantly improves the performance even at low SNR. This indicates how the proposed technique is resilient to the noise level. Additionally, for fixed SNR, we show that the proposed technique yields a near-optimal performance using only few pilot measurements. This can significantly reduce the pilot overhead in FDD massive MIMO systems.

Our work in this paper differs from the CS-based techniques in the sense that we do not assume any sparsity in the underlying channel parameters. Instead, we relax this constraint and consider that the channel parameters have a particular structure which is not necessarily sparse. We capture this structure using a generator, and then, incorporate it as a prior into our channel estimation process to improve the accuracy and reduce the pilot overhead. \label{comp_rev}The DNN-based techniques in~\cite{9048929, 8795533, arnold2019enabling, 9057648} mainly rely on direct channel mapping between the uplink and downlink, without requiring any downlink training. The performance of these techniques is highly affected by the uplink-downlink frequency separation, hardware impairments, shadow fading, and SNR. While we have not shown numerically, \emph{in theory}, our proposed technique can address such deficiencies by incorporating a DGM-learned prior into channel estimation process. Unlike the work in \cite{9154297}, in this paper, we study the channel estimation problem in FDD massive MIMO systems. Furthermore, instead of the channel distribution, we learn the underlying distribution of channel parameters.

\emph{Organization}:
 We introduce the system model, including the received signal model, channel model, and FDD partial reciprocity in Section \ref{Sec:sys_model}. Section \ref{Sec:ch_est} introduces the channel estimation problem, where we formulate the problem followed by our estimation technique. For the sake of the paper being self-contained, DGMs are briefly introduced in Section \ref{Sec:DGMs}. In Section \ref{complexity_and_ident}, we analytically study the convergence, complexity and identifiability of our estimation algorithm. Section \ref{Sec:Sims} presents implementation details and simulation results to illustrate the efficacy of the proposed technique. Section \ref{Sec:conclusions} provides conclusions. Finally, Appendix \ref{AppA} provides some preliminary definitions and lemmas used to prove Lemma \ref{J_lipschitz} in Appendix \ref{lip}.

\emph{Notation}:
We use bold upper and lower-case letters to denote matrices and vectors, respectively. $\mathbb{E}_{\bx}\left[\cdot\right]$ denotes statistical expectation over the random variable $\bx$. $\bz \backsim \mathbb{P}_z(\bz)$ denotes that the random vector $\bz$ follows the
probability distribution (pdf) of $\mathbb{P}_z(\bz)$
. The transpose operation is represented as $(\cdot)^T$; $\bI_N$ denotes an $N \times N$ identity matrix; $\mathbb{R}_+$ denotes the set of non-negative real numbers; $\|\cdot\|$ denotes the norm-2 operation; $|x|$ represents the absolute value of $x$; $\bX\preccurlyeq \bY$ means $\|\bX\| \leq \|\bY\|$; $|\mathcal{I}|$ represents the cardinality of set $\mathcal{I}$; and ${\mathcal{{I}}^c}$ the compliment set of ${\mathcal{{I}}}$; ${\mathcal{\bar{I}}}$ represents a sorted array whose entries are in the set $\mathcal{I}$. We use $\Big[ \bx_i\Big]_{i \in \mathcal{I}}$ to denote a vector whose $j$-th block entry, for $j=1,2, \cdots, |\mathcal{I}|$, is given by $\bx_{j}$, for $j \in {\mathcal{\bar{I}}}$. In a similar way, $\Big[ \bX_i\Big]_{i \in \mathcal{I}}$ denotes a matrix whose $j$-th block entry, for $j=1,2, \cdots, |\mathcal{I}|$, is given by $\bX_{j}$, for $j \in {\mathcal{\bar{I}}}$.

\section{System Model} \label{Sec:sys_model}
We consider a single-cell single-user communication system. The BS is equipped with a uniform linear array (ULA) with $M\gg 1$ antenna elements, while the UE is single-antenna. The communication between the BS and UE is performed in FDD mode. In the uplink, the UE communicates with the BS at frequency $f_{\rm up}$, while in the downlink, the BS communicates with the UE at frequency $f_{\rm dl}$. The frequency difference between $f_{\rm up}$ and $f_{\rm dl}$ is assumed to be relatively small. Both uplink and downlink frequency bands are of bandwidth $B$.
\subsection{Received Signal Model}
We assume that orthogonal frequency division duplex (OFDM) technology is used in both uplink and downlink commutation with $K$ subcarriers. Let $\mathcal{K}_{\rm up}$  denote the set of subcarrier indices used for uplink training. During the uplink training at the $k$-th subcarrier, the UE transmits training symbol $s_k$, $k \in \mathcal{K}_{\rm up}$, where $|s_k|^2=P_T$, and $P_T$ is the transmit power. The received signal at BS over the $k$-th subcarrier is given by
\begin{align} \label{UL_received signal}
\by^{\rm up}_k\triangleq \bh_k^{\rm up}s_k +\bn^{\rm up}_k, \,\,\, k \in \mathcal{K}_{\rm up},
\end{align}
where $\bh_k^{\rm up}$ is an $M \times 1$ uplink channel vector between the BS and the UE and $\bn^{\rm up}_k$ is an $M \times 1$ noise vector at the $k$-th subcarrier that is drawn independently and identically from a complex Gaussian distribution with zero mean and variance $\sigma_{\rm n}^2$.

In the downlink, let $\mathcal{K}_{\rm dl}$ and $\mathcal{M}_{\rm dl}$ denote, respectively, the set of subcarrier indices and the set of antenna elements dedicated for training. We assume that $p$ training symbols are transmitted over the $m$-th antenna ($m \in \mathcal{M}_{\rm dl}$) over the $k$-th subcarrier ($ k \in \mathcal{K}_{\rm dl}$). The received signal at the UE over the $k$-th subcarrier at the $i$-th time slot is given by
\begin{align} \label{DL_received signal_ith}
y^{(i),\rm dl}_k\triangleq \bs^{(i),\rm dl}_k{\bh_k^{\rm dl}} +n^{(i),\rm dl}_k, \,\,\, k \in \mathcal{K}_{\rm dl}, \; i=1,2, \cdots, p,
\end{align}
where $\bh_k^{\rm dl}$ is an $|\mathcal{M}_{\rm dl}| \times 1$ downlink channel vector at the $k$-th subcarrier between the BS and the UE, $\bs^{(i),\rm dl}_k \triangleq \left[ s^{(i),\rm dl}_{k,1}\;\;  s^{(i),\rm dl}_{k,2} \;\; \cdots \;\; s^{(i),\rm dl}_{k,|\mathcal{M}_{\rm dl}|}\right]$ is the $1 \times |\mathcal{M}_{\rm dl}|$ vector of downlink training symbols transmitted at the $i$-th time slot over the $k$-th subcarrier across all antenna elements in the set $\mathcal{M}_{\rm dl}$. We assume that $\|\bs^{(i),\rm dl}_k \|^2= P_T$. $n^{(i),\rm dl}_k$ denotes the noise term at the $i$-th time slot over the $k$-th subcarrier, drawn independently and identically from a complex Gaussian distribution with zero mean and variance $\sigma_{\rm n}^2$. Note that, in general $|\mathcal{M}_{\rm dl}|\leq M$. For the case when $|\mathcal{M}_{\rm dl}|< M$, we assume that the antenna elements in the set  $\mathcal{M}_{\rm dl}^c$ do not transmit during the downlink training.

Collecting the received signal across all $p$ training time slots over the $k$-th subcarrier, we can write
\begin{align} \label{DL_received signal}
\by^{\rm dl}_k\triangleq \bS^{\rm dl}_k{\bh_k^{\rm dl}} +\bn^{\rm dl}_k, \,\,\, k \in \mathcal{K}_{\rm dl},
\end{align}
where $\bS^{\rm dl}_k $ is a $p \times |\mathcal{M}_{\rm dl}|$ matrix of downlink training symbols with $\bs^{(i),\rm dl}_k$ on its $i$-th row, $i=1, 2, \cdots, p$, $\by^{\rm dl}_k \triangleq \left[ y^{(1),\rm dl}_k \;\; y^{(2),\rm dl}_k \;\; \cdots \;\; y^{(p),\rm dl}_k\right]^T$, and $\bn^{\rm dl}_k \triangleq \left[ n^{(1),\rm dl}_k\;\; n^{(2),\rm dl}_k\;\; \cdots\;\; n^{(p),\rm dl}_k \right]^T$. Next, we present our channel model.

\subsection{Channel Model}
To characterize the wireless channel between the BS antenna array and the UE, we consider the following geometric channel model. We assume that the propagation channel between the BS and the UE in the uplink consists of $L^{\rm up}$ paths. Through the $l$-th path, the signal travels the distance $d_l$ between the UE and the BS. Also, let $\alpha_l^{\rm up} \in \mathbb{R}_+$, $\phi_l^{\rm up}\in [0, 2\pi]$, $\theta_l^{\rm up}\in [0, 2\pi]$, and $\tau_l^{\rm up} \in \mathbb{R}_+$, for $l=1,2, \ldots, L^{\rm up}$, denote the random path gain, the random phase change, the random azimuth angle of the signal received, and the random delay corresponding to the $l$-th path in the uplink, respectively. Using this notation, the channel response between the UE and the BS at the $k$-th subcarrier is given by \cite{8697125,DeepMIMO1}
\begin{align} \label{UL_ch_model_def}
\bh_k^{\rm up}=\sum_{l=1}^{L^{\rm up}}\alpha_l^{\rm up}e^{j\left(\phi_l^{\rm up} + \frac{2\pi k}{K}\tau_l^{\rm up} B\right)}\ba(\theta_l^{\rm up}, \lambda_k^{\rm up}),\,\,\, k \in \mathcal{K}_{\rm up},
\end{align}
where $\lambda_k^{\rm up}\triangleq \frac{c}{f_c^{\rm up}+kB/K}$ is the wavelength of the $k$-th subcarrier in the uplink, and $c$ is the speed of light \footnote{Note that the carrier phase shift $2 \pi f_c^{\rm up} \tau_l $ is absorbed in the random phase $\phi_l^{\rm up}$.}. Since $kB/K$ is very small compared to $f_c^{\rm up}$, we ignore the subcarrier index in the array response $\ba(\theta_l, \lambda)$.
Denoting $\bar{d}$ as the antenna spacing in the ULA, the array response in \eqref{UL_ch_model_def} is given by
\begin{align} \label{array_resp_def}
\ba(\theta_l, \lambda)\triangleq \left[1\;\; e^{j\frac{2\pi}{\lambda}\bar{d}\sin{\theta_l} }\;\; \cdots\;\; e^{j\frac{2\pi }{\lambda}\bar{d}(M-1)\sin{\theta_l} } \right]^T.
\end{align}

Similarly, the downlink communication channel at the $k$-th subcarrier is given by
\begin{align} \label{DL_ch_model_def}
\bh_k^{\rm dl}=\sum_{l=1}^{L^{\rm dl}} \alpha_l^{\rm dl} e^{j\left(\phi_l^{\rm dl} + \frac{2\pi k}{K}\tau_l^{\rm dl} B\right)}\bb(\theta_l^{\rm dl}, \lambda^{\rm dl}),\,\,\,\, k \in \mathcal{K}_{\rm dl},
\end{align}
where $L^{\rm dl}$ is number of path in the downlink, $\lambda^{\rm dl}\triangleq \frac{\nu}{f_c^{\rm dl}}$ is the wavelength of the downlink carrier frequency, $\alpha_l^{\rm dl} \in \mathbb{R}_+$, $\phi_l^{\rm dl}\in [0, 2\pi]$, $\theta_l^{\rm dl}\in [0, 2\pi]$, and $\tau_l^{\rm dl} \in \mathbb{R}_+$, for $l=1,2, \ldots, L^{\rm dl}$, respectively denote the random path gain, the random phase change, the random azimuth angle of the received signal and the random delay corresponding to the $l$th path in the downlink. $\bb(\theta_l, \lambda)$ is a subvector of $\ba(\theta_l, \lambda)$ where its $i$th entry is $e^{j\frac{2\pi}{\lambda}\bar{d} (\mathcal{M}_{\rm dl}^i -1)\sin{\theta_l} }$, with $\mathcal{M}_{\rm dl}^i$ being the $i$th smallest member of set $\mathcal{M}_{\rm dl}$.

\subsection{FDD Partial Reciprocity}\label{FDD_partial_reciprocity}
In FDD communication, since uplink and downlink communication between the BS and UE occurs over different frequency bands, reciprocity between $\bh_k^{\rm up}$ and $\bh_k^{\rm dl}$ does not hold in general. However, since uplink and downlink channels share the same propagation environment, it has been shown that partial reciprocity exists between uplink and downlink channels \cite{8697125}.

It is observed via measurements, and verified using theoretical analysis that a portion of uplink and downlink channel parameters are frequency-independent \cite{9322570}. Specifically, since the signal of each propagation path travels the same distance at the same speed in both uplink and downlink communication link, the delay of each propagation path is the same in both uplink and downlink, i.e., $\tau_l^{\rm dl}= \tau_l^{\rm up}\triangleq \tau_l$ \cite{9322570}. Furthermore, it is shown, via both the measurement and ray tracing simulations, that the directional angle of each communication path are the same in both uplink and downlink, i.e., $L^{\rm dl}=L^{\rm up}\triangleq L$, $\theta_l^{\rm dl}= \theta_l^{\rm up}= \theta_l$\cite{7136154, 9322570}. It is also shown that, while the gain of each communication path is not exactly the same in both uplink and downlink in general, the downlink path gains are very close to that of uplink. According to Figs. 5 and 6 of \cite{9322570}, the average power delay profile of the uplink and downlink is shown to be approximately the same. Our analysis on the DeepMIMO dataset \cite{DeepMIMO1} also indicates that the \emph{relative} error between $\alpha_l^{\rm dl}$ and $\alpha_l^{\rm up}$ is only $0.8\%$\label{8_percent}. That is why, in this paper, we use the following approximation $\alpha_l^{\rm dl}\approx \alpha_l^{\rm up}= \alpha_l$. On the other hand, the existing measurements and analysis does not provide enough evidence that $\phi_l^{\rm dl}$ and $\phi_l^{\rm up}$ to be the same. This implies that the channel translation requires downlink training.
\section{Channel Estimation} \label{Sec:ch_est}
Leveraging the partial reciprocity of channel parameters in FDD, we develop a strategy for downlink channel estimation as explained in the sequel. We aim to estimate the frequency-independent parameters during the uplink training, while the frequency-specific channel parameters are being estimated via downlink training. To better characterize the channel estimation process, let us define $\balpha\triangleq[\alpha_1\;\; \alpha_2\;\; \ldots\;\; \alpha_L]^T$, $\boldsymbol{\phi}^{\rm up}\triangleq [\phi_1^{\rm up}\;\; \phi_2^{\rm up}\;\; \ldots \;\; \phi_L^{\rm up}]^T$, $\boldsymbol{\phi}^{\rm dl}\triangleq [\phi_1^{\rm dl}\;\; \phi_2^{\rm dl}\;\; \ldots\;\; \phi_L^{\rm dl}]^T$, $\boldsymbol{\tau} \triangleq [\tau_1\;\; \tau_2\;\; \ldots\;\; \tau_L]^T$, and $\btheta=[\theta_1\;\; \theta_2\;\; \ldots\;\; \theta_L]^T$, while the tuples
\begin{align} \label{x_def}
\bx^{\rm up}&\triangleq \left(\bx, {\boldsymbol{\phi}^{\rm up}}\right)\triangleq \left(\balpha, \boldsymbol{\tau}, \btheta, {\boldsymbol{\phi}^{\rm up}}\right),\nonumber\\
\bx^{\rm dl}&\triangleq \left(\bx, {\boldsymbol{\phi}^{\rm dl}}\right)\triangleq \left(\balpha, \boldsymbol{\tau}, \btheta, {\boldsymbol{\phi}^{\rm dl}}\right),
\end{align}
 capture the uplink and downlink channel parameters, respectively. In order to estimate $\bh_k^{\rm dl}$, for $k=1,2,\cdots,K$, we estimate frequency-independent channel parameters, namely $\balpha$, $\boldsymbol{\tau}$, and $\btheta$ using uplink training as these parameters are common in both uplink and downlink channels, while ${\boldsymbol{\phi}}^{\rm dl}$ is estimated during downlink training with a much less training overhead.
 The details are presented in the next two subsections.
\subsection{Uplink Training} \label{uplink_training}
 Here, we aim to estimate $\bx^{\rm up}$ using the uplink training. To do so, we stack the observations across all subcarriers as $\by^{\rm up}\triangleq \Big[ \by^{\rm up}_k\Big]_{k \in \mathcal{K}_{\rm up}}$, and defining $\mathcal{A}(\bx^{\rm up})\triangleq \Big[ \bh_k^{\rm up} s_k\Big]_{k \in \mathcal{K}_{\rm up}}$ as well as $\bn^{\rm up} \triangleq \Big[ \bn^{\rm up}_k\Big]_{k \in \mathcal{K}_{\rm up}}$, the collected received signal in the uplink is expressed as
\begin{align} \label{UL_received signal_non_linear}
\by^{\rm up} = \mathcal{A}(\bx^{\rm up}) +\bn^{\rm up},
\end{align}
where $\mathcal{A}(\bx^{\rm up})$ is a non-linear function of $\bx^{\rm up}$. Using the LS estimation approach, we obtain the estimate of $\bx^{\rm up}$ by solving the following minimization problem:
\begin{align} \label{LS_opt_x_up}
\min_{\balpha, \boldsymbol{\tau}, \btheta, {\boldsymbol{\phi}^{\rm up}}}&J_{\rm up}\left(\balpha, \boldsymbol{\tau}, \btheta, {\boldsymbol{\phi}^{\rm up}}\right)\nonumber\\
{\rm s.t.}&\,\,\,\, \balpha\succcurlyeq 0, \,\,\,\, \boldsymbol{\tau} \succcurlyeq 0\nonumber\\
& \,\,\,\, 0\preccurlyeq \btheta \preccurlyeq 2\pi, \,\,\,\, 0\preccurlyeq {\boldsymbol{\phi}^{\rm up}} \preccurlyeq 2\pi,
\end{align}
where $J_{\rm up}\left(\balpha, \boldsymbol{\tau}, \btheta, {\boldsymbol{\phi}^{\rm up}}\right) \triangleq \|\by^{\rm up}- \mathcal{A}\left(\balpha, \boldsymbol{\tau}, \btheta, {\boldsymbol{\phi}^{\rm up}}\right)\|_2^2$. The optimization problem \eqref{LS_opt_x_up} is difficult to solve analytically due to non-linear structure of the objective function. Even solving \eqref{LS_opt_x_up} numerically (using for example coordinate descent or gradient projection \cite{Kelley_book}) is challenging mainly because of the high-dimensionality of the search space as well as the fact that the objective function is not convex. Therefore, any numerical approach will suffer from the convergence issues. To alleviate such difficulties, instead of solving \eqref{LS_opt_x_up} directly, we use an approach which is based on DGMs. Variational auto-encoders (VAEs) and GANs are well-known examples of DGMs. In this technique, we replace the tuple $\left(\balpha, \boldsymbol{\tau}, \btheta\right)$ by $G(\bz)$, which is a mapping from $\bz$ (a latent variable) to the domain of $\left(\balpha, \boldsymbol{\tau}, \btheta\right)$. \textit{In doing so, we implicitly assume that the tuple $\left(\balpha, \boldsymbol{\tau}, \btheta\right)$ is described in terms of a low-dimensional $\bz$ through $G(\cdot)$ - a vector valued function of a vector valued variable.
}That is, the tuple $\left(\balpha, \boldsymbol{\tau}, \btheta\right)$ has a structure dictated by $G(\bz)$. \label{z_random}Here, we assume that the entries of vector $\bz$ are random variables from a Gaussian distribution (i.e., $\bz \sim \mathcal{N}(\b0, \bI_d)$, where $d$ is the length of vector $\bz$). We will elaborate later on how to find $G(\bz)$ using a GAN architecture.

Given $G(\bz)$, we rewrite our LS\ problem as\begin{align} \label{LS_opt_x_up_expressed_in_z}
\min_{\bz , {\boldsymbol{\phi}^{\rm up}}}&J_{\rm up}\left(G(\bz), {\boldsymbol{\phi}^{\rm up}}\right)\nonumber\\
{\rm s.t.}&\,\,\,\, 0\preccurlyeq {\boldsymbol{\phi}^{\rm up}} \preccurlyeq 2\pi.
\end{align}
Note that \eqref{LS_opt_x_up} and \eqref{LS_opt_x_up_expressed_in_z} are not equivalent, since solving \eqref{LS_opt_x_up_expressed_in_z} results in tuples $\left(\balpha, \boldsymbol{\tau}, \btheta\right)$ which belong to the range of $G(\bz)$, whereas in \eqref{LS_opt_x_up} there is no such restriction. However, given the above assumption, theoretically from \eqref{LS_opt_x_up} to  \eqref{LS_opt_x_up_expressed_in_z}, there is no optimality loss, as the optimal $\left(\balpha, \boldsymbol{\tau}, \btheta\right)$ in \eqref{LS_opt_x_up} is on the manifold that $G(\bz)$ generates samples on. Therefore, the optimal solution to \eqref{LS_opt_x_up} is in the range of $G(\bz)$. Note that the objective function in \eqref{LS_opt_x_up_expressed_in_z} is a periodic function with respect to ${\boldsymbol{\phi}^{\rm up}}$, i.e., $J_{\rm up}\left(G(\bz), {\boldsymbol{\phi}^{\rm up}}\right)= J_{\rm up}\left(G(\bz), {\boldsymbol{\phi}^{\rm up} + 2i\pi}\right)$ for integer $i$. This implies that if ${\boldsymbol{\phi}^{\rm up}}^*$ is an optimal solution  out of the interval $\left[0,2\pi\right]$, without loss of optimality, we can project it back to $\left[0,2\pi\right]$ by removing multiple integers of $2\pi$. Therefore, we can ignore the constraint in \eqref{LS_opt_x_up_expressed_in_z} and solve the following unconstrained optimization problem:
\begin{align} \label{LS_opt_x_up_expressed_in_z_unconst}
\min_{\bz , {\boldsymbol{\phi}^{\rm up}}} J_{\rm up}\left(G(\bz), {\boldsymbol{\phi}^{\rm up}}\right).
\end{align}

To efficiently solve \eqref{LS_opt_x_up_expressed_in_z_unconst}, we jointly and iteratively update $\bz$ and ${\boldsymbol{\phi}^{\rm up}}$ using the gradient descent algorithm.
To do so, we need to know the mapping $G(\cdot)$. Later we show how this mapping can be determined
using the GAN architecture.

\textbf{Remark 1}: \label{remark} Despite the non-linearity and possibly the non-convex structure of the optimization problem in \eqref{LS_opt_x_up}, solving \eqref{LS_opt_x_up_expressed_in_z_unconst} is easier. This can be explained as follows: 1) The tuple $\left(\balpha, \boldsymbol{\tau}, \btheta\right)$ is described in terms of a low-dimensional $\bz$ via a vector valued function $G(\cdot)$, and therefore, the search space in \eqref{LS_opt_x_up_expressed_in_z_unconst} is in $\bz$ and $\boldsymbol{\phi}^{\rm up}$ which has a much smaller dimension than the optimization variables in \eqref{LS_opt_x_up}. 2) We will show in Lemma \ref{J_lipschitz} in Section \ref{Sec:conv_analysis} that the objective function of \eqref{LS_opt_x_up_expressed_in_z_unconst} is a Lipschitz-continuous function of its variables. Therefore, the convergence of steepest descent algorithm in solving  \eqref{LS_opt_x_up_expressed_in_z_unconst} is guaranteed. This fact does not necessarily hold true for the objective function in \eqref{LS_opt_x_up} due to a larger search space as well as the box constraint over each dimension. 3) Unlike \eqref{LS_opt_x_up}, in \eqref{LS_opt_x_up_expressed_in_z_unconst}, we incorporate the underlying distribution of channel parameters to act as a prior. This not only eliminates the box constraints in \eqref{LS_opt_x_up}, but also provides a better estimation of channel parameters even with highly noisy observations. In our simulations, and for the sake of comparison, we solve \eqref{LS_opt_x_up} based on the R2F2 algorithm \cite{R2F2}. We will show how incorporating $G(\bz)$ improves the performance.

\textbf{Remark 2}: \label{remark2} The distribution of channel parameters $(\boldsymbol{\alpha}, \boldsymbol{\tau}, \boldsymbol{\theta})$ depends on physical properties of the propagation environment (such as geometry and material of surrounding objects), operating frequencies, user densities as well as data traffic. Importantly, they are independent of the number of antennas at the BS as well as the number of subcarriers \cite{6515173, 6834753, 8798669, 9034179}. Here, \emph{we use $G(\bz)$ to model the statistical distribution of $(\boldsymbol{\alpha}, \boldsymbol{\tau}, \boldsymbol{\theta})$}. Note that each BS uses only one $G(\bz)$. For a BS in a different propagation environment, $G(\bz)$ needs to be retrained.

\subsection{Downlink Training}
Relying on the partial reciprocity of channel parameters in FDD, we use the frequency-independent channel parameters (i.e., $\left(\balpha, \boldsymbol{\tau}, \btheta\right)$ that is  estimated during the uplink training), and estimate ${\boldsymbol{\phi}^{\rm dl}}$ using a fewer training symbols in downlink. To estimate ${\boldsymbol{\phi}^{\rm dl}}$, we use the LS technique, similar to what presented in Subsection \ref{uplink_training}. Given $\left(\balpha, \boldsymbol{\tau}, \btheta\right)$, we define $\bg({\boldsymbol{\phi}^{\rm dl}}) \triangleq \left[e^{j {\phi}_1^{\rm dl}} \;\; e^{j {\phi}_2^{\rm dl}} \;\; \cdots\;\; e^{j {\phi}_L^{\rm dl}} \right]^T$. We also define $\gamma_l^{k}\triangleq \alpha_l e^{j\frac{2\pi k}{K}\tau_{l}B} $ and the following $p\times L$ matrix
\begin{align} \label{B_k_uplink_def}
\bB_k\left(\balpha, \boldsymbol{\tau}, \btheta\right) \triangleq \bS^{\rm dl}_k\left[ \gamma_1^{k}\bb(\theta_1,  \lambda)\;\;  \gamma_2^{k}\bb(\theta_2,  \lambda)\;\; \cdots \;\; \gamma_L^{k}\bb(\theta_L, \lambda) \right].
\end{align}
Then, we can express \eqref{DL_received signal}  in terms of ${\boldsymbol{\phi}^{\rm dl}}$ as
\begin{align} \label{DL_received signal_expressed_gdl}
\by^{\rm dl}_k\triangleq \bB_k\left(\balpha, \boldsymbol{\tau}, \btheta\right)\bg({\boldsymbol{\phi}^{\rm dl}}) +\bn^{\rm dl}_k, \,\,\, k \in \mathcal{K}_{\rm dl}.
\end{align}
Now, stacking the observations across all subcarriers in the set $\mathcal{K}_{\rm dl}$, and defining $\by^{\rm dl}\triangleq \Big[ \by^{\rm dl}_k\Big]_{k \in \mathcal{K}_{\rm dl}}$, $\mathcal{B}\triangleq \Big[ \bB_k\left(\balpha, \boldsymbol{\tau}, \btheta\right) \Big]_{k \in \mathcal{K}_{\rm dl}}$ , and $\bn^{\rm dl} \triangleq \Big[ \bn^{\rm dl}_k\Big]_{k \in \mathcal{K}_{\rm dl}}$, the collected received signal in the downlink is given by
\begin{align} \label{DL_received signal_non_linear}
\by^{\rm dl} = \mathcal{B}\bg({\boldsymbol{\phi}^{\rm dl}}) +\bn^{\rm dl}.
\end{align}
Defining $J_{\rm dl}\left(\boldsymbol{\phi}^{\rm dl}\right)\triangleq \|\by^{\rm dl} - \mathcal{B}\bg({\boldsymbol{\phi}^{\rm dl}})\|_2^2$, we now solve
the following LS problem\begin{align} \label{LS_opt_dl}
\min_{\boldsymbol{\phi}^{\rm dl}}& \,\,\,\,  J_{\rm dl}\left(\boldsymbol{\phi}^{\rm dl}\right),
\end{align}
to find the optimal ${\boldsymbol{\phi}^{\rm dl}}$. Note that in the optimization problem \eqref{LS_opt_dl}, as in \eqref{LS_opt_x_up_expressed_in_z_unconst}, we have ignored the constraint $0\preccurlyeq\boldsymbol{\phi}^{\rm dl}\preccurlyeq 2\pi$ due to the fact that the objective function is a periodic function of $\boldsymbol{\phi}^{\rm dl}$, i.e., $J_{\rm dl}\left(\boldsymbol{\phi}^{\rm dl}\right) = J_{\rm dl}\left(\boldsymbol{\phi}^{\rm dl}+2i\pi\right)$, for integer $i$. As in \eqref{LS_opt_x_up_expressed_in_z_unconst}, the optimization problem \eqref{LS_opt_dl} is an unconstrained least squares problem and can be solved using the gradient descent algorithm. Once $\boldsymbol{\phi}^{\rm dl}$ is estimated it is fed back to BS for the subsequent data transmission\footnote{\label{quantize_feedback}In this paper, we do not consider the quantization of $\boldsymbol{\phi}^{\rm dl}$ obtained in \eqref{LS_opt_dl}.}.

\section{Deep Generative Model}\label{Sec:DGMs}
In the previous section, we described an LS-based procedure to obtain the parameters that defines the downlink channel. the crucial step is the mapping $G(\bz)$ that leads to $\bx^{\rm up}$. In this section, we explain how to find $G(\bz)$ using DGMs. Before we get into $G(\bz)$, and for the sake of the paper being self-contained, we briefly introduce DGMs and explain how we can use them to find $G(\bz)$.

The core idea of DGMs is to represent a high-dimensional and complex distribution of data $\bx$ (in this case $\bx = \left(\balpha, \boldsymbol{\tau}, \btheta\right)$) using a deterministic mapping over a low-dimensional random vector $\bz$ which has a well behaved pdf (e.g., uniform or Gaussian). Specifically, a DGM is a function $G(\bz)$ that maps a low-dimensional random  vector $\bz \in \mathbb{R}^d$, typically drawn independently from Gaussian or uniform distribution, to a high-dimensional vector $\bx_g \triangleq G(\bz) \in \mathbb{R}^{n}$ where $n\geq d$ (in practice, due to the structure of $\bx$ we can have $n \gg d$). The mapping $G(\cdot)$ is determined such that the distribution of $\bx_g$,  generated by $G(\cdot)$, matches the distribution of the real world data vector $\bx$. In other words, using  the transformation $G(\cdot)$ from a simple and low-dimensional distribution, we generate samples that belong to the same manifold as $\bx$ does. This implies that any generated sample from $G(\cdot)$ already satisfies the constraints in \eqref{LS_opt_x_up}\footnote{ \label{foot_const}To show this, let us define $\mathcal{C} \triangleq \{ \bx = (\balpha, \boldsymbol{\tau}, \btheta)| \balpha\succcurlyeq 0, \boldsymbol{\tau} \succcurlyeq 0,  0\preccurlyeq \btheta \preccurlyeq 2\pi,\}$ the feasible set for the objective function~\eqref{LS_opt_x_up}. Let $\mathcal{R}_x$ denote the domain of the true sample $\bx$, and $\mathcal{R}_g$ denote the range of $G(\bz)$. We know that $\mathcal{R}_x \subset \mathcal{C}$. Theoretically, $G(\bz)$ is trained such that it generates realistic samples $\bx_g$ such that $D(\bx)$ cannot distinguish $\bx_g$ among the true samples $\bx$. This implies that any $\bx_g \in \mathcal{R}_g$  belongs to the same manifold as $\bx$ does. This further implies that $\bx_g \in \mathcal{R}_x$. Therefore, we conclude that $\bx_g \in \mathcal{C}$, i.e., any sample generated from $G(\bz)$ satisfies the constraints in~\eqref{LS_opt_x_up}.}.
The function $G(\cdot)$ is parameterized by a deep neural network, which is trained in an unsupervised way as explained below. One well-known example of  DGM is the family of GANs. In this paper, we use the GAN architecture to find $G(\cdot)$. Next, we briefly introduce GANs and explain how we use GANs to obtain $G(\cdot)$.

\subsection{GANs}
GANs are among the most powerful DGMs that are  used to capture the distribution of data ~\cite{NIPS2014_5ca3e9b1}. As shown in Fig. \ref{Fig_GAN}, a GAN consists of two fully-connected feed-forward neural networks, namely a generator network $G_{{\cal W}_g}(\bz): \mathbb{R}^d \rightarrow \mathbb{R}^n$, and a discriminator network $D_{{\cal W}_d} (\bx):\mathbb{R}^n \rightarrow [0,1]$, where ${\cal W}_g$ and ${\cal W}_d$ represent, respectively, the sets of weights of the generator and discriminator networks. The generator network $G_{{\cal W}_g}(\bz)$ maps the input random vector  $\bz \backsim \mathbb{P}_z(\bz)$, into the data space $\bx_g \backsim \mathbb{P}_g(\bx)$. Here, $\mathbb{P}_z(\bz)$ represents the pdf of the  input random vector $\bz$ and  is usually chosen to be $\mathbb{P}_z(\bz)= \mathcal{N}(\b0, \bI_d)$, and $\mathbb{P}_g(\bx)$ represents the pdf of the generated samples. The discriminator network, $D_{{\cal W}_d}(\cdot)$, receives the two sets of inputs: one set consists of  the samples  $ \bx_g$ generated by $G_{{\cal W}_g}(\bz)$ and the other set consists of  the true samples $ \bx$. The discriminator network $D_{{\cal W}_d}(\cdot)$ is meant to correctly distinguish between the fake samples $\bx_g$ and the true samples $\bx$. Effectively, the goal of the generator network $G_{{\cal W}_g}(\cdot)$ is to generate fake samples such that the discriminator network $D_{{\cal W}_d}(\bx)$ cannot distinguish them from the true samples. Meanwhile, the goal of the discriminator network $D_{{\cal W}_d}(\cdot)$ is to correctly distinguish between $\bx_g$ and $\bx$. To do so, $D_{{\cal W}_d}(\bx)$ provides the probability that $\bx$ to be a real sample. If $\mathbb{P}_r(\bx)$ represents the distribution of true samples, the goal is to have $\mathbb{P}_g(\bx)= \mathbb{P}_r(\bx)$, by optimally adjusting ${\cal W}_g$ and ${\cal W}_d$. Below, we explain how to adjust the weights.

\subsubsection{Training of GANs}\label{GAN_train}

$G_{{\cal W}_g}(\cdot)$ and $D_{{\cal W}_d}(\cdot)$ are trained simultaneously and iteratively via the following two-player min-max game \cite{NIPS2014_5ca3e9b1}:
\begin{align} \label{min_max_gan}
\min_{{\cal W}_g} \max_{{\cal W}_d}\,\,\, \tilde L(G_{{\cal W}_g}(\cdot),D_{{\cal W}_d}(\cdot)),
\end{align}
where $\tilde L(G_{{\cal W}_g}(\cdot),D_{{\cal W}_d}(\cdot)),$ is the loss function and defined as
\begin{align} \label{LGD_def}
\tilde L(G_{{\cal W}_g}(\cdot),D_{{\cal W}_d}(\cdot))&\triangleq \mathbb{E}_{\bx}\left[\log D_{{\cal W}_d}(\bx)\right] \nonumber\\
&+ \mathbb{E}_{\bz }\left[\log\left( 1- D_{{\cal W}_d}(G_{{\cal W}_g}(\bz))\right)\right].
\end{align}

\begin{figure}[!t]
    \centering
\resizebox{!}{6cm}{\includegraphics{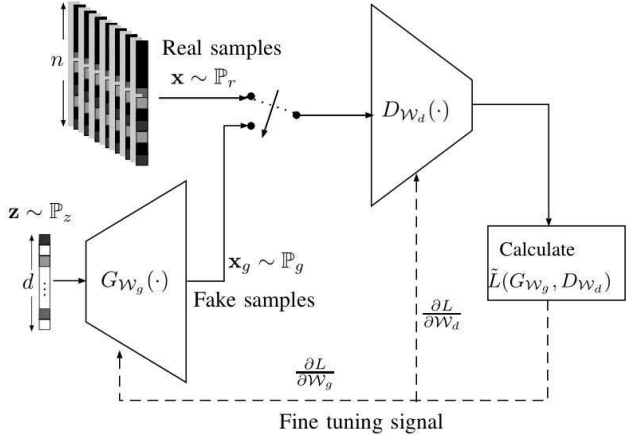}}
\caption{The GAN structure. }
\label{Fig_GAN}
\end{figure}

Note that the objective of training  $G_{{\cal W}_g}(\cdot)$ is to fool $D_{{\cal W}_d} (\cdot)$ by generating the \emph{realistic} samples $\bx_g$ such that $D_{{\cal W}_d} (\cdot)$ assigns a high probability to $\bx_g$ being true samples. This can be done by maximizing  $D_{{\cal W}_d} (\bx_g)$ (or equivalently, minimizing $\log\left( 1- D_{{\cal W}_d}(G_{{\cal W}_g}(\bz))\right)$ over ${\cal W}_g$, as given in the second term in \eqref{LGD_def}). In the meantime, $D_{{\cal W}_d} (\cdot)$ aims to correctly distinguish $\bx_g$, generated by $G_{{\cal W}_g}(\cdot)$, from the true samples $\bx$. In other words, $D_{{\cal W}_d} (\bx_g)$ is meant to be close to zero while $D_{{\cal W}_d} (\bx)$ has to be close to one. Therefore, ${\cal W}_d$ is chosen such that $\log D_{{\cal W}_d}(\bx) + \log\left( 1- D_{{\cal W}_d}(G_{{\cal W}_g}(\bz))\right)$ is maximized. It has been shown that such a competitive interplay between  $G_{{\cal W}_g}(\bz)$ and $D_{{\cal W}_d} (\bx)$ converges to an equilibrium where $P_g(\bx)= P_r (\bx)$  \cite{NIPS2014_5ca3e9b1} . At this point, the generator produces realistic $\bx_g$ such that the discriminator is unable to differentiate between $\bx$ and $\bx_g$, i.e.,  $D_{{\cal W}_d} (\bx)= D_{{\cal W}_d} (\bx_g)=\frac{1}{2}$.

\subsubsection{Mode collapse issue and its solution}\label{mode_collapes}

Despite the success of GANs in learning the underlying distribution of data, training of GANs is challenging due to the instability in the training and sensitivity to the hyper-parameters. In \cite{10.5555/3157096.3157346}, the authors propose different solutions to improve the training stability and robustness. On the other hand, \emph{mode collapse,} an issue that can hinder the training of GANs, refers to collapsing of large volumes of probability mass into a few modes. This means that although the generator produces meaningful samples, these samples belong to \emph{only} few modes of the data distribution; therefore, the samples produced by the generator do not fully represent the underlying distribution of real data.

Different solutions has been proposed to address the mode collapse issue; in this paper, we use a regularized GAN (Reg-GAN) proposed in \cite{DBLP:conf/iclr/CheLJBL17}. Compared to the original form of GANs, a Reg-GAN uses a regularizer term that meant to penalize the missing modes. To do so, together with the generator and discriminator, an encoder network $E_{{\cal W}_e}(\bx): \bx \rightarrow \bz$ is trained to help the generator to avoid the missing modes, where ${\cal W}_e$ is the set of weights of the encoder network.

To perform this training, two regularizing terms are considered in training of the $E_{{\cal W}_e}(\cdot)$ and $G_{{\cal W}_g}(\cdot)$ networks. One regularizer is based on the fact that if $G_{{\cal W}_g}\left(E_{{\cal W}_e}(\cdot)\right)$ is a good \textit{auto-encoder}, then, for any $\bx_{\rm o} \in {\cal M}_{\rm o}$, where ${\cal M}_{\rm o}$ is the set of missing modes, we obtain $\bx_{\rm o} \approx G_{{\cal W}_g}\left(E_{{\cal W}_e}(\bx_o)\right)$. Therefore, $\mathbb{E}_{\bx}\left[\|\bx - G_{{\cal W}_g}\left(E_{{\cal W}_e}(\bx)\right)\|^2\right]$ is the loss function considered in the training of  $E_{{\cal W}_e}(\cdot)$ and $G_{{\cal W}_g}(\cdot)$ networks, as a regularizer to penalize the generator for any missing samples including the samples of minor modes. A second regularizer, $\mathbb{E}_{\bx}\left[\log D_{{\cal W}_d}\left( G_{{\cal W}_g}\left( E_{{\cal W}_e}(\bx) \right )\right)\right]$ is used to encourage $G_{{\cal W}_g}(E_{{\cal W}_e}(\cdot))$ to generate realistic samples such that $ D_{{\cal W}_d}(\cdot)$ assigns, to these samples, a high probability of being a true sample. Therefore, we can achieve a fair probability distribution across different modes. The regularized loss functions for the generator, the encoder and discriminator are respectively given in \eqref{T_G}, \eqref{T_E}, and \eqref{T_D}, where $\lambda_1$ and $\lambda_2$ are the regularizer's coefficients.

 \begin{figure*}[t]
\begin{align}
T_G\left({\cal W}_g, {\cal W}_d, {\cal W}_e\right)&= - \mathbb{E}_{\bz }\left[\log D_{{\cal W}_d}\left( G_{{\cal W}_g}(\bz)\right) \right]+ \mathbb{E}_{\bx }\left[\lambda_1\|\bx - G_{{\cal W}_g}\left(E_{{\cal W}_e}(\bx)\right)\|^2+\lambda_2\log D_{{\cal W}_d}\left( G_{{\cal W}_g}\left(E_{{\cal W}_e}(\bx)\right)\right)\right] \label{T_G} \\
T_E\left({\cal W}_g, {\cal W}_d, {\cal W}_e\right)&= \mathbb{E}_{\bx }\left[\lambda_1\|\bx - G_{{\cal W}_g}\left(E_{{\cal W}_e}(\bx)\right)\|^2+\lambda_2\log D_{{\cal W}_d}\left( G_{{\cal W}_g}\left(E_{{\cal W}_e}(\bx)\right)\right)\right]\label{T_E}\\
T_D\left({\cal W}_g, {\cal W}_d, {\cal W}_e\right)&= \mathbb{E}_{\bx }\left[\log D_{{\cal W}_d}(\bx)\right] + \mathbb{E}_{\bz }\left[\log\left( 1- D_{{\cal W}_d}\left(G_{{\cal W}_g}(\bz)\right)\right)\right].\label{T_D}
\end{align}
    \end{figure*}

\emph{Training procedure}: During the training, we aim to find ${\cal W}_d$, ${\cal W}_g$ and ${\cal W}_e$. This can be done by jointly and iteratively maximizing $T_D$, and minimizing $T_G$ and $T_E$. Before we start the training process, ${\cal W}_d$, ${\cal W}_g$ and ${\cal W}_e$ are randomly initialized. In each of the subsequent iterations, we sample a mini-batch of size $m$ from both the training set $\left(\left\{\bx^{(i)}\right\}_{i=1}^{m}\right)$ and noise samples $\left(\left\{\bz^{(i)}\right\}_{i=1}^{m}\right)$. For fixed ${\cal W}_g$ and ${\cal W}_e$, we first update ${\cal W}_d$ by ascending in the direction of the gradient of $T_D$. Then, while fixing ${\cal W}_d$ and ${\cal W}_e$, we update ${\cal W}_g$ by descending in the opposite  direction of the gradient of $T_G$. Similarly, ${\cal W}_e$ is updated by descending in the opposite direction of the gradient of $T_E$ for fixed ${\cal W}_g$ and ${\cal W}_d$.

The Reg-GAN training procedure is summarized in Algorithm \ref{Algo_MRGAN}.
\begin{algorithm*}
    \caption{Reg-GAN Training Algorithm }
    \textbf{Inputs}: Data set, number of epochs, mini-batch size $m$, $\lambda_1$, and $\lambda_2$.\\
    \textbf{Outputs}: Trained $G(\bz)$ .
    \begin{algorithmic}[1]
        \State \textbf{for each} epoch \textbf{do}:
        \State \hskip1.0em  Sample mini-batches of $m$ noise samples $\{\bz^{(1)}, \bz^{(2)}, \cdots, \bz^{(m)}\}$ and $m$ data samples $\{\bx^{(1)}, \bx^{(2)}, \cdots, \bx^{(m)}\}$ from $P_z(\bz)$ and $P_r(\bx)$, respectively.

        \State \hskip1.0em Update discriminator $D_{{\cal W}_d}(\cdot)$ by ascending in the direction of its stochastic gradient:
\begin{align}
\nabla_{{{\cal W}_d}}\frac{1}{m}\sum_{i=1}^{m} \Big[\log D_{{\cal W}_d}(\bx^{(i)}) + \log\left( 1- D_{{\cal W}_d}\left(G_{{\cal W}_g}(\bz^{(i)})\right)\right)\Big]\nonumber
.
\end{align}

\State \hskip1.0em  Sample mini-batches of $m$ noise samples $\{\bz^{(1)}, \bz^{(2)}, \cdots, \bz^{(m)}\}$ and $m$ data samples $\{\bx^{(1)}, \bx^{(2)}, \cdots, \bx^{(m)}\}$  from $P_z(\bz)$ and $P_r(\bx)$, respectively.

        \State \hskip1.0em Update generator $G_{{\cal W}_g}(\cdot)$ by descending in the opposite direction of its stochastic gradient:
\begin{align}
\nabla_{{{\cal W}_g}}\frac{1}{m}\sum_{i=1}^{m} \Big[-\log D_{{\cal W}_d}\left( G_{{\cal W}_g}(\bz^{(i)})\right)
+\lambda_1\|\bx^{(i)} - G_{{\cal W}_g}\left(E_{{\cal W}_e}(\bx^{(i)})\right)\|^2+\lambda_2\log D_{{\cal W}_d}\left( G_{{\cal W}_g}\left(E_{{\cal W}_e}(\bx^{(i)})\right)\right)\Big]\nonumber
\end{align}

\State \hskip1.0em  Sample mini-batches of $m$ noise samples $\{\bz^{(1)}, \bz^{(2)}, \cdots, \bz^{(m)}\}$ and $m$ data samples $\{\bx^{(1)}, \bx^{(2)}, \cdots, \bx^{(m)}\}$  from $P_z(\bz)$ and $P_r(\bx)$, respectively.

        \State \hskip1.0em Update encoder $E_{{\cal W}_e}(\cdot)$ by descending in the opposite direction of its stochastic gradient:
\begin{align}
\nabla_{{{\cal W}_e}}\frac{1}{m}\sum_{i=1}^{m} \Big[\lambda_1\|\bx^{(i)} - G_{{\cal W}_g}\left(E_{{\cal W}_e}(\bx^{(i)})\right)\|^2+\lambda_2\log D_{{\cal W}_d}\left( G_{{\cal W}_g}\left(E_{{\cal W}_e}(\bx^{(i)})\right)\right)\Big]\nonumber
\end{align}

        \State \textbf{end}
    \end{algorithmic}
    \label{Algo_MRGAN}
\end{algorithm*}
Note that, in Algorithm \ref{Algo_MRGAN}, the gradient-based update can be implemented using any standard gradient-based algorithm. In this paper, to speed up the convergence, we use a momentum-based gradient update.

\section{Convergence, complexity, and identifiability} \label{complexity_and_ident}
\subsection{Convergence} \label{Sec:conv_analysis}
In this subsection, we analytically study the convergence of our estimation algorithm provided in Subsection \ref{uplink_training}. Here, we aim to show that the objective function decreases from one iteration to the next. The following lemma provides an insight into the convergence.
\begin{lemma} \label{J_lipschitz}
The gradient of $J_{\rm up}\left(G_{{\cal W}_g}(\bz), {\boldsymbol{\phi}^{\rm up}}\right)$ is a Lipschitz-continuous function.
\end{lemma}
\emph{Proof}: See Appendix \ref{lip}.

Let us define ${\tilde{\bz}}\triangleq \left(\bz, \boldsymbol{\phi}\right)$ and $\ell(\tilde{\bz})\triangleq J_{\rm up}\left(G_{{\cal W}_g}(\bz), {\boldsymbol{\phi}}\right)$. According to  Lemma \ref{J_lipschitz}, we can write $\| \nabla \ell(\tilde{\bz}^{(i)}) - \nabla \ell(\tilde{\bz}^{(j)}) \| \leq  C_\ell \| \tilde{\bz}^{(i)} - \tilde{\bz}^{(j)} \| $, where $C_\ell$ is the Lipschitz constant. This now means that $\nabla^2\ell(\tilde{\bz})\preccurlyeq C_\ell\bI$, which further implies that
\begin{align} \label{f_quad}
\bv^T\nabla^2\ell(\tilde{\bz})\bv\leq C_\ell\|\bv\|^2
\end{align}
holds true for any vector $\bv$. On the other hand, if we consider the gradient descent step $\eta = 1/C_\ell$, then the update from the $i$th to the $(i+1)$th iteration is written as
\begin{align} \label{update_z_stepsize}
\tilde{\bz}^{(i+1)}= \tilde{\bz}^{(i)} - \frac{1}{C_\ell} \bigtriangledown \ell(\tilde{\bz}^{(i)}).
\end{align}
Let us use the Taylor expansion and expand $\ell(\tilde \bz^{(i+1)})$ as
\begin{align} \label{Taylor expansion}
&\ell(\tilde{\bz}^{(i+1)}) \nonumber\\
 &=\ell(\tilde{\bz}^{(i)}) + \bigtriangledown \ell(\tilde{\bz}^{(i)}) \left(\tilde{\bz}^{(i+1)} - \tilde{\bz}^{(i)}\right)\nonumber\\
&+ \frac{1}{2} \left(\tilde{\bz}^{(i+1)} - \tilde{\bz}^{(i)}\right)^T \bigtriangledown^2 \ell(\tilde{\bz}^{(i)})\left(\tilde{\bz}^{(i+1)} - \tilde{\bz}^{(i)}\right)\nonumber\\
&\overset{(a)}{\leq} \ell(\tilde{\bz}^{(i)}) - \frac{1}{C_\ell} \bigtriangledown \ell(\tilde{\bz}^{(i)})^T \bigtriangledown \ell(\tilde{\bz}^{(i)}) + \frac{C_\ell}{2}\|\tilde{\bz}^{(i+1)} - \tilde{\bz}^{(i)}\|^2
\nonumber\\
&= \ell(\tilde{\bz}^{(i)}) - \frac{1}{C_\ell} \|\bigtriangledown \ell(\tilde{\bz}^{(i)})\|^2 + \frac{1}{2C_\ell} \|\bigtriangledown^2 \ell(\tilde{\bz}^{(i)})\|^2
\nonumber\\
&= \ell(\tilde{\bz}^{(i)}) - \frac{1}{2C_\ell} \|\bigtriangledown^2 \ell(\tilde{\bz}^{(i)})\|^2,
\end{align}
where $(a)$ follows directly from \eqref{f_quad} and \eqref{update_z_stepsize}. From the last equality in \eqref{Taylor expansion}, we observe that in the $i$th step $\ell(\tilde{\bz}^{(i+1)}) \leq \ell(\tilde{\bz}^{(i)}), $ or equivalently,
\begin{align} \label{J_reduce_step1}
J_{\rm up}\left(G_{{\cal W}_g}(\bz^{(i+1)}), {\boldsymbol{\phi}^{{\rm up},(i+1)}}\right)\leq J_{\rm up}\left(G_{{\cal W}_g}(\bz^{(i)}), {\boldsymbol{\phi}^{{\rm up},(i)}}\right).
\end{align}
The inequality in \eqref{J_reduce_step1} holds true for any $\eta \leq 1/C_\ell$. The inequality of \eqref{J_reduce_step1} shows that the LS cost function decreases from one iteration to the next. Since the cost function has a lower bound (of zero), the algorithm converges.

\subsection{Computational complexity} \label{comp_complex}
To  solve~\eqref{LS_opt_x_up_expressed_in_z} (or equivalently~\eqref{LS_opt_x_up_expressed_in_z_unconst}), we use a gradient descent algorithm to jointly and iteratively update $\bz$ and ${\boldsymbol{\phi}^{\rm up}}$. The gradient descent requires $\mathcal{O}\left(\log(1/\varepsilon)\right)$ iterations, where $\varepsilon$ is the chosen error tolerance. Each iteration of gradient descent has the following two steps: we first fix ${\boldsymbol{\phi}^{\rm up}}$ and update $\bz$, then for a fixed $\bz$, we update ${\boldsymbol{\phi}^{\rm up}}$. In the first step, for a fixed ${\boldsymbol{\phi}^{\rm up}}$, the computational complexity of updating $\bz$ is dominated by the complexity of calculating $\frac{\partial}{\partial \bz}J_{\rm up} = \frac{\partial J_{\rm up}}{\partial \bx} \frac{\partial \bx}{\partial \bz} $. The complexity of the first term is given by $\mathcal{O}(M|\mathcal{K}_{\rm up}|L^3)$. The second term, however, is related to the gradient of the $G(\bz)$ with respect to $\bz$, which is characterized by a neural network. The computational complexity of this term is given by $\mathcal{O}\left(dnq\right)$, where $n$ is the number of network layers and $q$ is the maximum number of nodes in each layer. In the second step, the computational complexity of updating ${\boldsymbol{\phi}^{\rm up}}$ is given by $\mathcal{O}(M|\mathcal{K}_{\rm up}|L^3)$. Therefore, the overall complexity of solving ~\eqref{LS_opt_x_up_expressed_in_z} is $\mathcal{O}(dnqM|\mathcal{K}_{\rm up}|L^3\log(1/\varepsilon))$. Similarly, we use a gradient descent algorithm to solve~\eqref{LS_opt_dl}, which incurs $\mathcal{O}\left(\log(1/\varepsilon)\right)$ iterations. The complexity of updating ${\boldsymbol{\phi}^{\rm dl}}$ in each iteration is given by $\mathcal{O}(p|\mathcal{K}_{\rm dl}|L^3)$. Therefore, the overall complexity of optimization problem~\eqref{LS_opt_dl} is given by $\mathcal{O}\left(p|\mathcal{K}_{\rm dl}|L^3\log(1/\varepsilon)\right)$.

\subsection{Identifiability}\label{ident}
During the downlink training, we assume that $p$ training symbols are transmitted over the $k$-th subcarrier ($k \in \mathcal{K}_{\rm dl}$) and the $m$th antenna ($m \in \mathcal{M}_{\rm dl}$). The value of $p$ depends on the $|\mathcal{M}_{\rm dl}|$. The reason is that the $p|\mathcal{K}_{\rm dl}|\times L$ matrix $\mathcal{B}$ in \eqref{DL_received signal_non_linear} is of rank $\min\left(|\mathcal{M}_{\rm dl}|, L\right)$. If $|\mathcal{M}_{\rm dl}|< L$, then ${\rm rank} (\mathcal{B}) = |\mathcal{M}_{\rm dl}|$, leading to an identifiability issue in recovering $\{\phi_l^{\rm dl}\}_{l=1}^L$ in \eqref{DL_received signal_non_linear}. If $|\mathcal{M}_{\rm dl}|> L$, the matrix $\mathcal{B}$ becomes full column-rank, and therefore estimating $\{\phi_l^{\rm dl}\}_{l=1}^L$ requires $p|\mathcal{K}_{\rm dl}|\geq L$, or equivalently, $p\geq L/|\mathcal{K}_{\rm dl}|$.

\section{Implementation and simulation results}\label{Sec:Sims}
In this section, we explain a detailed implementation of our estimation technique. Simulation results are also provided to evaluate the performance of the proposed technique.

\subsection{Experiment setup and dataset}
\subsubsection{Dataset Generation} Throughout the simulations, we consider an indoor massive MIMO scenario. An example of such scenario is the "l1" scenario provided by DeepMIMO dataset \cite{DeepMIMO1} which is generated by the 3D ray-tracing simulator Wireless InSite \cite{Remcom}. The "l1" scenario comprises a $10 \times 10 \times 2.5 $ meters room with $2$ tables inside the room. There are $M=64$ antennas mounted on the ceiling. The users are spread inside the room across the $x$-$y$ plane with each of them being $1$ meter above the floor. The communication between the BS antennas and each UE is in FDD mode and uses $K$ OFDM subcarriers \footnote{\label{freq_select_foornote}According to the DeepMIMO dataset~\cite{DeepMIMO1}, the maximum channel delay spread is $\tau_{\rm max} = 94.5$ ns. Given the bandwidth $B= 20$ $({\rm MHz})$, the sample time is $T_s \simeq \frac{1}{B} =  50$ $({\rm ns})$. Since $\tau_{\rm max} > T_s$, the channels in this experiment suffers from frequency selectivity.}. The uplink and downlink operating frequencies are respectively $2.4$ GHz and $2.5$ GHz. The DeepMIMO dataset parameters are given in Table \ref{tab:DeepMIMO_params}.

\begin{table}[t]
\caption {DeepMIMO dataset parameters} \label{tab:DeepMIMO_params}
\centering
\begin{tabular}{ c||c  }

 \hline
  Scenario & l1-2.4 GHz and 11-2.5 GHz \\
 \hline
 Number of BS & 1\\
 \hline
Active BS  & [32]  \\
 \hline
 Number of BS antenna in (x,y,z) & (64,1,1) \\
\hline
Antenna Spacing & 0.5$\lambda$ \\
\hline
Active users row [first, last] & [1,500] \\
\hline
Number of propagation path & 5 \\
\hline
$K$ & 16 \\
\hline
OFDM sampling factor & 1 \\
\hline
OFDM limit & 16 \\
\hline
 $B$ & 20 MHz \\
\hline
Noise Power & $-174$ dBm/Hz \\
\hline
 \end{tabular}
\end{table}
 Using the above ray-tracing scenario as well as the parameters given in Table \ref{tab:DeepMIMO_params}, DeepMIMO generates the dataset. The generated dataset comprising the channel vectors between each UE and the BS across all $K$ subcarriers at both uplink and downlink frequencies, i.e., $2.4$ GHz and $2.5$ GHz, and the underlying channel parameters for each propagation path, namely $\balpha$, $\boldsymbol{\tau}$, $\btheta$, ${\boldsymbol{\phi}^{\rm up}}$, and ${\boldsymbol{\phi}^{\rm dl}}$, and the UE $x$-$y$ coordinates.
\subsubsection{Preprocessing} \label{pre_porocess}To avoid any permutation ambiguity, the paths are sorted from the lowest delay to the highest delay. To prepare the data to train our GAN, we note that the ranges of $\balpha$, $\boldsymbol{\tau}$, $\btheta$, ${\boldsymbol{\phi}^{\rm up}}$, and ${\boldsymbol{\phi}^{\rm dl}}$ are all in different scales. We use the following transformations to have them in a reasonable range: vector $\balpha$ is transformed by taking $\log_{10}$ of each of its entries, i.e., $\tilde{\balpha} \triangleq \log_{10} \balpha$; $\boldsymbol{\tau}$ is scaled by $T_s\triangleq 1/B$ as $\tilde{\boldsymbol{\tau}} \triangleq \boldsymbol{\tau}/T_s$; $\btheta$, ${\boldsymbol{\phi}^{\rm up}}$, and ${\boldsymbol{\phi}^{\rm dl}}$ are all expressed in terms of radians. After scaling/transforming these feature vectors, we concatenate them to form the overall feature vectors as
\begin{align} \label{overall_feature_vec}
\tilde{\bx}^{\rm up} &\triangleq \left[ \tilde{\balpha} , \tilde{\boldsymbol{\tau}}, \btheta, {\boldsymbol{\phi}^{\rm up}}\right]\\
\tilde{\bx}^{\rm dl} &\triangleq \left[ \tilde{\balpha} , \tilde{\boldsymbol{\tau}}, \btheta, {\boldsymbol{\phi}^{\rm dl}}\right].
\end{align}
Before we start training, we divide  $\tilde{\bx}^{\rm up}$ and $\tilde{\bx}^{\rm dl}$ by their maximum value, thereby normalizing their entries into the range of $[-1,1]$. The maximum value here refers to the global maximum value across all entries of sample vectors. To form the training and testing dataset, after scaling and normalizing the data, we first shuffle and then we split the data such that $80\%$ is dedicated for training and $20\%$ for testing.

\subsection{Network architectures and training}
\subsubsection{Network architecture} We use the Reg-GAN structure given in Section \ref{mode_collapes}, implemented in Pytorch\footnote{\label{pytorch}Simulations are based on Pytorch v1.3.1 available at https://pytorch.org/docs/1.3.1/.}. The generator $G_{{\cal W}_g}(\cdot): \bz \rightarrow \bx$ takes an input sample $\bz \sim \mathcal{N}(\b0, \bI_d)$, of size $d=8$, and passes it through a fully-connected neural network which has $3$ layers with $[10, 12, 14]$ neurons in the respective layers. It then generates samples $\bx \in \mathbb{R}^{n}$, with $n=15$. As for the encoder network $E_{{\cal W}_e}(\cdot): \bx \rightarrow \bz$, the input and output size are $15$ and $8$, respectively. There are also $3$ hidden layers with $[14, 12, 10]$  neurons in the corresponding layers. Likewise, the discriminator is implemented using a fully-connected neural network with the input size $15$, output size $1$, and $4$ hidden layers with $[12, 8, 4, 2]$ neurons in the corresponding layers.

To improve the stability and convergence during the training, in all the networks, we use Leaky ReLU activation function with slope $0.2$. Additionally, the generator uses the hyperbolic tangent (tanh) activation function in the output layer. We employ $20\%$ dropout in each hidden layer to further improve the generalization capability of the generator.

\subsubsection{Training}
The outcome of the training process is the generator $G_{{\cal W}_g}(\cdot)$ that finds the optimal $\bz$, denoted by $\bz^*$, for each realization in the test set. The training process follows directly from  Algorithm \ref{Algo_MRGAN}. We randomly initialize ${\cal W}_d$, ${\cal W}_g$, and ${\cal W}_e$. In each epoch, we sample a mini-batch of size $m$ from both the training set $\left(\left\{\bx^{(i)}\right\}_{i=1}^{m}\right)$ and the set of noise samples $\left(\left\{\bz^{(i)}\right\}_{i=1}^{m}\right)$. We first update ${\cal W}_d$ by ascending in the direction of the gradient of $T_D$ given in \eqref{T_D} assuming ${\cal W}_g$ and ${\cal W}_e$ are fixed. Then, we update ${\cal W}_g$ by descending in the opposite direction of the  gradient of $T_G$ given in \eqref{T_G} assuming ${\cal W}_d$ and ${\cal W}_e$ are fixed. Similarly, ${\cal W}_e$ is updated by descending in the opposite direction of  the gradient of $T_E$ given in \eqref{T_E}. Note that, in these steps, the sample mean is used instead of mathematical expectation. Furthermore, thanks to \verb"autograd", the gradient descent/ascend in each epoch is carried out by the differentiation capability implemented in Pytorch. The training parameters are specified in the Table \ref{tab:train_params}.

\subsubsection{Hyper-parameter selection}\label{hyper_param}
During the GAN training, hyper-parameters are the mini-batch size, learning rate, number of epochs, generator and discriminator optimizer, activation function, number of layers and the number of nodes in each layer. These hyper-parameters are chosen based on a grid search. Since the size of the grid is huge, we choose only $20$ randomly-selected grid points. For each grid point, we train the network and assess the the generated samples. The assessment is based on the following two criteria. 1) Generalizablity: The capability of $G(\bz)$ to generate realistic-like samples. This is measured by the rate at which discriminator can correctly distinguish the true samples among the fake samples. Ideally, the rate has to be as close as to $0.5$. 2) Mode collapse: this criterion is used to make sure that the generator is capable of generating different samples. This can be measured by calculating the distance (norm-2) between the generated samples for many \emph{different} random vectors $\bz$ as input.

\subsection{Simulation Scenarios}\label{sim_scenarios}
 We consider the following simulation scenarios. \label{L_known}Note that in all simulation scenarios, we assume that $L$ is known (or at least a maximum value of $L$ is chosen).
\subsubsection{UP-LMMSE} In this scenario, LMMSE-based channel estimation is used for uplink training only. We assume that all $64$ antennas as well as all $16$ subcarriers are used for training.

\subsubsection{UP-GAN} \label{GAN_UP_scenario} This is our DGM-based channel estimation technique in the uplink. Similar to the LMMSE-based technique we use all $64$ antennas as well as all $16$ subcarriers in the uplink\label{up_GAN_pg}. In this scenario, we solve \eqref{LS_opt_x_up_expressed_in_z_unconst} to find $\bz$ (or $\left(\balpha, \boldsymbol{\tau}, \btheta \right)$) and $\boldsymbol{\phi}^{\rm up}$ followed by channel reconstruction using \eqref{UL_ch_model_def}.


\subsubsection{DL-GAN} \label{GAN_DL_scenario} In this scenario, we aim to obtain $\bh_k^{\rm dl}$ using the estimates of $\left(\balpha, \boldsymbol{\tau}, \btheta \right)$, obtained using UP-GAN, and $\boldsymbol{\phi}^{\rm dl}$. Given $\left(\balpha, \boldsymbol{\tau}, \btheta \right)$, we find $\boldsymbol{\phi}^{\rm dl}$ using \eqref{LS_opt_dl}.

\begin{table}
\caption {GAN Training Parameters} \label{tab:train_params}
\centering
\begin{tabular}{ c||c  }

 \hline
 Training dataset size & $120000$ \\
 \hline
 Test dataset size & $30000$\\
\hline
Mini batch size & $512$\\
\hline
Epochs & $20000$\\
\hline
Optimizer & Adam\\
\hline
Learning rate, $\beta$ & $10^{-3}$, $(0.9, 0.999)$\\
\hline
$\lambda_1$, $\lambda_2$ & $10^{-2}$, $10^{-2}$\\
\hline
 \end{tabular}
\end{table}

\subsubsection{DL-Full-Reciprocity} In this scenario, similar to DL-GAN (explained in \ref{GAN_DL_scenario}), we aim to obtain $\bh_k^{\rm dl}$. Here, there is no downlink training and it only serves as a benchmark in our simulations. There are two cases in this scenario. In the first case, we assume that full reciprocity exists between the uplink and downlink. Specifically, we assume $\boldsymbol{\phi}^{\rm dl}= \boldsymbol{\phi}^{\rm up}$. In the second case, we assume $\boldsymbol{\phi}^{\rm dl}= 2 \pi f_{\rm dl}\boldsymbol{\tau}$. In both cases, we reconstruct $\bh_k^{\rm dl}$ using $\left(\balpha, \boldsymbol{\tau}, \btheta, \boldsymbol{\phi}^{\rm dl} \right)$.

\subsubsection{DL-LS} \label{DL_LS_def}
In this scenario, we ignore that $\balpha^{\rm up} \approx \balpha^{\rm dl}$. We first estimate $\left(\balpha^{\rm up}, \boldsymbol{\tau}, \btheta \right)$ during the uplink training. Then, we reuse $\left(\boldsymbol{\tau}, \btheta \right)$ in order to estimate $\balpha^{\rm dl}$ and $\boldsymbol{\phi}^{\rm dl}$ using the same pilots budget used in DL-GAN. Defining $\brho \triangleq \left[\rho_1\,\, \rho_2\, \cdots \, \rho_L\right]^T$ $\rho_l \triangleq \alpha_l^{\rm dl}e^{j{\phi}_l^{\rm dl}}$, we can rewrite \eqref{DL_received signal_non_linear} as $\by^{\rm dl} = \tilde{\mathcal{B}}\brho +\bn^{\rm dl}$, for some matrix $\tilde{\mathcal{B}}$. We then use least squares to find $\brho$.

\subsubsection{DL-Modified-R2F2}
 In this scenario, we ignore the mapping function $G_{{\cal W}_g}(\cdot)$ and directly solve \eqref{LS_opt_x_up}. The constrained optimization problem in \eqref{LS_opt_x_up} is solved \label{r2f2_numerical}numerically using an approach based on the R2F2 algorithm proposed in \cite{DinaR2F2}. Note that the R2F2 algorithm, in its original form, is based on the assumption that the downlink channel parameters are the same as its uplink counterparts. As explained in Section \ref{FDD_partial_reciprocity}, this assumption is found to be invalid due to the frequency dependency of $\phi_l^{\rm dl}$ and $\phi_l^{\rm up}$. To account for the disparity in phases, we first solve \eqref{LS_opt_x_up} using a similar technique provided in \cite{DinaR2F2}. Then, using the so-obtained $\left(\balpha, \boldsymbol{\tau} ,\btheta\right)$, we use \eqref{LS_opt_dl} to find $\phi_l^{\rm dl}$. Throughout the simulations, we refer to this technique as \emph{DL-Modified-R2F2}. In this scenario, we assume that $|\mathcal{M}_{\rm dl}|=64$.

To implement DL-Modified-R2F2, the optimization problem \eqref{LS_opt_x_up} is solved using coordinate descent. This approach involves the division of the parameters of the optimization problem into smaller sets for which the constraints are separable. The optimization is then carried out over each of these sets iteratively while treating the variables in the other sets to be constants, thereby reducing the computation complexity. Conceptually, in the feasible set, the algorithm iteratively converges to a minimum by taking strides along directions parallel to the parameter-set axes. In this case, the separability of constraints is obtained by taking $\balpha$, $\boldsymbol{\tau}$, $\btheta$, and $\boldsymbol{\phi}^{\rm up}$ as four parameter sets all of which have box constraints. Since the objective function is non-convex, the global optimality of this technique is not guaranteed. To avoid local minima, we initiate the optimization from $10$ randomly-chosen initial points and choose the solution with the least value of the objective function.

Note that the above R2F2-based technique is highly sensitive to gradient descent step that we choose for each parameter. This technique is also very slow since it requires multiple random initializations as well as a large number of iterations to converge.

\begin{figure*}[!h]
    \centering
    \begin{minipage}[t]{.47\textwidth}
\psfrag{NMSE}[c][c]{{\small NMSE (dB)}}
\psfrag{SNR}[c][c]{ {\small SNR (dB)}}
\psfrag{downlink, Modified R2-F2}{ {\small DL-Modified-R2F2}}
\psfrag{uplink, GAN, K16, M64}{ {\small UP-GAN, $|\mathcal{M}_{\rm up}|=64$}}
\psfrag{uplink, LMMSE, K16, M64,123456}{ {\small UP-LMMSE, $|\mathcal{M}_{\rm up}|=64$}}
\psfrag{downlink, MMSE, K16, M64}{ {\small DL-LMMSE, $|\mathcal{M}_{\rm up}|=64$}}
\psfrag{Full reciprocity}{ {\small DL-Full-Reciprocity, $\boldsymbol{\phi}^{\rm dl}= \boldsymbol{\phi}^{\rm up}$}}
\psfrag{Full reciprocity, phase adjust}{ {\small DL-Full-Reciprocity, $\boldsymbol{\phi}^{\rm dl}= 2 \pi f_{\rm dl}\boldsymbol{\tau}$}}
\psfrag{downlink, GAN, K16, M64}{ {\small DL-GAN, $|\mathcal{M}_{\rm dl}|=64$}}
\psfrag{downlink, GAN, K16, M32}{ {\small DL-GAN, $|\mathcal{M}_{\rm dl}|=32$}}
\psfrag{downlink, GAN, K16, M16}{ {\small DL-GAN, $|\mathcal{M}_{\rm dl}|=16$}}
\psfrag{downlink, GAN, K16, M4}{ {\small DL-GAN, $|\mathcal{M}_{\rm dl}|=4$}}
\psfrag{downlink, Least-squares}{ {\small DL-LS}}
\resizebox{!}{6.6cm}{\includegraphics{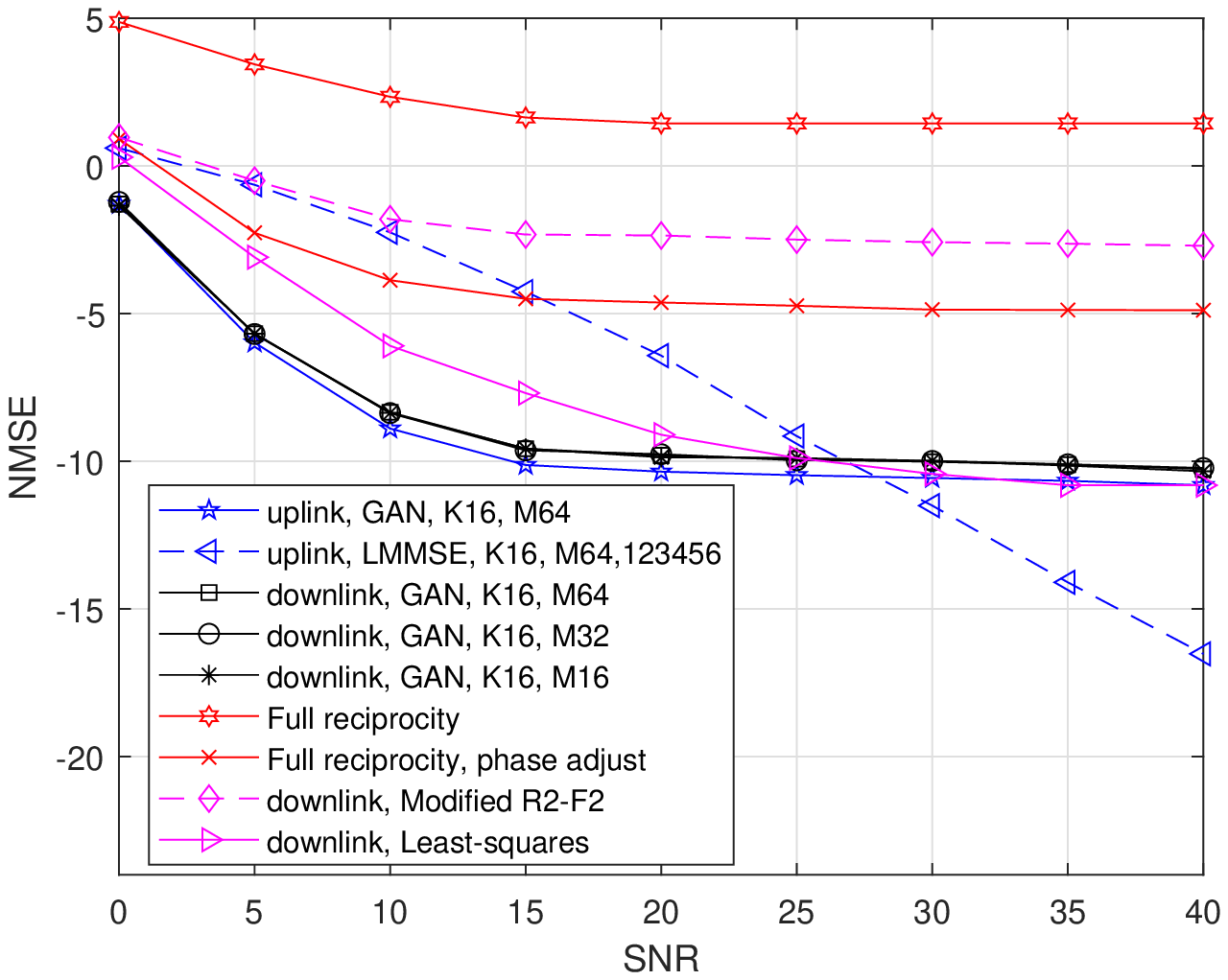}}
\caption{NMSE vs. SNR, $|\mathcal{K}_{\rm up}|=|\mathcal{K}_{\rm dl}|=K=16$, and $p=|\mathcal{K}_{\rm dl}|$.}
\label{NMSE_1}
    \end{minipage}%
    \hfill
    \begin{minipage}[t]{0.47\textwidth}
       \centering{
\psfrag{NMSE}[c][c]{{\small  NMSE (dB)}}
\psfrag{sub}[c][c]{ { $p$}}
\psfrag{downlink, Modified R2-F2, M64}{ {\small DL-Modified-R2F2}}
\psfrag{uplink, GAN, K16, M64}{ {\small UP-GAN, $|\mathcal{M}_{\rm up}|=64$}}
\psfrag{uplink, MMSE, K16, M4}{ {\small UP-LMMSE, $M=4$}}
\psfrag{downlink, GAN, PT70, M64}{ {\small DL-GAN, $|\mathcal{M}_{\rm dl}|=64$}}
\psfrag{downlink, GAN, PT70, M16}{ {\small DL-GAN, $|\mathcal{M}_{\rm dl}|=16$}}
\psfrag{downlink, GAN, PT70, M32}{ {\small DL-GAN, $|\mathcal{M}_{\rm dl}|=32$}}
\resizebox{!}{6.6cm}{\includegraphics{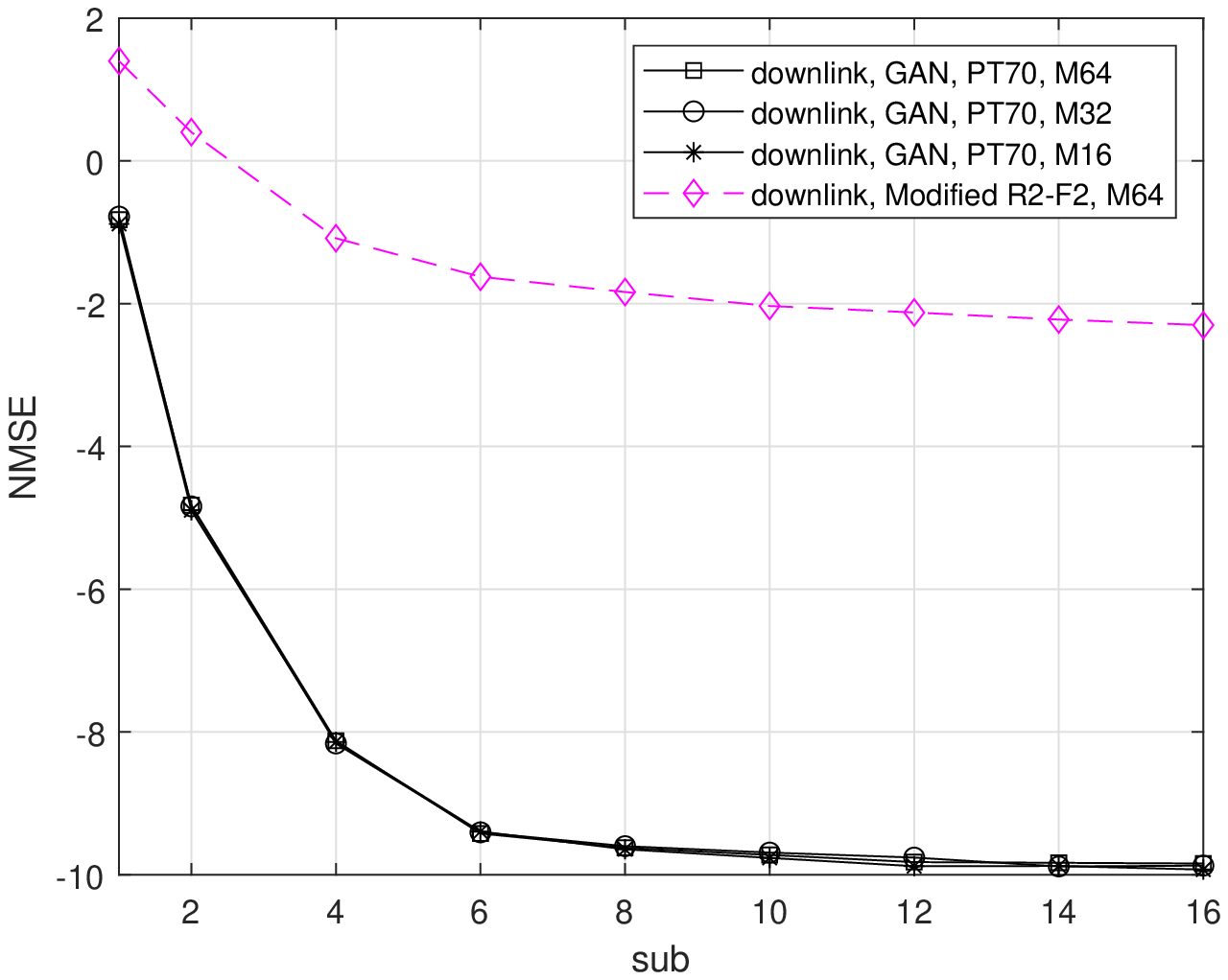}}
\caption{NMSE vs. $p$, with $|\mathcal{K}_{\rm dl}|=p$, $|\mathcal{K}_{\rm up}|=K=16$, and SNR= 20 (dB).}
\label{NMSE_2}}
    \end{minipage}
\end{figure*}

\subsection{Simulation Results}
\subsubsection{Normalized MSE} We first use the normalized mean squared error (NMSE) averaged across all subcarriers and is given by
\begin{align}
{\rm NMSE} = \frac{1}{K} \sum_{k=1}^{K}\mathbb{E}\left[ \frac{\|\bh_k - \hat{\bh}_k\|^2}{\|\bh_k \|^2}\right],
\end{align}
where $\bh_k$ and  $\hat{\bh}_k$ are, respectively, the true and the estimated channel vector at the $k$-th subcarrier. We later expand our simulations to calculate the corresponding rate and symbol error rate (SER). Throughout the simulations, we assume that $p= |\mathcal{K}_{\rm dl}|$, and $\varepsilon = 0.01$.

Fig. \ref{NMSE_1} depicts NMSE vs. SNR. In the uplink, we compare the performance of our DGM-based technique (i.e., UP-GAN) with the UP-LMMSE assuming that all $64$ antennas as well as all $16$ subcarriers are utilized for training. As shown in this figure, our DGM-based technique outperforms, by a large gap, the UP-LMMSE approach in all practical ranges of SNR. Such a performance gap stems from the fact that UP-LMMSE cannot perform well in this range of SNR, since the pilot measurements
received are of poor quality. In contrast, by exploiting the prior knowledge of the underlying distribution of channel parameters, captured using DGMs, UP-GAN outperforms the LMMSE-based technique by a large margin. In other words, since $G_{{\cal W}_g}(\cdot)$ is trained to generate a realistic channel parameters from the low-dimensional $\bz$, the UP-GAN obtains an estimate of $\left(\balpha, \boldsymbol{\tau}, \btheta\right)$ leading to an accurate estimate of the channels even at  very low SNR.

As we increase the SNR, beyond $15$ dB, we observe that there is no further improvement in the performance of DGM-based technique. Such an error floor is mainly due to the limitations in representation capability of $G_{{\cal W}_g}(\cdot)$. The generator $G_{{\cal W}_g}(\cdot)$ cannot generate the exact channel parameters $\left(\balpha, \boldsymbol{\tau}, \btheta \right)$; it can only generate realistic samples. That is, the DGM-based technique cannot always outperform the LMMSE-based technique, but it can do much better in practical ranges of SNR.

In the downlink, we plotted the performance of DL-GAN for different values of $|\mathcal{M}_{\rm dl}|$. As shown, DL-GAN behaves similar to UP-GAN. This is expected because, in both scenarios, we use the same $G_{{\cal W}_g}(\cdot)$, and their performance is mainly related to the representation capability of $G_{{\cal W}_g}(\cdot)$.
Compared to the DL-Full-Reciprocity scenario, DL-GAN performs significantly  better. This indicates that the full-reciprocity assumption between the uplink and downlink does not hold. Indeed, in DL-GAN, we used a short training sequence over fewer antennas to estimate  $\boldsymbol{\phi}^{\rm dl}$ (the frequency-specific component of the channel), which yields a large performance gain in downlink channel estimation.

\label{DL_LS_explain}Additionally, DL-LS \emph{marginally} outperforms the GAN technique at very high ranges of SNR. In particular, for this range of SNR, DL-LS provides a better estimate for $(\balpha^{\rm dl}, \boldsymbol{\phi^{\rm dl}})$ compared to what DL-GAN does. For low to medium ranges of SNR, the proposed DL-GAN technique provides a significantly better estimate of the downlink channel. This performance gap stems from the fact that DL-LS cannot perform well in this range of SNR, since the pilot measurements are of poor quality. In contrast, by exploiting the prior knowledge of the underlying distribution of channel parameters, captured using DGMs, DL-GAN provides a better estimate of channel parameters $(\balpha^{\rm dl}, \boldsymbol{\phi^{\rm dl}})$ at a low-to-medium range of SNR.

Note that, if we ignore the mapping function $G_{{\cal W}_g}(\cdot)$ (i.e., directly solving \eqref{LS_opt_x_up}, which is exactly what DL-Modified-R2F2 does), the steepest descent would suffer from convergence issues due to a larger search space as well as the box constraint over each dimension. Furthermore, the solution in this case is noise-sensitive, meaning that, at low SNRs, although we might be able to minimize the objective function to some extent, the solution is far from its true value. Instead, in our DGM-based technique, we address these issues by incorporating $G_{{\cal W}_g}(\cdot)$ as a sort of  prior. The search space is limited to the domain of $\bz$, which has a much lower dimension compared to $\bx=\left(\balpha, \boldsymbol{\tau}, \btheta \right)$.
\begin{figure*}[h]
    \centering
    \begin{minipage}[t]{.47\textwidth}
        \centering{
\psfrag{NMSE}[c][c]{{\small  NMSE (dB)}}
\psfrag{M}[c][c]{ {\small $|\mathcal{M}_{\rm dl}|$}}
\psfrag{downlink, GAN, K16}{ {\small DL-GAN, $p=16$}}
\psfrag{downlink, GAN, K8}{ {\small DL-GAN, $p=8$}}
\psfrag{downlink, GAN, K4}{ {\small DL-GAN, $p=4$}}
\resizebox{!}{6.6cm}{\includegraphics{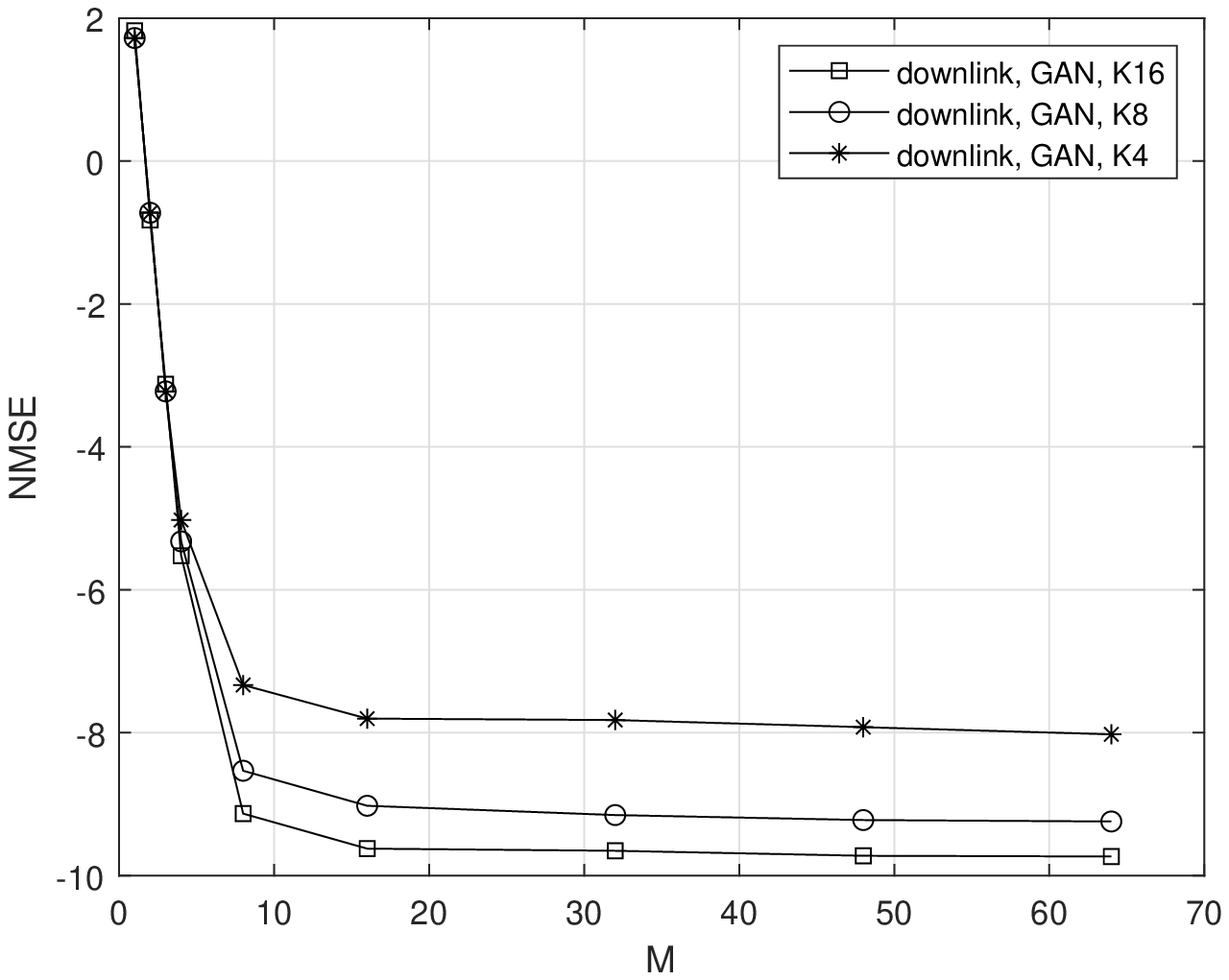}}
\caption{NMSE vs. $|\mathcal{M}_{\rm dl}|$, with $|\mathcal{K}_{\rm up}|=K=16$, $|\mathcal{K}_{\rm dl}|=p$, and SNR= 20 (dB).}
\label{NMSE_3}}
    \end{minipage}%
    \hfill
    \begin{minipage}[t]{0.47\textwidth}
       \centering{
\psfrag{Rate}[c][c]{{\small Rate (bps/Hz)}}
\psfrag{SNR}[c][c]{ {\small SNR (dB)}}
\psfrag{uplink, known}{ {\small UP, Perfect CSI}}
\psfrag{downlink, known}{ {\small DL, Perfect CSI}}
\psfrag{uplink, MMSE, M64, K16,12345678}{ {\small UP-LMMSE, $|\mathcal{M}_{\rm up}|=64$}}
\psfrag{uplink, GAN, M64, K16}{ {\small UP-GAN, $|\mathcal{M}_{\rm up}|=64$}}
\psfrag{downlink, MMSE, M64, K16}{ {\small DL-LMMSE, $|\mathcal{M}_{\rm up}|=64$}}
\psfrag{downlink, GAN, M64, K16}{ {\small DL-GAN, $|\mathcal{M}_{\rm dl}|=64$}}
\psfrag{downlink, GAN, M32, K16}{ {\small DL-GAN, $|\mathcal{M}_{\rm dl}|=32$}}
\psfrag{downlink, GAN, M16, K16}{ {\small DL-GAN, $|\mathcal{M}_{\rm dl}|=16$}}
\psfrag{downlink, Modified R2F2}{ {\small DL-Modified-R2F2}}
\psfrag{Full reciprocity}{ {\small DL-Full-Reciprocity, $\boldsymbol{\phi}^{\rm dl}= \boldsymbol{\phi}^{\rm up}$}}
\psfrag{Full reciprocity, phase adjust}{ {\small DL-Full-Reciprocity, $\boldsymbol{\phi}^{\rm dl}= 2 \pi f_{\rm dl}\boldsymbol{\tau}$}}
\resizebox{!}{6.6cm}{\includegraphics{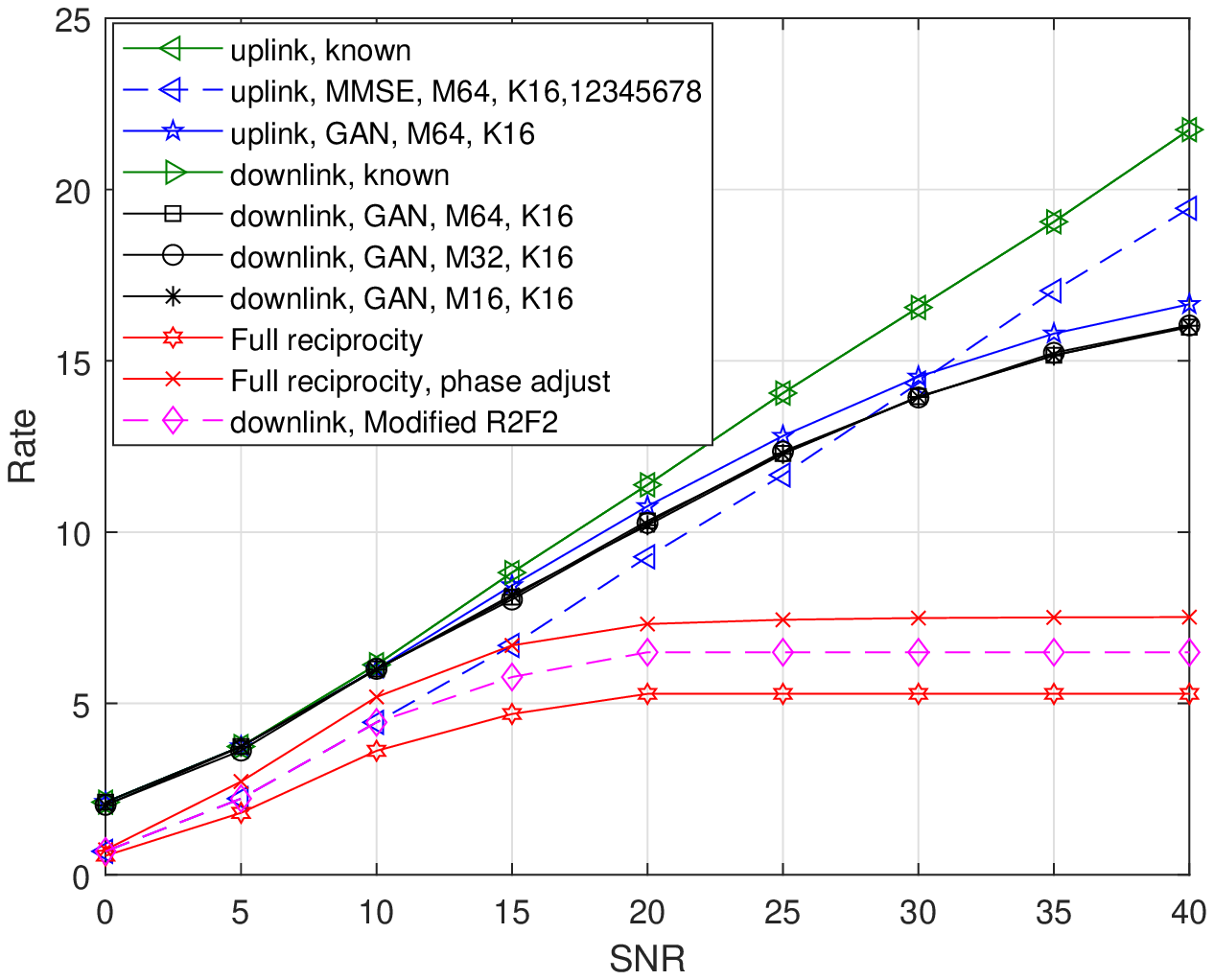}}
\caption{Rate per subcarrier vs. SNR, $|\mathcal{K}_{\rm up}|=|\mathcal{K}_{\rm dl}|=K=16$, and $p=|\mathcal{K}_{\rm dl}|$.}
\label{Rate_1}}
    \end{minipage}
\end{figure*}

In addition, the weights ${\cal W}_g$ of the generator $G_{{\cal W}_g}(\cdot)$ encode a probability distribution over the space of $\left(\balpha, \boldsymbol{\tau}, \btheta \right)$, such that we can draw samples from that distribution from a low-dimensional standard Gaussian distribution. This implies that any sample generated by $G_{{\cal W}_g}(\cdot)$ satisfies the constraints in \eqref{LS_opt_x_up}(as shown in footnote \ref{foot_const}). Therefore, our DGM-based technique performs well even with low-dimensional noisy pilot measurements. This explains the robustness of our proposed technique to the noise level.

It is worth mentioning that, if we have enough pilot measurements, the performance of the proposed DL-GAN is independent of $|\mathcal{M}_{\rm dl}|$ as long as $|\mathcal{M}_{\rm dl}|\geq L$. The reason is that the training power $P_T$ in each time slot is distributed among the pilot symbols transmitted across the $|\mathcal{M}_{\rm dl}|$ antennas in the downlink. Indeed, adding more antennas for downlink training does not affect the received SNR at the UE. However, when $|\mathcal{M}_{\rm dl}|< L$, the $p|\mathcal{K}_{\rm dl}|\times L$ matrix $\mathcal{B}$ in \eqref{DL_received signal_non_linear} is of rank $|\mathcal{M}_{\rm dl}|$, leading to an identifiability issue in recovering $\{\phi_l^{\rm dl}\}_{l=1}^L$ from \eqref{DL_received signal_non_linear}. This becomes more evident in Fig. \ref{NMSE_3}, where we plot NMSE vs. $|\mathcal{M}_{\rm dl}|$. Note that the requirement that $|\mathcal{M}_{\rm dl}|\geq L$ can be relaxed if train the channel in the time domain\cite{8490886, 8647635}. However, the time domain approach is beyond the scope of this paper.

Fig. \ref{NMSE_2} plots NMSE vs. $p$, where we compare the performance of DL-GAN (for different $|\mathcal{M}_{\rm dl}|$) with DL-Modified-R2F2 at SNR = $20$ (dB). The performance gap between the DL-Modified-R2F2 and our DGM-based technique still exists even by increasing $p$ (adding more observations). This indicates the lack of convergence in DL-Modified-R2F2 scenario. In the DL-GAN scenario, adding more pilot measurements, either by increasing  $p$ or $|\mathcal{K}_{\rm dl}|$, will improve the channel estimation performance. \label{K_increase}The reason is that the number of underlying channel parameters is independent of the number subcarriers or the length of pilots. Therefore, by increasing either of these quantities, we collect more observation that in turn, yields a better estimate of the same number of channel parameters. However, as shown, such an improvement is not consistent, i.e., DL-GAN does not always gain by increasing the amount of pilot measurements. This is mainly limited by the accuracy of the estimation of $\left(\balpha, \boldsymbol{\tau}, \btheta \right)$ in the uplink, which is dominated by the representation capability of $G_{{\cal W}_g}(\cdot)$.

Similar to what we observed in Fig. \ref{NMSE_1}, when $|\mathcal{M}_{\rm dl}|\geq L$, adding more antennas for downlink training does not improve the performance. For very small $p$ (i.e., $p=1$ and $2$), the DL-GAN yields poor performance even with large $|\mathcal{M}_{\rm dl}|$. This limited performance is due to the fact that the matrix $\mathcal{B}$ in \eqref{DL_received signal_non_linear} is not full-column rank in this range of $p$, and imposes an identifiability issue in recovery of $\{\phi_l^{\rm dl}\}_{l=1}^L$ from \eqref{DL_received signal_non_linear}.

To further assess the performance of our DGM-based technique,  Fig. \ref{NMSE_3} plots the NMSE vs. $|\mathcal{M}_{\rm dl}|$, the number of antennas used for downlink training. It is shown that, for fixed $p$ (or $|\mathcal{K}_{\rm dl}|$), more pilot measurements does not always benefit the proposed DGM-based technique. Similar to what we observed in Fig. \ref{NMSE_2}, this observation can be attributed to the accuracy of the estimate of $\left(\balpha, \boldsymbol{\tau}, \btheta \right)$ obtained via the uplink training, which, as mentioned, is dominated by the representation capability of $G_{{\cal W}_g}(\cdot)$ .

\begin{figure*}[t]
    \centering
    \begin{minipage}[t]{.47\textwidth}
        \centering{
\psfrag{Rate}[c][c]{{\small Rate (bps/Hz)}}
\psfrag{sub}[c][c]{ {\small $p$}}
\psfrag{downlink, known}{ {\small DL, Perfect CSI}}
\psfrag{downlink, Modified R2F2}{ {\small DL-Modified-R2F2}}
\psfrag{downlink, GAN, M64, PT70}{ {\small DL, GAN, $|\mathcal{M}_{\rm dl}|=64$}}
\psfrag{downlink, GAN, M8, PT70}{ {\small DL-GAN, $|\mathcal{M}_{\rm dl}|=16$}}
\psfrag{downlink, GAN, M32, PT70}{ {\small DL-GAN, $|\mathcal{M}_{\rm dl}|=32$}}
\resizebox{!}{6.6cm}{\includegraphics{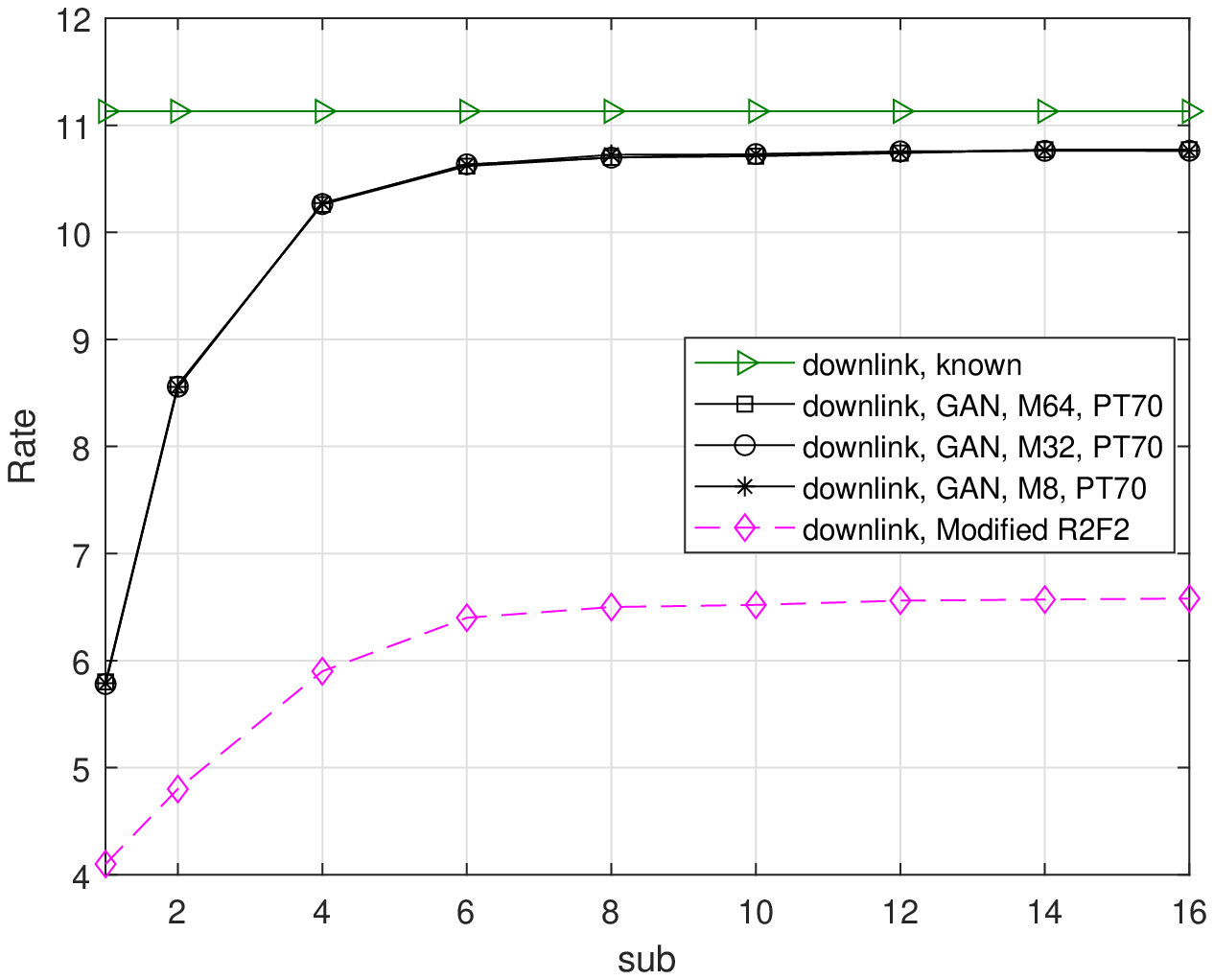}}
\caption{Rate per subcarrier vs. $p$, with $|\mathcal{K}_{\rm dl}|=p$, $|\mathcal{K}_{\rm up}|=K=16$, and SNR= 20 (dB).}
\label{Rate_2}}
    \end{minipage}%
    \hfill
    \begin{minipage}[t]{0.47\textwidth}
       \centering{
\psfrag{Rate}[c][c]{{\small Rate (bps/Hz)}}
\psfrag{M}[c][c]{ {\small $|\mathcal{M}_{\rm dl}|$}}
\psfrag{downlink, known}{ {\small DL, Perfect CSI}}
\psfrag{downlink, GAN, K16}{ {\small DL-GAN, $p=16$}}
\psfrag{downlink, GAN, K8}{ {\small DL-GAN, $p=8$}}
\psfrag{downlink, GAN, K4}{ {\small DL-GAN, $p=4$}}
\resizebox{!}{6.6cm}{\includegraphics{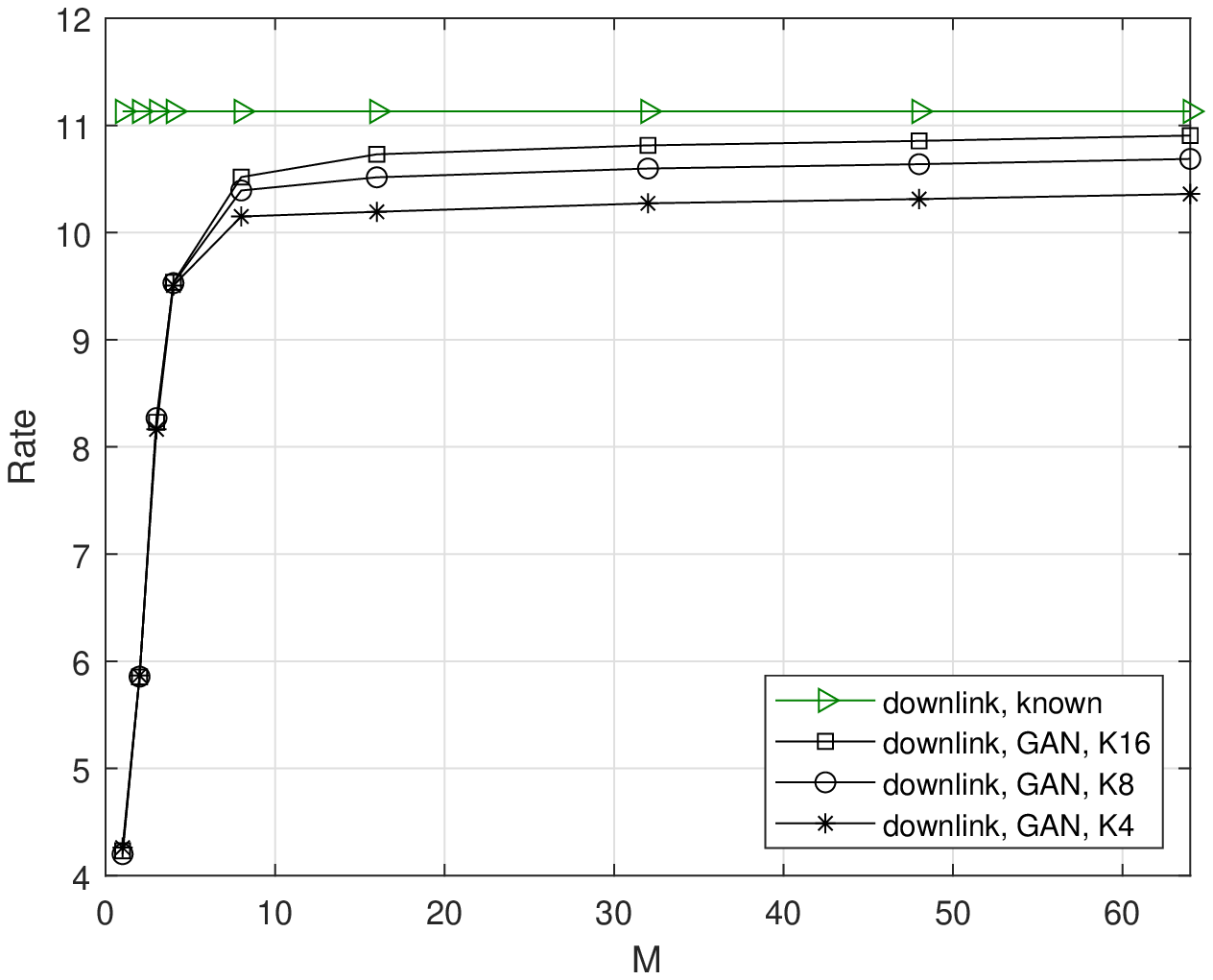}}
\caption{Rate per subcarrier vs. $|\mathcal{M}_{\rm dl}|$, with $|\mathcal{K}_{\rm up}|=K=16$, $|\mathcal{M}_{\rm up}|=M=64$, $|\mathcal{K}_{\rm dl}| =p$ and SNR= 20 (dB).}
\label{Rate_3}}
    \end{minipage}
\end{figure*}

\subsubsection{Rate}
In this subsection, we explore the impact of our proposed technique on the achievable rate. The plots in Figs. \ref{Rate_1}, \ref{Rate_2}, and \ref{Rate_3} correspond to the NMSE curves we studied in the previous subsection. Fig. \ref{Rate_1} plots the rate vs. SNR. In the uplink, the UP-GAN outperforms the UP-LMMSE technique in practical ranges of SNR as it does in Fig. \ref{NMSE_1}. This is due to a better estimate of channel matrix in this range of SNR. This, by itself, is attributed to the rich prior stored in the weights of $G_{{\cal W}_g}(\cdot)$ network. For the same reason, the DL-GAN yields a much better rate performance compared to what  DL-Full-Reciprocity and DL-Modified-R2F2 do. Note that the saturation in rate at high SNR is related to the error floor in channel estimation, which comes from the limits in representation capability of $G_{{\cal W}_g}(\bz)$.

In Figs. \ref{Rate_2} and \ref{Rate_3}, assuming SNR= $20$ (dB), we assess the rate performance of our DGM-based technique vs. $p$ and $|\mathcal{M}_{\rm dl}|$, respectively. Here, we assume that the frequency-independent features $\left(\balpha, \boldsymbol{\tau}, \btheta \right)$ are estimated in the uplink as in the UP-GAN scenario. The goal here is to show how the number of downlink pilot observations will affect the rate performance of the proposed technique. As shown in these figures, the rate quickly saturates by increasing $p$ or $|\mathcal{M}_{\rm dl}|$. In other words, adding more training observations does not always improve the rate performance. This can be explained as follows. Due to the partial reciprocity between the uplink and downlink channels, the proposed DGM-based technique will use the estimate of $\left(\balpha, \boldsymbol{\tau}, \btheta \right)$ obtained via the uplink training. Since, these features are not perfectly estimated, mainly because of the representation capability of $G_{{\cal W}_g}(\cdot)$, they will affect the estimation of $\boldsymbol{\phi}^{\rm dl}$ in downlink in a way that adding more training observation does not improve the estimation of $\boldsymbol{\phi}^{\rm dl}$, and therefore, the estimate of channel. The same argument holds true when we increase the training power as shown in Fig.~\ref{Rate_1}.

\begin{figure}
\centering{
\psfrag{SER}[c][c]{{\small SER}}
\psfrag{SNR}[c][c]{ {\small SNR (dB)}}
\psfrag{uplink, known, 1234567890123456}{ {\small UP, Perfect CSI}}
\psfrag{downlink, known}{ {\small DL, Perfect CSI}}
\psfrag{uplink, MMSE, M64, K16}{ {\small UP-LMMSE, $|\mathcal{M}_{\rm up}|=64$}}
\psfrag{downlink, MMSE, M64, K16}{ {\small DL-LMMSE, $|\mathcal{M}_{\rm up}|=64$}}
\psfrag{uplink, GAN, M64, K16}{ {\small UP-GAN, $|\mathcal{M}_{\rm up}|=64$}}
\psfrag{downlink, GAN, M64, K16}{ {\small DL-GAN, $|\mathcal{M}_{\rm dl}|=64$}}
\psfrag{downlink, GAN, M8, K16}{ {\small DL-GAN, $|\mathcal{M}_{\rm dl}|=16$}}
\psfrag{downlink, GAN, M32, K16}{ {\small DL-GAN, $|\mathcal{M}_{\rm dl}|=32$}}
\psfrag{downlink, Modified R2-F2}{ {\small DL-Modified-R2F2}}
\psfrag{Full reciprocity}{ {\small DL-Full-Reciprocity, $\boldsymbol{\phi}^{\rm dl}= \boldsymbol{\phi}^{\rm up}$}}
\psfrag{Full reciprocity, phase adjust}{ {\small DL-Full-Reciprocity, $\boldsymbol{\phi}^{\rm dl}= 2 \pi f_{\rm dl}\boldsymbol{\tau}$}}
\resizebox{!}{6.6cm}{\includegraphics{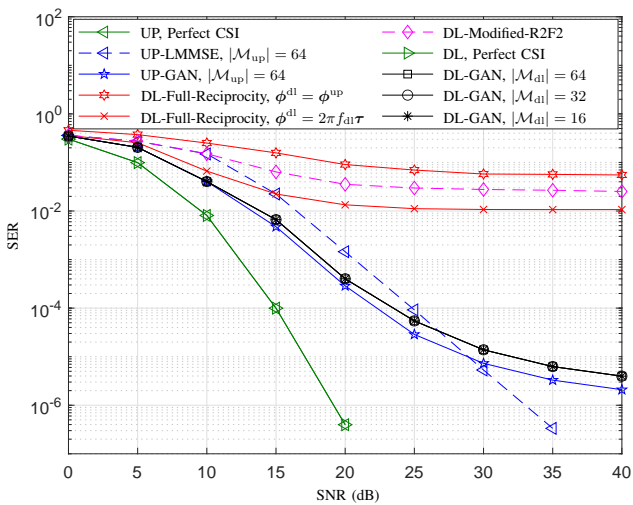}}
\caption{SER per subcarrier vs. SNR, $|\mathcal{K}_{\rm up}|=|\mathcal{K}_{\rm dl}|=K=16$ and $p=|\mathcal{K}_{\rm dl}|$.}
\label{SER_1}}
\end{figure}

\subsubsection{Symbol Error Rate}
The SER vs. SNR plot is given in Fig.~\ref{SER_1}, where we use QPSK modulation. As can be seen from this figure, in practical ranges of SNR, the UP-GAN yields a better SER performance compared to our benchmark UP-LMMSE. This observation is expected since the UP-GAN provides a better estimate of the channel for this range of SNR. As stated before, due to the limits in representation capability of $G_{{\cal W}_g}(\cdot)$, which leads to an error floor in channel estimation (see Fig. \ref{NMSE_1}), the SER does not always improve by increasing the SNR. On the contrary, the SER performance of UP-LMMSE improves as we increase SNR.

In the downlink, the SER of DL-GAN hits a floor limit as we increase SNR. This is due to the fact that the increase in SNR does not always improve the channel estimate in DL-GAN (see Fig. \ref{NMSE_1}). This observation is mainly derived by the limits imposed by $G_{{\cal W}_g}(\cdot)$. Ignoring the role of $G_{{\cal W}_g}(\bz)$, as to what DL-Modified-R2F2 does, yields a huge SER loss. Overall, our DGM-based technique yields a performance gain at low to mid range of SNR.

\begin{figure}
\centering{
\psfrag{NMSE}[c][c]{{\small NMSE (dB)}}
\psfrag{sigma}[c][c]{ {$\sigma_{\phi}$ (degrees)}}
\psfrag{downlink, GAN, SNR10,123456}{ {\small DL-GAN, SNR = $10$ (dB)}}
\psfrag{downlink, GAN, SNR20}{ {\small DL-GAN, SNR = $20$ (dB)}}

\resizebox{!}{6.6cm}{\includegraphics{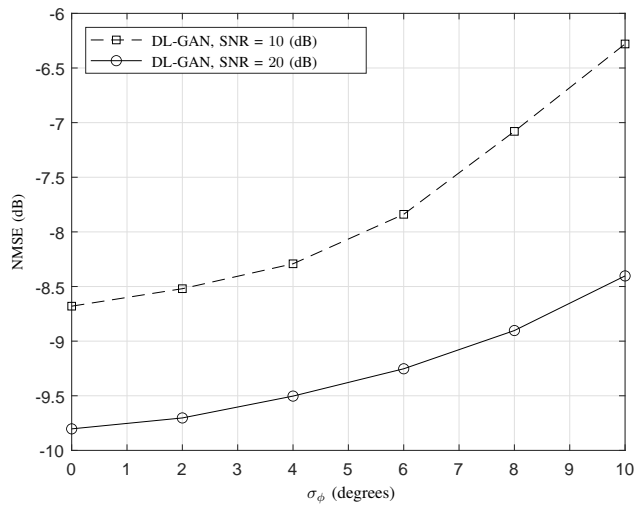}}
\caption{NMSE vs $\sigma_{\phi}$, $|\mathcal{K}_{\rm up}|=|\mathcal{K}_{\rm dl}|=K=16$, $|\mathcal{M}_{\rm dl}|=64$, and $p=|\mathcal{K}_{\rm dl}|$.}
\label{NMSE_vs_feed_1}}
\end{figure}

\subsubsection{Feedback error}\label{feedback_err}
So far we assumed that there is no error when the estimates of $\boldsymbol{\phi}^{\rm dl}$ are fed back to the BS. However, to show the effect of feedback error, we assume that the BS receives $\boldsymbol{\tilde{\phi}}^{\rm dl}$, a noisy version of the $\boldsymbol{\phi}^{\rm dl}$, i.e., $\boldsymbol{\tilde{\phi}}^{\rm dl} = \boldsymbol{\phi}^{\rm dl} + \varepsilon_{\phi}$, where $\varepsilon_{\phi}$ is a Gaussian random variable with zero mean and standard deviation $\sigma_{\phi}$. In Fig. \ref{NMSE_vs_feed_1}, we plot the downlink channel NMSE versus $\sigma_{\phi}$,  for ${\rm SNR} = 10 ,\, 20 \,{\rm (dB)}$. As expected, considering the feedback error will degrade the downlink channel NMSE. However, this degradation is not the same for different SNRs. As we increase $\sigma_{\phi}$, the NMSE is considerably less affected at ${\rm SNR} = 20$ dB compared to that in ${\rm SNR} = 10$ dB. This is mainly because, at high SNR, the estimates of channel parameters have a better quality compared to  low SNR regimes.

\textbf{Remark 3}:\label{remark3} The proposed channel estimation technique in this paper is not limited to ULAs, and it can be extended to 2D arrays. As long as the model relating the angles to received signals is known (i.e., the ``steering vector'') our approach will remain valid. We would have to estimate more parameters such as the elevation and azimuth angles.

\section{Conclusions}\label{Sec:conclusions}
In this paper, we proposed a deep generative model (DGM)-based technique for FDD massive MIMO downlink channel estimation. The proposed technique relies on the partial reciprocity between the uplink and downlink channels, meaning that a portion of the underlying channel parameters are frequency-independent, and they are shared in both uplink and downlink channels. These parameters are estimated via the uplink training. Then, the frequency-specific channel parameters are estimated via downlink training using a very short training signal. To do so, we assume that the frequency-independent channel parameters can be modeled using  some probability distribution, which is learned using DGMs. We showed that our proposed DGM-based channel estimation outperforms the conventional channel estimation techniques in practical ranges of SNR. Our technique also yields a near-optimal performance using only few pilot measurements, indicating a significant reduction in training and feedback overhead in FDD massive MIMO systems. \label{future_work}The work in this paper can be further extended to incorporate the already-existing DoA estimation techniques in estimating the channel parameters as well as considering more practical scenarios such as RF calibration errors.

\bibliographystyle{IEEEtran}
\bibliography{reference}

\newpage
\onecolumn
\begin{appendices}

\section{Preliminary definitions and lemmas  }\label{AppA}
\begin{definition}\label{lip_def}
 A function $f: X \rightarrow Y$ is said to be Lipschitz-continuous if it satisfies
 \begin{align}
D_Y \left(f(\bx^{(1)}), f(\bx^{(2)})\right)\leq L_f D_X \left(\bx^{(1)}, \bx^{(2)}\right),
\end{align}
for some real-valued $L_f\geq 0$ and distant metrics $D_X$ and $D_Y$. The value of $L_f$ is known as the Lipschitz constant, and the function can be referred to as being $L_f$-Lipschitz.
\end{definition}

\begin{lemma} \label{lip_composition}
If $f_1(\bx)$ is an $L_1$-Lipschitz  function, and $f_2(\bx)$ is an $L_2$-Lipschitz  function, then $f_1(f_2(\bx))$ is an $L_1L_2$-Lipschitz function.
\end{lemma}
\emph{Proof}: We can write
\begin{align}
\left|f_1(f_2(\bx^{(i)}))- f_1(f_2(\bx^{(j)}))\right|\leq L_1 \left|f_2(\bx^{(i)})- f_2(\bx^{(j)})\right| \leq L_1L_2 \left|\bx^{(i)}- \bx^{(j)}\right|. \nonumber
\end{align}

\begin{lemma} \label{lip_sum}
If $f_n(\bx)$ is be an $L_n$-Lipschitz function, then $f(x)=\sum_{n=1}^{N}f_n(\bx)$ is an $L_f$-Lipschitz function, where $L_f \triangleq \max \left\{L_1, L_2, \cdots, L_N\right\}$
\end{lemma}
\emph{Proof}: We can write
\begin{align}
\left|f(\bx^{(i)})- f(\bx^{(j)})\right|&= \left|\sum_{n=1}^{N}f_n(\bx^{(i)})-\sum_{n=1}^{N}f_n(\bx^{(j)})\right|\leq  \sum_{n=1}^{N} \left|f_n(\bx^{(i)})-f_n(\bx^{(j)})\right|\nonumber\\
&\leq  \sum_{n=1}^{N} L_n \left|\bx^{(i)}-\bx^{(j)}\right| \leq L_f \left|\bx^{(i)}-\bx^{(j)}\right|. \nonumber
\end{align}

\begin{lemma} \label{QM_AM_inequality}
(Quadratic and Arithmetic mean inequality): Let $a_1, a_2, \cdots, a_N$ be positive real numbers, then
\begin{align}
\sqrt{\frac{a_1^2 + a_2^2 + \cdots + a_N^2}{n}}\geq \frac{a_1 + a_2 + \cdots + a_N}{n}.
\end{align}
The equality occurs when  $a_1 = a_2 = \cdots = a_N$.
\end{lemma}
\emph{Proof}: See \cite{book_Inequalities}.
\begin{lemma} \label{norm_subvec_inequality}
Given $\bx=\left[\bx_1^T \;\;\bx_2^T\right]^T$ and $\by=\left[\by_1^T\;\; \by_2^T\right]^T$, we have the following inequality
\begin{align}
\| \bx - \by \|\leq \| \bx_1 - \by_1 \| + \| \bx_2 - \by_2 \|, \nonumber
\end{align}
where the equality holds when $\bx=\by$.
\end{lemma}
\begin{proof}
We can write
\begin{align}
\| \bx - \by \|&= \sqrt{ {\| \bx_1 - \by_1 \|}^2 + {\| \bx_2 - \by_2 \|}^2} \nonumber\\
&\leq \sqrt{ \left(\| \bx_1 - \by_1 \|+ \| \bx_2 - \by_2 \|\right)^2}\nonumber\\
 &= \| \bx_1 - \by_1 \| + \| \bx_2 - \by_2 \|.\nonumber
\end{align}
\end{proof}

\section{Proof of Lipschitz continuity of gradient of $J_{\rm up}\left(G_{{\cal W}_g}(\bz), {\boldsymbol{\phi}^{\rm up}}\right)$} \label{lip}
\begin{proof}
We prove for $|\mathcal{K}_{\rm up}|=1$. The Lipschitz continuity for $|\mathcal{K}_{\rm up}|>1$ can be straightforwardly extended. Let us define ${\tilde{\bz}}\triangleq \left(\bz, \boldsymbol{\phi}\right)$ and $P(\tilde{\bz})\triangleq \triangledown_{\tilde{\bz}}J_{\rm up}\left(G_{{\cal W}_g}(\bz), {\boldsymbol{\phi}}\right)$. We aim to show that there exists $\lambda_p$ such that
\begin{align} \label{P_lip_intro}
\| P(\tilde{\bz}^{(1)}) - P(\tilde{\bz}^{(2)}) \| \leq \lambda_p \| \tilde{\bz}^{(1)} - \tilde{\bz}^{(2)}\|.
\end{align}
 Using the triangle inequality that
\begin{align} \label{fact_sum_upperbound}
\| f(a_1, a_2) - f(b_1, b_2)\| &= \| f(a_1, a_2) - f(a_1, b_2)+  f(a_1, b_2) - f(b_1, b_2)\| \nonumber\\
&\leq \| f(a_1, a_2) - f(a_1, b_2)\|+ \| f(a_1, b_2) - f(b_1, b_2)\|,
\end{align}
we can write
\begin{align} \label{P_lip_intro_split}
\| P({\bz}^{(1)}, {\boldsymbol{\phi}^{(1)}}) - P({\bz}^{(2)}, {\boldsymbol{\phi}^{(2)}}) \| \leq \underbrace{\| P({\bz}^{(1)}, {\boldsymbol{\phi}^{(1)}}) - P({\bz}^{(2)}, {\boldsymbol{\phi}^{(1)}}) \|}_{I_{\bz}}+ \underbrace{\| P({\bz}^{(2)}, {\boldsymbol{\phi}^{(1)}}) - P({\bz}^{(2)}, {\boldsymbol{\phi}^{(2)}}) \|}_{I_{\boldsymbol{\phi}}}.
\end{align}
In the next subsection, we show that $I_{\bz}$ and $I_{\boldsymbol{\phi}}$ are bounded from above.

 \subsection{Bounds on $I_{\bz}$}
 For fixed $\boldsymbol{\phi}$, we denote $f(G_{{\cal W}_g}(\bz))\triangleq \triangledown_{\tilde{\bz}}J_{\rm up}\left(G_{{\cal W}_g}(\bz), {\boldsymbol{\phi}}\right)$ a vector-valued composition function. To prove the Lipschitz-continuity of $f(G_{{\cal W}_g}(\cdot))$, we show that $f(\cdot)$  and $G_{{\cal W}_g}(\cdot)$ are both Lipschitz-continuous, and therefore, according to Lemma \ref{lip_composition}, their composition is also a Lipschitz-continuous function.

Given $\bx=G_{{\cal W}_g}(\bz)$ and defining ${\tilde{\bx}}\triangleq \left(\bx, \boldsymbol{\phi}\right)$, for fixed  $\boldsymbol{\phi}$, we first show the Lipschitz-continuity of $f(\bx)\triangleq \triangledown_{\tilde{\bx}}J_{\rm up}\left(\bx, {\boldsymbol{\phi}}\right)$. We can write $f(\bx) = \sum_{m=1}^{M} f_m(\bx)$, where $f_m(\bx),$ for $m=1,1, \cdots, M$, is given as
\begin{align} \label{f_m_def}
f_m(\bx) \triangleq -2\left(y^{\rm up}_m - h_m(\tilde{\bx})s_k\right)\triangledown_{\tilde{\bx}}h_m(\tilde{\bx}),
\end{align}
where $y^{\rm up}_m$ is the received signal at the $m$th antenna in the uplink, $h_m(\tilde{\bx})\triangleq \sum_{l=1}^{L} \alpha_l e^{j(\phi_l +\beta\tau_l + \gamma_m \sin{\theta_l})}$, $\beta \triangleq \frac{2\pi k}{K} B$, $\gamma_m \triangleq \frac{2\pi}{\lambda}md$, for $m=0,1, \cdots, M-1$. Also,
\begin{align} \label{grad_f_m_def}
&\triangledown_{\tilde{\bx}}h_m(\tilde{\bx}) = \nonumber\\
&\left[\underbrace{\frac{\partial h_m(\tilde{\bx})}{\partial \alpha_1}, \cdots, \frac{\partial h_m(\tilde{\bx})}{\partial \alpha_L}}_{=\triangledown_{\balpha}h_m(\tilde{\bx})}, \underbrace{\frac{\partial h_m(\tilde{\bx})}{\partial \tau_1}, \cdots, \frac{\partial h_m(\tilde{\bx})}{\partial \tau_L}}_{=\triangledown_{\boldsymbol{\tau}}h_m(\tilde{\bx})}, \underbrace{\frac{\partial h_m(\tilde{\bx})}{\partial \theta_1}, \cdots, \frac{\partial h_m(\tilde{\bx})}{\partial \theta_L}}_{=\triangledown_{\btheta}h_m(\tilde{\bx})}, \underbrace{\frac{\partial h_m(\tilde{\bx})}{\partial \phi_1}, \cdots, \frac{\partial h_m(\tilde{\bx})}{\partial \phi_L}}_{=\triangledown_{\boldsymbol{\phi}}h_m(\tilde{\bx})}   \right]^T,
\end{align}
where, for $l=1,2,\cdots, L$, and defining $\omega_l \triangleq \phi_l +\beta\tau_l + \gamma_m \sin{\theta_l}$, the individual partial derivatives are given by
\begin{align} \label{grad_f_m_individual_def}
\frac{\partial h_m(\tilde{\bx})}{\partial \alpha_l} &=  e^{j\omega_l}\\
\frac{\partial h_m(\tilde{\bx})}{\partial \tau_l} &= j\beta \alpha_l e^{j\omega_l}\\
\frac{\partial h_m(\tilde{\bx})}{\partial \theta_l} &= j\gamma_m \alpha_l \cos{\theta_l} e^{j\omega_l}\\
\frac{\partial h_m(\tilde{\bx})}{\partial \phi_l} &= j\alpha_l e^{j\omega_l}.
\end{align}

To prove the the Lipschitz-continuity of $f_m(\bx)$, according to Definition \ref{lip_def}, and using the $2$-norm as the metric, we show that there exists a finite $\lambda_m\geq 0$, such that
\begin{align} \label{fm_lip_def}
\| f_m(\bx^{(1)}) - f_m(\bx^{(2)}) \|  \leq  \lambda_m \|\bx^{(1)} - \bx^{(2)}\|,
\end{align}
where $\bx^{(i)} \triangleq \left(\balpha^{(i)}, \boldsymbol{\tau}^{(i)}, \btheta^{(i)}\right)$ for $i=1,2$. Using \eqref{fact_sum_upperbound}, we express the left hand side in \eqref{fm_lip_def} as
\begin{align} \label{fm_lip_def_triangle}
\| f_m(\bx^{(1)}) - f_m(\bx^{(2)}) \| & \leq  \underbrace{\| f_m(\balpha^{(1)}, \boldsymbol{\tau}^{(1)}, \btheta^{(1)}) - f_m(\balpha^{(2)}, \boldsymbol{\tau}^{(1)}, \btheta^{(1)})\|}_{\triangleq T_{\alpha}}\nonumber\\
& + \underbrace{\| f_m(\balpha^{(2)}, \boldsymbol{\tau}^{(1)}, \btheta^{(1)}) - f_m(\balpha^{(2)}, \boldsymbol{\tau}^{(2)}, \btheta^{(1)})\|}_{\triangleq T_{\tau}}\nonumber\\
& +\underbrace{ \| f_m(\balpha^{(2)}, \boldsymbol{\tau}^{(2)}, \btheta^{(1)}) - f_m(\balpha^{(2)}, \boldsymbol{\tau}^{(2)}, \btheta^{(2)})\|}_{\triangleq T_{\theta}}.
\end{align}
Next, we show that, each term on the right hand side of \eqref{fm_lip_def_triangle} is bounded.
\subsubsection{Bounds on $T_{\alpha}$}
Here, we assume that $\boldsymbol{\tau}^{(1)}$ and $ \btheta^{(1)}$ are fixed. Using the Triangle inequality of \eqref{fact_sum_upperbound}, we can write
\begin{align} \label{T_alpha_upperbound}
T_{\alpha} & \leq \| f_m(\alpha_1^{(1)}, \bar{\balpha}_1,  \boldsymbol{\tau}^{(1)}, \btheta^{(1)}) - f_m(\alpha_1^{(2)}, \bar{\balpha}_1,\boldsymbol{\tau}^{(1)}, \btheta^{(1)})\| \nonumber\\
& + \| f_m(\alpha_2^{(1)}, \bar{\balpha}_2,  \boldsymbol{\tau}^{(1)}, \btheta^{(1)}) - f_m(\alpha_2^{(2)}, \bar{\balpha}_2,\boldsymbol{\tau}^{(1)}, \btheta^{(1)})\| \nonumber\\
& \qquad  \qquad \qquad \vdots
\nonumber\\
& + \| f_m(\alpha_L^{(1)}, \bar{\balpha}_L,  \boldsymbol{\tau}^{(1)}, \btheta^{(1)}) - f_m(\alpha_L^{(2)}, \bar{\balpha}_L,\boldsymbol{\tau}^{(1)}, \btheta^{(1)})\|,
\end{align}
where $\bar{\balpha}_l\triangleq \left(\alpha_1^{(2)}, \cdots, \alpha_{l-1}^{(2)}, \alpha_{l+1}^{(1)}, \cdots, \alpha_L^{(1)}\right)$. Let us define $T_{\alpha_l} \triangleq \| f_m(\alpha_l^{(1)}, \bar{\balpha}_l,  \boldsymbol{\tau}^{(1)}, \btheta^{(1)}) - f_m(\alpha_l^{(2)}, \bar{\balpha}_l,\boldsymbol{\tau}^{(1)}, \btheta^{(1)})\|$. We now aim to show that there exists a bounded $M_{\alpha_l}\geq 0 $, such that $T_{\alpha_l}\leq M_{\alpha_l} \| \alpha_l^{(1)} - \alpha_l^{(2)} \|$. To do so, we note that
\begin{alignat}{2} \label{T_alpha_l_upperbound}
T_{\alpha_l} = 2\| &\left(y^{\rm up}_m - h_m(\alpha_l^{(1)}, \bar{\balpha}_l,  \boldsymbol{\tau}^{(1)}, \btheta^{(1)})s_k\right)\triangledown h_m(\alpha_l^{(1)}, \bar{\balpha}_l,  \boldsymbol{\tau}^{(1)}, \btheta^{(1)}) \nonumber\\
- &\left(y^{\rm up}_m - h_m(\alpha_l^{(2)}, \bar{\balpha}_l,  \boldsymbol{\tau}^{(1)}, \btheta^{(1)})s_k\right)\triangledown h_m(\alpha_l^{(2)}, \bar{\balpha}_l,  \boldsymbol{\tau}^{(1)}, \btheta^{(1)})\|\nonumber\\
\overset{(a)}{\leq} 2\| &\left(y^{\rm up}_m - h_m(\alpha_l^{(1)}, \bar{\balpha}_l,  \boldsymbol{\tau}^{(1)}, \btheta^{(1)})s_k\right)\bar{\bv}_{\alpha_l} - \left(y^{\rm up}_m - h_m(\alpha_l^{(2)}, \bar{\balpha}_l,  \boldsymbol{\tau}^{(1)}, \btheta^{(1)})s_k\right)\bar{\bv}_{\alpha_l}\| \nonumber\\
+2\| &\left(y^{\rm up}_m - h_m(\alpha_l^{(1)}, \bar{\balpha}_l,  \boldsymbol{\tau}^{(1)}, \btheta^{(1)})s_k\right)\bv_{\alpha_l}^{(1)} - \left(y^{\rm up}_m - h_m(\alpha_l^{(2)}, \bar{\balpha}_l,  \boldsymbol{\tau}^{(1)}, \btheta^{(1)})s_k\right)\bv_{\alpha_l}^{(2)}\| \nonumber\\
\overset{(b)}{=} 2\| & \bar{\bv}_{\alpha_l}s_k \| \left| h_m(\alpha_l^{(1)}, \bar{\balpha}_l,  \boldsymbol{\tau}^{(1)}, \btheta^{(1)})- h_m(\alpha_l^{(2)}, \bar{\balpha}_l,  \boldsymbol{\tau}^{(1)}, \btheta^{(1)})\right|\nonumber\\
+ 2\| &\left(y^{\rm up}_m - h_m(\alpha_l^{(1)}, \bar{\balpha}_l,  \boldsymbol{\tau}^{(1)}, \btheta^{(1)})s_k\right){\alpha_l}^{(1)} \check{{\bv}}_{\alpha_l} - \left(y^{\rm up}_m - h_m(\alpha_l^{(2)}, \bar{\balpha}_l,  \boldsymbol{\tau}^{(1)}, \btheta^{(1)})s_k\right){\alpha_l}^{(2)} \check{{\bv}}_{\alpha_l}\| \nonumber\\
\overset{(c)}{\leq} 2\| & \bar{\bv}_{\alpha_l}s_k \| \left| h_m(\alpha_l^{(1)}, \bar{\balpha}_l,  \boldsymbol{\tau}^{(1)}, \btheta^{(1)})- h_m(\alpha_l^{(2)}, \bar{\balpha}_l,  \boldsymbol{\tau}^{(1)}, \btheta^{(1)})\right|\nonumber\\
+ 2 \| &y^{\rm up}_m \check{{\bv}}_{\alpha_l}\| \left|{\alpha_l}^{(1)} - {\alpha_l}^{(2)}\right| +2 \| \check{{\bv}}_{\alpha_l}s_k \| \left| h_m(\alpha_l^{(1)}, \bar{\balpha}_l,  \boldsymbol{\tau}^{(1)}, \btheta^{(1)}){\alpha_l}^{(1)} - h_m(\alpha_l^{(2)}, \bar{\balpha}_l,  \boldsymbol{\tau}^{(1)}, \btheta^{(1)}){\alpha_l}^{(2)} \right|\nonumber\\
\overset{(d)}{=} 2\| & \bar{\bv}_{\alpha_l}s_k \| \left| {\alpha_l}^{(1)} - {\alpha_l}^{(2)}\right|\nonumber\\
+ 2 \| &y^{\rm up}_m \check{{\bv}}_{\alpha_l}\| \left|{\alpha_l}^{(1)} - {\alpha_l}^{(2)}\right| +2 \| \check{{\bv}}_{\alpha_l}s_k \| \left| {\alpha_l}^{(1)} + {\alpha_l}^{(2)} \right| \left| {\alpha_l}^{(1)} - {\alpha_l}^{(2)} \right|\nonumber\\
= 2 ( \| &\bar{\bv}_{\alpha_l}s_k \|+ \| y^{\rm up}_m \check{{\bv}}_{\alpha_l}\| + \| \check{{\bv}}_{\alpha_l}s_k \| \left| {\alpha_l}^{(1)} + {\alpha_l}^{(2)} \right|)\left| {\alpha_l}^{(1)} - {\alpha_l}^{(2)} \right| \nonumber\\
\overset{(e)}{\triangleq}M_{\alpha_l} &\left| \alpha_l^{(1)} - \alpha_l^{(2)} \right|,
\end{alignat}
where the inequality in $(a)$ follows from Lemma \ref{norm_subvec_inequality}, ${\bv}_{\alpha_l}^{(i)}$, $i=1,2$, is a vector that captures those entries of $\triangledown h_m(\alpha_l^{(i)}, \bar{\balpha}_l,  \boldsymbol{\tau}^{(1)}, \btheta^{(1)})$ which are function of $\alpha_l^{(i)}$, and $\bar{\bv}_{\alpha_l}$ is a vector that captures those entries of $\triangledown h_m(\alpha_l^{(1)}, \bar{\balpha}_l,  \boldsymbol{\tau}^{(1)}, \btheta^{(1)})$ which are independent of $\alpha_l$. The second term in $(b)$ follows from ${\bv}_{\alpha_l}^{(i)} =  {\alpha_l}^{(i)} \check{{\bv}}_{\alpha_l}\triangleq {\alpha_l}^{(i)} e^{j\omega_l} \left[j,   j\beta, j\gamma_m \cos \theta_l\right]^T $. The inequality in $(c)$ follows from the fact that $\|\ba-\bb\|\leq \|\ba\|+\|\bb\|$ for $\ba$ and $\bb$ to be arbitrary vectors. In $(d)$, we use the facts that $\left| h_m(\alpha_l^{(1)}, \bar{\balpha}_l,  \boldsymbol{\tau}^{(1)}, \btheta^{(1)})- h_m(\alpha_l^{(2)}, \bar{\balpha}_l,  \boldsymbol{\tau}^{(1)}, \btheta^{(1)})\right| = \left|{\alpha_l}^{(1)} - {\alpha_l}^{(2)}\right|$ as well as  $\left| h_m(\alpha_l^{(1)}, \bar{\balpha}_l,  \boldsymbol{\tau}^{(1)}, \btheta^{(1)}){\alpha_l}^{(1)} - h_m(\alpha_l^{(2)}, \bar{\balpha}_l,  \boldsymbol{\tau}^{(1)}, \btheta^{(1)}){\alpha_l}^{(2)} \right| = \left|\left({\alpha_l}^{(1)} + {\alpha_l}^{(2)}\right) \left({\alpha_l}^{(1)} - {\alpha_l}^{(2)}\right)\right|$. Note that $M_{\alpha_l}$ in $(e)$ is bounded. This is mainly due to the fact that $\left\{\alpha_l\right\}_{l=1}^{L}$ are being generated from $G(\bz)$ which are bounded.
%
%
Therefore, $M_{\alpha_l}\geq 0$  is bounded. Now, from \eqref{T_alpha_upperbound},we obtain that,
\begin{align} \label{T_alpha_upperbound1}
T_{\alpha} \leq & M_{\alpha_1} \left| \alpha_1^{(1)} - \alpha_1^{(2)} \right|+ M_{\alpha_2} \left| \alpha_2^{(1)} - \alpha_2^{(2)} \right|+ \cdots+ M_{\alpha_L} \left| \alpha_L^{(1)} - \alpha_L^{(2)} \right|\nonumber\\
\overset{(a)}{\leq} & M_{\alpha} \left( \left| \alpha_1^{(1)} - \alpha_1^{(2)} \right|+ \left| \alpha_2^{(1)} - \alpha_2^{(2)} \right|+ \cdots+ \left| \alpha_L^{(1)} - \alpha_L^{(2)} \right|\right)\nonumber\\
\overset{(b)}{\leq} & M_{\alpha}\sqrt{L}\sqrt{\left( \alpha_1^{(1)} - \alpha_1^{(2)} \right)^2+ \left( \alpha_2^{(1)} - \alpha_2^{(2)} \right)^2+ \cdots+ \left( \alpha_L^{(1)} - \alpha_L^{(2)} \right)^2}\nonumber\\
=& \lambda_{\alpha}\| \balpha^{(1)} - \balpha^{(2)}\|,
\end{align}
where in $(a)$, $M_{\alpha}\triangleq \max\left\{M_{\alpha_1}, M_{\alpha_2}, \cdots, M_{\alpha_L}\right\}$, $(b)$ follows from the inequality between quadratic and arithmetic mean given in Lemma \ref{QM_AM_inequality}, and $\lambda_{\alpha}\triangleq M_{\alpha}\sqrt{L}$.

\subsubsection{Bounds on $T_{\tau}$}
Similar to what we did in previous subsection, we assume that $\balpha^{(2)}$ and $ \btheta^{(1)}$ are fixed. Using the triangle inequality of \eqref{fact_sum_upperbound}, we can write
\begin{align} \label{T_alpha_upperbound}
T_{\tau} & \leq \| f_m(\balpha^{(2)}, {\tau}_1^{(1)},  \bar{\boldsymbol{\tau}}_1, \btheta^{(1)}) - f_m(\balpha^{(2)}, {\tau}_1^{(2)},  \bar{\boldsymbol{\tau}}_1, \btheta^{(1)})\| \nonumber\\
& + \| f_m(\balpha^{(2)}, {\tau}_2^{(1)},  \bar{\boldsymbol{\tau}}_2, \btheta^{(1)}) - f_m(\balpha^{(2)}, {\tau}_2^{(2)},  \bar{\boldsymbol{\tau}}_2, \btheta^{(1)})\| \nonumber\\
& \qquad  \qquad \qquad \vdots
\nonumber\\
& + \| f_m(\balpha^{(2)}, {\tau}_L^{(1)},  \bar{\boldsymbol{\tau}}_L, \btheta^{(1)}) - f_m(\balpha^{(2)}, {\tau}_L^{(2)},  \bar{\boldsymbol{\tau}}_L, \btheta^{(1)})\|,
\end{align}
where $\bar{\boldsymbol{\tau}}_l\triangleq \left(\tau_1^{(2)}, \cdots, \tau_{l-1}^{(2)}, \tau_{l+1}^{(1)}, \cdots, \tau_L^{(1)}\right)$. Let us define $T_{\tau_l} \triangleq \| f_m(\balpha^{(2)}, {\tau}_l^{(1)},  \bar{\boldsymbol{\tau}}_l, \btheta^{(1)}) - f_m(\balpha^{(2)}, {\tau}_l^{(2)},  \bar{\boldsymbol{\tau}}_l, \btheta^{(1)})\| $. We now aim to show that there exist $M_{\tau_l}\geq 0 $ and bounded, such that $T_{\tau_l}\leq M_{\tau_l} \left| \tau_l^{(1)} - \tau_l^{(2)} \right|$. To do so, we can write
\begin{alignat}{2} \label{T_tau_l_upperbound}
T_{\tau_l} = 2\| &\left(y^{\rm up}_m - h_m(\balpha^{(2)}, {\tau}_l^{(1)},  \bar{\boldsymbol{\tau}}_l, \btheta^{(1)})s_k\right)\triangledown h_m(\balpha^{(2)}, {\tau}_l^{(1)},  \bar{\boldsymbol{\tau}}_l, \btheta^{(1)}) \nonumber\\
- &\left(y^{\rm up}_m - h_m(\balpha^{(2)}, {\tau}_l^{(2)},  \bar{\boldsymbol{\tau}}_l, \btheta^{(1)})s_k\right)\triangledown h_m(\balpha^{(2)}, {\tau}_l^{(2)},  \bar{\boldsymbol{\tau}}_l, \btheta^{(1)})\|\nonumber\\
\overset{(a)}{\leq} 2\| &\left(y^{\rm up}_m - h_m(\balpha^{(2)}, {\tau}_l^{(1)},  \bar{\boldsymbol{\tau}}_l, \btheta^{(1)})s_k\right)\bar{\bv}_{\tau_l} - \left(y^{\rm up}_m - h_m(\balpha^{(2)}, {\tau}_l^{(2)},  \bar{\boldsymbol{\tau}}_l, \btheta^{(1)})s_k\right)\bar{\bv}_{\tau_l}\| \nonumber\\
+2\| &\left(y^{\rm up}_m - h_m(\balpha^{(2)}, {\tau}_l^{(1)},  \bar{\boldsymbol{\tau}}_l, \btheta^{(1)})s_k\right)\bv_{\tau_l}^{(1)} - \left(y^{\rm up}_m - h_m(\balpha^{(2)}, {\tau}_l^{(2)},  \bar{\boldsymbol{\tau}}_l, \btheta^{(1)})s_k\right)\bv_{\tau_l}^{(2)}\| \nonumber\\
\overset{(b)}{=} 2\| & \bar{\bv}_{\tau_l}s_k \| \left| h_m(\balpha^{(2)}, {\tau}_l^{(1)},  \bar{\boldsymbol{\tau}}_l, \btheta^{(1)})- h_m(\balpha^{(2)}, {\tau}_l^{(2)},  \bar{\boldsymbol{\tau}}_l, \btheta^{(1)})\right|\nonumber\\
+ 2\| &\left(y^{\rm up}_m - h_m(\balpha^{(2)}, {\tau}_l^{(1)},  \bar{\boldsymbol{\tau}}_l, \btheta^{(1)})s_k\right)e^{j\beta \tau_l^{(1)}}\check{{\bv}}_{\tau_l} - \left(y^{\rm up}_m - h_m(\balpha^{(2)}, {\tau}_l^{(2)},  \bar{\boldsymbol{\tau}}_l, \btheta^{(1)})s_k\right)e^{j\beta \tau_l^{(2)}}\check{{\bv}}_{\tau_l}\| \nonumber\\
\overset{(c)}{\leq} 2\| & \bar{\bv}_{\tau_l}s_k \| \left| h_m(\balpha^{(2)}, {\tau}_l^{(1)},  \bar{\boldsymbol{\tau}}_l, \btheta^{(1)})- h_m(\balpha^{(2)}, {\tau}_l^{(2)},  \bar{\boldsymbol{\tau}}_l, \btheta^{(1)})\right|\nonumber\\
+ 2 \| &y^{\rm up}_m \check{{\bv}}_{\tau_l}\| | e^{j\beta \tau_l^{(1)}}- e^{j\beta \tau_l^{(2)}}| +2 \| \check{{\bv}}_{\tau_l}s_k \| \left| h_m(\balpha^{(2)}, {\tau}_l^{(1)},  \bar{\boldsymbol{\tau}}_l, \btheta^{(1)})e^{j\beta \tau_l^{(1)}} - h_m(\balpha^{(2)}, {\tau}_l^{(2)},  \bar{\boldsymbol{\tau}}_l, \btheta^{(1)})e^{j\beta \tau_l^{(2)}} \right|\nonumber\\
\overset{(d)}{\leq} 2\| & \bar{\bv}_{\tau_l}s_k \||{\alpha_l}^{(2)}| | e^{j\beta \tau_l^{(1)}}- e^{j\beta \tau_l^{(2)}}|\nonumber\\
+ 2 \| &y^{\rm up}_m \check{{\bv}}_{\tau_l}\| | e^{j\beta \tau_l^{(1)}}- e^{j\beta \tau_l^{(2)}}| +2 \| \check{{\bv}}_{\tau_l}s_k \|  \left(\left|\alpha_l^{(2)}\right| \left|e^{j2\beta \tau_l^{(1)}}- e^{j2\beta \tau_l^{(2)}}\right|
+\left|\tilde{\alpha}_l^{(2)}\right| \left|e^{j\beta \tau_l^{(1)}}- e^{j\beta \tau_l^{(2)}}\right|\right)\nonumber\\
= 2 \Big( \| &\bar{\bv}_{\tau_l}s_k \||{\alpha_l}^{(2)}| +\| y^{\rm up}_m \check{{\bv}}_{\tau_l}\| + \| \check{{\bv}}_{\tau_l}s_k \| \left|\tilde{\alpha}_l^{(2)}\right|\Big) \left|e^{j\beta \tau_l^{(1)}}- e^{j\beta \tau_l^{(2)}}\right|\nonumber\\
+ 2 \| &\check{{\bv}}_{\tau_l}s_k \| \left|{\alpha}_l^{(2)}\right| \left|e^{j2\beta \tau_l^{(1)}}- e^{j2\beta \tau_l^{(2)}}\right|\nonumber\\
 \overset{(e)}{\leq} 2 \left|\beta\right| &\Big( \| \bar{\bv}_{\tau_l}s_k \||{\alpha_l}^{(2)}| +\| y^{\rm up}_m \check{{\bv}}_{\tau_l}\| + \| \check{{\bv}}_{\tau_l}s_k \| \left|\tilde{\alpha}_l^{(2)}\right|\Big) \left| \tau_l^{(1)}- \tau_l^{(2)}\right|\nonumber\\
+ 4 \left|\beta\right| &\| \check{{\bv}}_{\tau_l}s_k \| \left|{\alpha}_l^{(2)}\right| \left| \tau_l^{(1)}- \tau_l^{(2)}\right|\nonumber\\
 \triangleq M_{\tau_l} &\left|\tau_l^{(1)}-\tau_l^{(2)}\right|,
\end{alignat}
where the inequality in $(a)$ follows from Lemma \ref{norm_subvec_inequality}, ${\bv}_{\tau_l}^{(i)}$, $i=1,2$, is a vector that captures those entries of $\triangledown h_m(\balpha^{(2)}, {\tau}_l^{(i)},  \bar{\boldsymbol{\tau}}_l, \btheta^{(1)})$ which are functions of $\tau_l^{(i)}$, and $\bar{\bv}_{\tau_l}$ is a vector that captures those entries of $\triangledown h_m(\balpha^{(2)}, {\tau}_l^{(i)},  \bar{\boldsymbol{\tau}}_l, \btheta^{(1)})$ which are independent of $\tau_l$. The second term in $(b)$ follows from ${\bv}_{\tau_l}^{(i)} =  e^{j\beta \tau_l^{(i)}}\check{{\bv}}_{\tau_l}\triangleq e^{j\beta \tau_l^{(i)}} e^{j(\phi_l+\gamma_m \sin \theta_l^{(1)})} \left[   1, j\beta\alpha^{(2)}, j\gamma_m \cos \theta_l^{(1)}, j\alpha^{(2)}\right]^T $. The inequality in $(c)$ follows from the fact that $\|\ba-\bb\|\leq \|\ba\|+\|\bb\|$ for $\ba$ and $\bb$ to be arbitrary vectors. In $(d)$, we use the following inequalities: \begin{align} \label{fact1_tau}
\left| h_m(\balpha^{(2)}, {\tau}_l^{(1)},  \bar{\boldsymbol{\tau}}_l, \btheta^{(1)})- h_m(\balpha^{(2)}, {\tau}_l^{(2)},  \bar{\boldsymbol{\tau}}_l, \btheta^{(1)})\right| &= \left|{\alpha_l}^{(2)}e^{j(\phi_l +\gamma_m \sin \theta_l^{(1)})}\left(e^{j\beta \tau_l^{(1)}}- e^{j\beta \tau_l^{(2)}}\right)\right| \nonumber\\
&\leq |{\alpha_l}^{(2)}| | e^{j\beta \tau_l^{(1)}}- e^{j\beta \tau_l^{(2)}}|,
\end{align}
as well as
\begin{align} \label{fact2_tau}
\left| h_m(\balpha^{(2)}, {\tau}_l^{(1)},  \bar{\boldsymbol{\tau}}_l, \btheta^{(1)})e^{j\beta \tau_l^{(1)}} - h_m(\balpha^{(2)}, {\tau}_l^{(2)},  \bar{\boldsymbol{\tau}}_l, \btheta^{(1)})e^{j\beta \tau_l^{(2)}} \right| &\leq  \left|\alpha_l^{(2)}\right| \left|e^{j2\beta \tau_l^{(1)}}- e^{j2\beta \tau_l^{(2)}}\right|\nonumber\\
&+\left|\tilde{\alpha}_l^{(2)}\right| \left|e^{j\beta \tau_l^{(1)}}- e^{j\beta \tau_l^{(2)}}\right|,
\end{align}
where $\tilde{\alpha}_l^{(2)} \triangleq \sum_{i\neq l}\alpha_i^{(2)}e^{j\left(\phi_i + \gamma_m \sin \theta_i^{(1)}\right)}$. The inequality $(e)$, assuming $\tau_l^{(1)}\leq \tau_l^{(2)}$, follows from the following inequality:
\begin{align} \label{fact3_tau}
\left| e^{j\beta \tau_l^{(1)}}- e^{j\beta \tau_l^{(2)}}\right| = \left|\int_{\tau_l^{(1)}}^{\tau_l^{(2)}} j\beta e^{j\beta t}dt\right|\leq \int_{\tau_l^{(1)}}^{\tau_l^{(2)}} \left|j\beta e^{j\beta t}\right|dt= \left|\beta\right|\left|\tau_l^{(1)}-\tau_l^{(2)}\right|.
\end{align}
As stated before, since $\left\{\alpha_l\right\}_{l=1}^L$ are generated by $G(\bz)$ and the range of $G(\bz)$ is bounded, we conclude that $M_{\tau_l} \geq 0$ and it is bounded. Now, defining, $M_{\tau}\triangleq \max\left\{M_{\tau_1}, M_{\tau_2}, \cdots, M_{\tau_L}\right\}$, we obtain the upperbound for $T_{\tau}$  as
\begin{align} \label{T_tau_upperbound1}
T_{\tau} \leq & M_{\tau_1} \left| \tau_1^{(1)} - \tau_1^{(2)} \right|+ M_{\tau_2} \left| \tau_2^{(1)} - \tau_2^{(2)} \right|+ \cdots+ M_{\tau_L} \left| \tau_L^{(1)} - \tau_L^{(2)} \right|\nonumber\\
{\leq} & M_{\tau} \left( \left| \tau_1^{(1)} - \tau_1^{(2)} \right|+ \left| \tau_2^{(1)} - \tau_2^{(2)} \right|+ \cdots+ \left| \tau_L^{(1)} - \tau_L^{(2)} \right|\right)\nonumber\\
\overset{(a)}{\leq} & M_{\tau}\sqrt{L}\sqrt{\left( \tau_1^{(1)} - \tau_1^{(2)} \right)^2+ \left( \tau_2^{(1)} - \tau_2^{(2)} \right)^2+ \cdots+ \left( \tau_L^{(1)} - \tau_L^{(2)} \right)^2}\nonumber\\
=& \lambda_{\tau}\| \boldsymbol{\tau}^{(1)} - \boldsymbol{\tau}^{(2)}\|,
\end{align}
where $(a)$ follows from the inequality provided in Lemma \ref{QM_AM_inequality}, and $\lambda_{\tau}\triangleq M_{\tau}\sqrt{L}$.
\subsubsection{Bounds on $T_{\theta}$}
Assuming $\balpha^{(2)}$ and $\boldsymbol{\tau}^{(2)}$ are fixed, and using the triangle inequality of \eqref{fact_sum_upperbound}, we can write
\begin{align} \label{T_theta_upperbound}
T_{\theta} & \leq \| f_m(\balpha^{(2)}, \boldsymbol{\tau}^{(2)}, \theta_1^{(1)}, \bar{\btheta}_1) - f_m(\balpha^{(2)}, \boldsymbol{\tau}^{(2)}, \theta_1^{(2)}, \bar{\btheta}_1)\| \nonumber\\
& + \| f_m(\balpha^{(2)}, \boldsymbol{\tau}^{(2)}, \theta_2^{(1)}, \bar{\btheta}_2) - f_m(\balpha^{(2)}, \boldsymbol{\tau}^{(2)}, \theta_2^{(2)}, \bar{\btheta}_2)\| \nonumber\\
& \qquad  \qquad \qquad \vdots
\nonumber\\
& + \| f_m(\balpha^{(2)}, \boldsymbol{\tau}^{(2)}, \theta_L^{(1)}, \bar{\btheta}_L) - f_m(\balpha^{(2)}, \boldsymbol{\tau}^{(2)}, \theta_L^{(2)}, \bar{\btheta}_L)\|,
\end{align}
where we define $\bar{\btheta}_l\triangleq \left(\theta_1^{(2)}, \cdots, \theta_{l-1}^{(2)}, \theta_{l+1}^{(1)}, \cdots, \theta_L^{(1)}\right)$.  Let us define $T_{\theta_l} \triangleq \| f_m(\balpha^{(2)}, \boldsymbol{\tau}^{(2)}, \theta_l^{(1)}, \bar{\btheta}_l) - f_m(\balpha^{(2)}, \boldsymbol{\tau}^{(2)}, \theta_l^{(2)}, \bar{\btheta}_l)\| $. We show that there exists a bounded $M_{\theta_l}\geq 0$  such that $T_{\theta_l}\leq M_{\theta_l} \left| \theta_l^{(1)} - \theta_l^{(2)} \right|$. To do so, let us write
\begin{align} \label{T_theta_l_upperbound}
T_{\theta_l} = 2\| &\left(y^{\rm up}_m - h_m(\balpha^{(2)}, \boldsymbol{\tau}^{(2)}, \theta_l^{(1)}, \bar{\btheta}_l)s_k\right)\triangledown h_m(\balpha^{(2)}, \boldsymbol{\tau}^{(2)}, \theta_l^{(1)}, \bar{\btheta}_l) \nonumber\\
- &\left(y^{\rm up}_m - h_m(\balpha^{(2)}, \boldsymbol{\tau}^{(2)}, \theta_l^{(2)}, \bar{\btheta}_l)s_k\right)\triangledown h_m(\balpha^{(2)}, \boldsymbol{\tau}^{(2)}, \theta_l^{(2)}, \bar{\btheta}_l)\|\nonumber\\
\overset{(a)}{\leq} 2\| &\left(y^{\rm up}_m - h_m(\balpha^{(2)}, \boldsymbol{\tau}^{(2)}, \theta_l^{(1)}, \bar{\btheta}_l)s_k\right)\bar{\bv}_{\theta_l} - \left(y^{\rm up}_m - h_m(\balpha^{(2)}, \boldsymbol{\tau}^{(2)}, \theta_l^{(2)}, \bar{\btheta}_l)s_k\right)\bar{\bv}_{\theta_l}\| \nonumber\\
+2\| &\left(y^{\rm up}_m - h_m(\balpha^{(2)}, \boldsymbol{\tau}^{(2)}, \theta_l^{(1)}, \bar{\btheta}_l)s_k\right)\bv_{\theta_l}^{(1)} - \left(y^{\rm up}_m - h_m(\balpha^{(2)}, \boldsymbol{\tau}^{(2)}, \theta_l^{(2)}, \bar{\btheta}_l)s_k\right)\bv_{\theta_l}^{(2)}\| \nonumber\\
\overset{(b)}{\leq} 2\| & \bar{\bv}_{\theta_l}s_k \| \left| h_m(\balpha^{(2)}, \boldsymbol{\tau}^{(2)}, \theta_l^{(1)}, \bar{\btheta}_l)- h_m(\balpha^{(2)}, \boldsymbol{\tau}^{(2)}, \theta_l^{(2)}, \bar{\btheta}_l)\right|\nonumber\\
+ 2\| &y^{\rm up}_m \left( \bv_{\theta_l}^{(1)}- \bv_{\theta_l}^{(2)}\right)\| + 2|s_k|\|  h_m(\balpha^{(2)}, \boldsymbol{\tau}^{(2)}, \theta_l^{(1)}, \bar{\btheta}_l) \bv_{\theta_l}^{(1)}  -  h_m(\balpha^{(2)}, \boldsymbol{\tau}^{(2)}, \theta_l^{(2)}, \bar{\btheta}_l)\bv_{\theta_l}^{(2)}\| \nonumber\\
\overset{(c)}{\leq} 2\| & \bar{\bv}_{\theta_l}s_k \| \left|\alpha_l^{(2)}\right|\left|e^{j\gamma_m \sin \theta_l^{(1)}}- e^{j\gamma_m \sin \theta_l^{(2)}}\right| + 2| y^{\rm up}_m | \|\bv_{\theta_l}^{(1)}- \bv_{\theta_l}^{(2)}\| \nonumber\\
+2 \| &\check{{\bv}}_{\theta_l}s_k \| \Big( \left|\breve{\alpha}_l^{(2)}\right| \| \bv_{\theta_l}^{(1)} - \bv_{\theta_l}^{(2)} \| +  \| \check{\bv}_{\theta_l}  \| \left| e^{j\gamma_m \sin \theta_l^{(1)}}- e^{j\gamma_m \sin \theta_l^{(2)}}\right|\nonumber\\
&+  \left|\gamma_m {\alpha_l^{(2)}}^2\right| \left| \cos \theta_l^{(1)} e^{j\gamma_m \sin \theta_l^{(1)}}-  \cos \theta_l^{(2)} e^{j\gamma_m \sin \theta_l^{(2)}}\right|  \Big)\nonumber\\
{=} 2 \Big(\|  &\bar{\bv}_{\theta_l}s_k \| \left|\alpha_l^{(2)}\right| + \| \check{{\bv}}_{\theta_l}s_k \|  \| \check{\bv}_{\theta_l}  \|\Big)\left|e^{j\gamma_m \sin \theta_l^{(1)}}- e^{j\gamma_m \sin \theta_l^{(2)}}\right| \nonumber\\
+ 2&\left(| y^{\rm up}_m | +\| \check{{\bv}}_{\theta_l}s_k \|\left|\breve{\alpha}_l^{(2)}\right|\right)\|\bv_{\theta_l}^{(1)}- \bv_{\theta_l}^{(2)}\| \nonumber\\
+2 &\left(\| \check{{\bv}}_{\theta_l}s_k \|
 \left|\gamma_m {\alpha_l^{(2)}}^2\right|\right) \left| \cos \theta_l^{(1)} e^{j\gamma_m \sin \theta_l^{(1)}}-  \cos \theta_l^{(2)} e^{j\gamma_m \sin \theta_l^{(2)}}\right|  \nonumber\\
\overset{(d)}{\leq} 2 \Big(\|  &\bar{\bv}_{\theta_l}s_k \| \left|\alpha_l^{(2)}\right| + \| \check{{\bv}}_{\theta_l}s_k \|  \| \check{\bv}_{\theta_l}  \| + \left(| y^{\rm up}_m | +\| \check{{\bv}}_{\theta_l}s_k \|\left|\breve{\alpha}_l^{(2)}\right|\right) \left(1+\left|\beta \alpha_l^{(2)} \right|\right)\Big)\left|e^{j\gamma_m \sin \theta_l^{(1)}}- e^{j\gamma_m \sin \theta_l^{(2)}}\right| \nonumber\\
+2 &\left(\| \check{{\bv}}_{\theta_l}s_k \|
 \left|\gamma_m {\alpha_l^{(2)}}^2\right|+ \left(| y^{\rm up}_m | +\| \check{{\bv}}_{\theta_l}s_k \|\left|\breve{\alpha}_l^{(2)}\right|\right) \left|\gamma_m \alpha_l^{(2)} \right|\right) \left| \cos \theta_l^{(1)} e^{j\gamma_m \sin \theta_l^{(1)}}-  \cos \theta_l^{(2)} e^{j\gamma_m \sin \theta_l^{(2)}}\right|  \nonumber\\
\overset{(e)}{\leq} 2 \Big(\|  &\bar{\bv}_{\theta_l}s_k \| \left|\alpha_l^{(2)}\right| + \| \check{{\bv}}_{\theta_l}s_k \|  \| \check{\bv}_{\theta_l}  \| + \left(| y^{\rm up}_m | +\| \check{{\bv}}_{\theta_l}s_k \|\left|\breve{\alpha}_l^{(2)}\right|\right) \left(1+\left|\beta \alpha_l^{(2)} \right|\right)\Big)\left|\gamma_m \right|\left|\theta_l^{(1)}- \theta_l^{(2)}\right| \nonumber\\
+2 &\left(\| \check{{\bv}}_{\theta_l}s_k \|
 \left|\gamma_m {\alpha_l^{(2)}}^2\right|+ \left(| y^{\rm up}_m | +\| \check{{\bv}}_{\theta_l}s_k \|\left|\breve{\alpha}_l^{(2)}\right|\right) \left|\gamma_m \alpha_l^{(2)} \right|\right) \rho\left|\theta_l^{(1)}- \theta_l^{(2)}\right| \nonumber\\
\triangleq M_{\theta_l} &\left|\theta_l^{(1)}- \theta_l^{(2)}\right|,
\end{align}
 where the inequality in $(a)$ follows from Lemma \ref{norm_subvec_inequality}, ${\bv}_{\theta_l}^{(i)}$, $i=1,2$, is a vector that captures those entries of $\triangledown h_m(\balpha^{(2)}, \boldsymbol{\tau}^{(2)}, \theta_l^{(i)}, \bar{\btheta}_l)$ which are function of $\theta_l^{(i)}$, and $\bar{\bv}_{\theta_l}$ is a vector that captures those entries of $\triangledown h_m(\balpha^{(2)}, \boldsymbol{\tau}^{(2)}, \theta_l^{(1)}, \bar{\btheta}_l)$ which are independent of $\theta_l$. In the second term in $(b)$, we define ${\bv}_{\theta_l}^{(i)} \triangleq e^{j(\phi_l+\beta \tau_l^{(2)}+\gamma_m \sin \theta_l^{(i)})} \left[   1\; \; j\beta\alpha_l^{(2)}\;\; j\alpha_l^{(2)}\gamma_m \cos \theta_l^{(i)}\;\; j\alpha_l^{(2)}\right]^T $, and the inequality follows from the fact that $\|\ba-\bb\|\leq \|\ba\|+\|\bb\|$ for $\ba$ and $\bb$ to be arbitrary vectors. In $(c)$, we use the following inequalities:\begin{align} \label{fact1_theta}
\left| h_m(\balpha^{(2)}, \boldsymbol{\tau}^{(2)}, \theta_l^{(1)}, \bar{\btheta}_l)- h_m(\balpha^{(2)}, \boldsymbol{\tau}^{(2)}, \theta_l^{(2)}, \bar{\btheta}_l)\right| &= \left|{\alpha_l}^{(2)}e^{j(\phi_l +\beta \tau_l^{(2)})}\left(e^{j\gamma_m \sin \theta_l^{(1)}}- e^{j\gamma_m \sin \theta_l^{(2)}}\right)\right| \nonumber\\
&\leq \left|{\alpha_l}^{(2)}\right| \left| e^{j\gamma_m \sin \theta_l^{(1)}}- e^{j\gamma_m \sin \theta_l^{(2)}}\right|,
\end{align}
as well as
\begin{align} \label{fact2_theta}
&\|  h_m(\balpha^{(2)}, \boldsymbol{\tau}^{(2)}, \theta_l^{(1)}, \bar{\btheta}_l) \bv_{\theta_l}^{(1)}  -  h_m(\balpha^{(2)}, \boldsymbol{\tau}^{(2)}, \theta_l^{(2)}, \bar{\btheta}_l)\bv_{\theta_l}^{(2)}\| \nonumber\\
&\leq \left|\breve{\alpha}_l^{(2)}\right| \| \bv_{\theta_l}^{(1)} - \bv_{\theta_l}^{(2)} \| +  \| \check{\bv}_{\theta_l}  \| \left| e^{j\gamma_m \sin \theta_l^{(1)}}- e^{j\gamma_m \sin \theta_l^{(2)}}\right|\nonumber\\
&+  \left|\gamma_m {\alpha_l^{(2)}}^2\right| \left| \cos \theta_l^{(1)} e^{j\gamma_m \sin \theta_l^{(1)}}-  \cos \theta_l^{(2)} e^{j\gamma_m \sin \theta_l^{(2)}}\right|,
\end{align}
where $\breve{\alpha}_l^{(2)} \triangleq \sum_{i \neq l}\alpha_i^{(2)}e^{j(\phi_i +\beta \tau_i^{(2)} +\gamma_m\sin \theta_i^{(t)})}$, $t=2$ for $i<l$ and $t=1$, for $i>l$, and $\check{\bv}_{\theta_l} \triangleq \left[ \alpha_l^{(2)}e^{j(2\phi_l + 2\beta \tau_l^{(2)})} ,  j\beta {\alpha_l^{(2)}}^2e^{j(2\phi_l + 2\beta \tau_l^{(2)})}\right]^T$. In $(d)$, following from Lemma \ref{norm_subvec_inequality}, we use the the following inequality:
\begin{align} \label{fact3_theta}
\|\bv_{\theta_l}^{(1)}- \bv_{\theta_l}^{(2)}\| &\leq \left|e^{j\gamma_m \sin \theta_l^{(1)}}- e^{j\gamma_m \sin \theta_l^{(2)}}\right| + \left|\beta \alpha_l^{(2)} \right|  \left|e^{j\gamma_m \sin \theta_l^{(1)}}- e^{j\gamma_m \sin \theta_l^{(2)}}\right|\nonumber\\
&+ \left|\gamma_m \alpha_l^{(2)} \right|  \left|\cos \theta_l^{(1)} e^{j\gamma_m \sin \theta_l^{(1)}}- \cos \theta_l^{(2)} e^{j\gamma_m \sin \theta_l^{(2)}}\right|\nonumber\\
&=\left(1+\left|\beta \alpha_l^{(2)} \right|\right) \left|e^{j\gamma_m \sin \theta_l^{(1)}}- e^{j\gamma_m \sin \theta_l^{(2)}}\right|\nonumber\\
&+ \left|\gamma_m \alpha_l^{(2)} \right|  \left|\cos \theta_l^{(1)} e^{j\gamma_m \sin \theta_l^{(1)}}- \cos \theta_l^{(2)} e^{j\gamma_m \sin \theta_l^{(2)}}\right|.
\end{align}
The inequality $(e)$ follows from
\begin{align} \label{fact4_theta}
\left|e^{j\gamma_m \sin \theta_l^{(1)}}- e^{j\gamma_m \sin \theta_l^{(2)}}\right| &= \left|\int_{\theta_l^{(1)}}^{\theta_l^{(2)}} j\gamma_m \cos \theta e^{j\gamma_m \sin \theta}d\theta\right|\leq \int_{\theta_l^{(1)}}^{\theta_l^{(2)}} \left|j\gamma_m \cos \theta e^{j\gamma_m \sin \theta}\right|d\theta\nonumber\\
&= \left|\gamma_m \right|\int_{\theta_l^{(1)}}^{\theta_l^{(2)}} \left|\cos \theta \right|d\theta \leq \left|\gamma_m \right|\int_{\theta_l^{(1)}}^{\theta_l^{(2)}} 1d\theta =\left|\gamma_m \right|\left|\theta_l^{(1)}- \theta_l^{(2)}\right|,
\end{align}
and
\begin{align} \label{fact5_theta}
\left|\cos \theta_l^{(1)} e^{j\gamma_m \sin \theta_l^{(1)}}- \cos \theta_l^{(2)} e^{j\gamma_m \sin \theta_l^{(2)}}\right| &= \left|\int_{\theta_l^{(1)}}^{\theta_l^{(2)}} \left(j\gamma_m \cos^2 \theta - \sin \theta \right)e^{j\gamma_m \sin \theta}d\theta\right|\nonumber\\
&\leq \int_{\theta_l^{(1)}}^{\theta_l^{(2)}} \left|\left(j\gamma_m \cos^2 \theta - \sin \theta \right)e^{j\gamma_m \sin \theta} \right|d\theta\nonumber\\
&= \int_{\theta_l^{(1)}}^{\theta_l^{(2)}} \left|j\gamma_m \cos^2 \theta - \sin \theta \right|d\theta\nonumber\\
&\leq \int_{\theta_l^{(1)}}^{\theta_l^{(2)}} \rho d\theta = \rho \left| \theta_l^{(1)} - \theta_l^{(2)}\right|,
\end{align}
where $\rho \triangleq \max_{\theta} \left|j\gamma_m \cos^2 \theta - \sin \theta \right|$.
Note that since $\left\{\alpha_l\right\}_{l=1}^L$ are bounded, $M_{\theta_l} \geq 0$  is bounded. Now, defining $M_{\theta}\triangleq \max\left\{M_{\theta_1}, M_{\theta_2}, \cdots, M_{\theta_L}\right\}$, we obtain the upper bound for $T_{\theta}$ as
\begin{align} \label{T_tau_upperbound1}
T_{\theta} \leq & M_{\theta_1} \left| \theta_1^{(1)} - \theta_1^{(2)} \right|+ M_{\theta_2} \left| \theta_2^{(1)} - \theta_2^{(2)} \right|+ \cdots+ M_{\theta_L} \left| \theta_L^{(1)} - \theta_L^{(2)} \right|\nonumber\\
{\leq} & M_{\theta} \left( \left| \theta_1^{(1)} - \theta_1^{(2)} \right|+ \left| \tau_2^{(1)} - \tau_2^{(2)} \right|+ \cdots+ \left| \theta_L^{(1)} - \theta_L^{(2)} \right|\right)\nonumber\\
\overset{(a)}{\leq} & M_{\theta}\sqrt{L}\sqrt{\left( \theta_1^{(1)} - \theta_1^{(2)} \right)^2+ \left( \theta_2^{(1)} - \theta_2^{(2)} \right)^2+ \cdots+ \left( \theta_L^{(1)} - \theta_L^{(2)} \right)^2}\nonumber\\
=& \lambda_{\theta}\| \btheta^{(1)} - \btheta^{(2)}\|,
\end{align}
where $(a)$ follows from the inequality provided in Lemma \ref{QM_AM_inequality}, and $\lambda_{\theta}\triangleq M_{\theta}\sqrt{L}$.

Returning back to \eqref{fm_lip_def_triangle}, we can write
\begin{align} \label{fm_lip_def_triangle1}
\| f_m(\bx^{(1)}) - f_m(\bx^{(2)}) \| & \leq  \lambda_{\alpha}\| \balpha^{(1)} - \balpha^{(2)}\| +\lambda_{\tau}\| \boldsymbol{\tau}^{(1)} - \boldsymbol{\tau}^{(2)}\| + \lambda_{\theta}\| \btheta^{(1)} - \btheta^{(2)}\|\nonumber\\
& \leq  \lambda_{\alpha}\| \bx^{(1)} - \bx^{(2)}\| +\lambda_{\tau}\| \bx^{(1)} - \bx^{(2)}\| + \lambda_{\theta}\| \bx^{(1)} - \bx^{(2)}\| \nonumber\\
&=  \left(\lambda_{\alpha} + \lambda_{\tau}+ \lambda_{\theta} \right)\| \bx^{(1)} - \bx^{(2)}\| \triangleq  \lambda_m \| \bx^{(1)} - \bx^{(2)}\|.
\end{align}
 This implies that $f_m(\bx)$ is a Lipschitz-continuous function with respect to $\bx$. Now, according to Lemma \ref{lip_sum}, and defining $\lambda_f \triangleq \max \left\{ \lambda_1, \lambda_2, \cdots, \lambda_M\right\}$,  we conclude that $f(\bx) = \sum_{m=1}^{M} f_m(\bx)$ is a Lipschitz-continuous function, i.e.,
\begin{align} \label{f_lip}
\| f(\bx^{(1)}) - f(\bx^{(2)}) \| \leq \lambda_f \| \bx^{(1)} - \bx^{(2)}\|.
\end{align}

Note that the generator function $G_{{\cal W}_g}(\cdot)$ is approximated by a multi-layer neural network. It is shown that $G_{{\cal W}_g}(\bz)$ is a Lipschitz function \cite{10.5555/3327144.3327299, Gouk2021}. Let $\lambda_g$ denote the corresponding the Lipschitz constant. Exact estimation of $\lambda_g$ is beyond the scope of this paper. We now use Lemma \ref{lip_composition}, and conclude that $f(G_{{\cal W}_g}(\bz))$ is a $\lambda_f \lambda_g$-Lipschitz function. Therefore, we can write
\begin{align} \label{I_z_bounded_rewrite}
I_{\bz} \leq \lambda_f \lambda_g \| \bz^{(1)} - \bz^{(2)}\|.
\end{align}

\subsection{Bounds on $I_{\boldsymbol{\phi}}$}
Assuming ${\bz}^{(2)}$ is fixed, let $q({\boldsymbol{\phi}})\triangleq P({\bz}^{(2)}, {\boldsymbol{\phi}})$, we aim to show that there exists $\lambda_q$ such that
\begin{align} \label{I_phi_bounded}
I_{\boldsymbol{\phi}}=\| q({\boldsymbol{\phi}}^{(1)}) - q({\boldsymbol{\phi}}^{(2)})\| \leq  \lambda_q \| {\boldsymbol{\phi}}^{(1)} - {\boldsymbol{\phi}}^{(2)}\|.
\end{align}

Note that, for fixed $\bz^{(2)}$, or equivalently $\bx^{(2)}= G_{{\cal W}_g}\left(\bz^{(2)}\right)$, we can express $q({\boldsymbol{\phi}})= \triangledown_{\tilde{\bx}}J_{\rm up}\left(\bx^{(2)}, {\boldsymbol{\phi}}\right)$. We further express $q({\boldsymbol{\phi}}) = \sum_{m=1}^{M} q_m({\boldsymbol{\phi}})$, where $q_m({\boldsymbol{\phi}})$, for $\bx^{(2)}$, is given as
\begin{align} \label{q_m_def}
q_m({\boldsymbol{\phi}})\triangleq -2\left(y^{\rm up}_m - h_m(\bx^{(2)}, {\boldsymbol{\phi}})s_k\right)\triangledown_{\tilde{\bx}}h_m(\bx^{(2)}, {\boldsymbol{\phi}}).
\end{align}
We now show the Lipschitz-continuity of $q_m({\boldsymbol{\phi}})$. Using the triangle inequality of \eqref{fact_sum_upperbound}, we can write
\begin{align} \label{T_phi_upperbound}
\| q_m({\boldsymbol{\phi}}^{(1)}) - q_m({\boldsymbol{\phi}}^{(2)})\| & \leq \| q_m({\phi}_1^{(1)}, {\boldsymbol{\bar{\phi}}}_1) - q_m({\phi}_1^{(2)}, {\boldsymbol{\bar{\phi}}}_1)\| \nonumber\\
& + \| q_m({\phi}_2^{(1)}, {\boldsymbol{\bar{\phi}}}_2) - q_m({\phi}_2^{(2)}, {\boldsymbol{\bar{\phi}}}_2)\| \nonumber\\
& \qquad  \qquad \qquad \vdots
\nonumber\\
& + \| q_m({\phi}_L^{(1)}, {\boldsymbol{\bar{\phi}}}_L) - q_m({\phi}_L^{(2)}, {\boldsymbol{\bar{\phi}}}_L)\|,
\end{align}
where ${\boldsymbol{\bar{\phi}}}_l\triangleq \left({\phi}_1^{(2)}, \cdots, {\phi}_{l-1}^{(2)}, {\phi}_{l+1}^{(1)}, \cdots, {\phi}_L^{(1)}\right)$. Let $T_{\phi_l} \triangleq  \| q_m({\phi}_l^{(1)}, {\boldsymbol{\bar{\phi}}}_l) - q_m({\phi}_l^{(2)}, {\boldsymbol{\bar{\phi}}}_l)\| $. We show that there exists a bounded $M_{\phi_l}\geq 0 $ such that $T_{\phi_l}\leq M_{\phi_l} \left| \phi_l^{(1)} - \phi_l^{(2)} \right|$. To do so, let us write
\begin{alignat}{2} \label{T_phi_l_upperbound}
T_{\phi_l} = 2\| &\left(y^{\rm up}_m - h_m(\bx^{(2)}, {\phi}_l^{(1)}, {\boldsymbol{\bar{\phi}}}_l)s_k\right)\triangledown h_m(\bx^{(2)}, {\phi}_l^{(1)}, {\boldsymbol{\bar{\phi}}}_l) \nonumber\\
- &\left(y^{\rm up}_m - h_m(\bx^{(2)}, {\phi}_l^{(2)}, {\boldsymbol{\bar{\phi}}}_l)s_k\right)\triangledown h_m(\bx^{(2)}, {\phi}_l^{(2)}, {\boldsymbol{\bar{\phi}}}_l)\|\nonumber\\
\overset{(a)}{\leq} 2\| &\left(y^{\rm up}_m - h_m(\bx^{(2)}, {\phi}_l^{(1)}, {\boldsymbol{\bar{\phi}}}_l)s_k\right)\bar{\bv}_{\phi_l} - \left(y^{\rm up}_m - h_m(\bx^{(2)}, {\phi}_l^{(2)}, {\boldsymbol{\bar{\phi}}}_l)s_k\right)\bar{\bv}_{\phi_l}\| \nonumber\\
+2\| &\left(y^{\rm up}_m - h_m(\bx^{(2)}, {\phi}_l^{(1)}, {\boldsymbol{\bar{\phi}}}_l)s_k\right)\bv_{\phi_l}^{(1)} - \left(y^{\rm up}_m - h_m(\bx^{(2)}, {\phi}_l^{(2)}, {\boldsymbol{\bar{\phi}}}_l)s_k\right)\bv_{\phi_l}^{(2)}\| \nonumber\\
\overset{(b)}{=} 2\| & \bar{\bv}_{\phi_l}s_k \| \left| h_m(\bx^{(2)}, {\phi}_l^{(1)}, {\boldsymbol{\bar{\phi}}}_l)- h_m(\bx^{(2)}, {\phi}_l^{(2)}, {\boldsymbol{\bar{\phi}}}_l)\right|\nonumber\\
+ 2\| &\left(y^{\rm up}_m - h_m(\bx^{(2)}, {\phi}_l^{(1)}, {\boldsymbol{\bar{\phi}}}_l)s_k\right)e^{j \phi_l^{(1)}}\check{{\bv}}_{\phi_l} - \left(y^{\rm up}_m - h_m(\bx^{(2)}, {\phi}_l^{(2)}, {\boldsymbol{\bar{\phi}}}_l)s_k\right)e^{j \phi_l^{(2)}}\check{{\bv}}_{\phi_l}\| \nonumber\\
\overset{(c)}{\leq} 2\| & \bar{\bv}_{\phi_l}s_k \| \left| h_m(\bx^{(2)}, {\phi}_l^{(1)}, {\boldsymbol{\bar{\phi}}}_l)- h_m(\bx^{(2)}, {\phi}_l^{(2)}, {\boldsymbol{\bar{\phi}}}_l)\right|\nonumber\\
+ 2 \| &y^{\rm up}_m \check{{\bv}}_{\phi_l}\| | e^{j \phi_l^{(1)}}- e^{j\phi_l^{(2)}}| +2 \| \check{{\bv}}_{\phi_l}s_k \| \left| h_m(\bx^{(2)}, {\phi}_l^{(1)}, {\boldsymbol{\bar{\phi}}}_l)e^{j \phi_l^{(1)}} - h_m(\bx^{(2)}, {\phi}_l^{(2)}, {\boldsymbol{\bar{\phi}}}_l)e^{j\phi_l^{(2)}} \right|\nonumber\\
\overset{(d)}{\leq} 2\| & \bar{\bv}_{\phi_l}s_k \||{\alpha_l}^{(2)}| | e^{j\phi_l^{(1)}}- e^{j\phi_l^{(2)}}|\nonumber\\
+ 2 \| &y^{\rm up}_m \check{{\bv}}_{\phi_l}\| | e^{j\phi_l^{(1)}}- e^{j \phi_l^{(2)}}| +2 \| \check{{\bv}}_{\phi_l}s_k \|  \left(\left|\breve{\alpha}_l^{(2)}\right|\left| e^{j\phi_l^{(1)}}- e^{j\phi_l^{(2)}}\right| +  \left|\check{\alpha}_l^{(2)}\right|\left| e^{j2\phi_l^{(1)}}- e^{j2\phi_l^{(2)}}\right|\right)\nonumber\\
= 2 \Big( \| &\bar{\bv}_{\phi_l}s_k \||{\alpha_l}^{(2)}| +\| y^{\rm up}_m \check{{\bv}}_{\phi_l}\| + \| \check{{\bv}}_{\phi_l}s_k \| \left|\breve{\alpha}_l^{(2)}\right|\Big) \left|e^{j\phi_l^{(1)}}- e^{j\phi_l^{(2)}}\right|\nonumber\\
+ 2 \| &\check{{\bv}}_{\phi_l}s_k \| \left|{\breve{\alpha}}_l^{(2)}\right| \left|e^{j2\phi_l^{(1)}}- e^{j2\phi_l^{(2)}}\right|\nonumber\\
 \overset{(e)}{\leq} 2 \Big( \| &\bar{\bv}_{\phi_l}s_k \||{\alpha_l}^{(2)}| +\| y^{\rm up}_m \check{{\bv}}_{\phi_l}\| + \| \check{{\bv}}_{\phi_l}s_k \| \left|\breve{\alpha}_l^{(2)}\right|\Big) \left|{\phi_l^{(1)}}- {\phi_l^{(2)}}\right| + 4 \| \check{{\bv}}_{\phi_l}s_k \| \left|{\breve{\alpha}}_l^{(2)}\right| \left|{\phi_l^{(1)}}- {\phi_l^{(2)}}\right|\nonumber\\
 \triangleq M_{\phi_l} &\left|{\phi_l^{(1)}}- {\phi_l^{(2)}}\right|,
\end{alignat}
where the inequality in $(a)$ follows from Lemma \ref{norm_subvec_inequality}, ${\bv}_{\phi_l}^{(i)}$, $i=1,2$, is a vector that captures those entries of $\triangledown h_m(\bx^{(2)}, {\phi}_l^{(i)}, {\boldsymbol{\bar{\phi}}}_l)$ which are function of $\phi_l^{(i)}$, and $\bar{\bv}_{\phi_l}$ is a vector that captures those entries of $\triangledown h_m(\bx^{(2)}, {\phi}_l^{(i)}, {\boldsymbol{\bar{\phi}}}_l)$ which are independent of $\phi_l$. The second term in $(b)$ follows from ${\bv}_{\phi_l}^{(i)} =  e^{j\phi_l^{(i)}}\check{{\bv}}_{\phi_l}\triangleq e^{j\phi_l^{(i)}} e^{j(\beta\tau_l^{(2)}+\gamma_m \sin \theta_l^{(2)})} \left[   1\;\; j\beta\alpha^{(2)}\;\; j\gamma_m \cos \theta_l^{(2)}\;\; j\alpha^{(2)}\right]^T $. The inequality in $(c)$ follows from the fact that $\|\ba-\bb\|\leq \|\ba\|+\|\bb\|$ for $\ba$ and $\bb$ to be arbitrary vectors. In $(d)$, we use the following inequalities:\begin{align} \label{fact1_phi}
\left| h_m(\bx^{(2)}, {\phi}_l^{(1)}, {\boldsymbol{\bar{\phi}}}_l)- h_m(\bx^{(2)}, {\phi}_l^{(2)}, {\boldsymbol{\bar{\phi}}}_l)\right| &= \left|{\alpha_l}^{(2)}e^{j(\beta \tau_l^{(2)}+\gamma_m \sin \theta_l^{(2)})}\left(e^{j\phi_l^{(1)}}- e^{j\phi_l^{(2)}}\right)\right| \nonumber\\
&\leq \left|{\alpha_l}^{(2)}\right| \left| e^{j\phi_l^{(1)}}- e^{j\phi_l^{(2)}}\right|,
\end{align}
\begin{align} \label{fact2_phi}
&\left| h_m(\bx^{(2)}, {\phi}_l^{(1)}, {\boldsymbol{\bar{\phi}}}_l)e^{j \phi_l^{(1)}} - h_m(\bx^{(2)}, {\phi}_l^{(2)}, {\boldsymbol{\bar{\phi}}}_l)e^{j\phi_l^{(2)}} \right| \leq \left|\breve{\alpha}_l^{(2)}\right|\left| e^{j\phi_l^{(1)}}- e^{j\phi_l^{(2)}}\right| +  \left|\check{\alpha}_l^{(2)}\right|\left| e^{j2\phi_l^{(1)}}- e^{j2\phi_l^{(2)}}\right|,
\end{align}
where $\breve{\alpha}_l^{(2)} \triangleq \sum_{i \neq l}\alpha_i^{(2)}e^{j(\phi_i^{(t)} +\beta \tau_i^{(2)} +\gamma_m\sin \theta_i^{(2)})}$, $t=2$ for $i<l$ and $t=1$, for $i>l$, and $\check{\alpha}_l^{(2)} \triangleq  \alpha_l^{(2)}e^{j( \beta \tau_l^{(2)} +\gamma_m\sin \theta_l^{(2)})} $. The inequality $(e)$ follows from \eqref{fact3_tau}.

Since $\left\{\alpha_l\right\}_{l=1}^L$ are bounded, we conclude that $M_{\phi_l} \geq 0$ and it is bounded. Now, defining, $M_{\phi}\triangleq \max\left\{M_{\phi_1}, M_{\phi_2}, \cdots, M_{\phi_L}\right\}$, we can write
\begin{align} \label{T_tau_upperbound1}
\| q_m({\boldsymbol{\phi}}^{(1)}) - q_m({\boldsymbol{\phi}}^{(2)})\| \leq & M_{\phi_1} \left| \phi_1^{(1)} - \phi_1^{(2)} \right|+ M_{\phi_2} \left| \phi_2^{(1)} - \phi_2^{(2)} \right|+ \cdots+ M_{\phi_L} \left| \phi_L^{(1)} - \phi_L^{(2)} \right|\nonumber\\
{\leq} & M_{\phi} \left( \left| \phi_1^{(1)} - \phi_1^{(2)} \right|+ \left| \phi_2^{(1)} - \phi_2^{(2)} \right|+ \cdots+ \left| \theta_L^{(1)} - \theta_L^{(2)} \right|\right)\nonumber\\
\overset{(a)}{\leq} & M_{\phi}\sqrt{L}\sqrt{\left( \phi_1^{(1)} - \phi_1^{(2)} \right)^2+ \left( \phi_2^{(1)} - \phi_2^{(2)} \right)^2+ \cdots+ \left( \phi_L^{(1)} - \phi_L^{(2)} \right)^2}\nonumber\\
=& \kappa_m\| {\boldsymbol{\phi}}^{(1)} - {\boldsymbol{\phi}}^{(2)}\|,
\end{align}
where $(a)$ follows from the inequality provided in Lemma \ref{QM_AM_inequality}, and $ \kappa_m \triangleq M_{\phi}\sqrt{L}$. Now, according to Lemma \ref{lip_sum}, and defining $\lambda_q \triangleq \max \left\{  \kappa_1,  \kappa_2 \cdots,  \kappa_M\right\}$,  we conclude that $q({\boldsymbol{\phi}}) = \sum_{m=1}^{M} q_m({\boldsymbol{\phi}})$ is a Lipschitz-continuous function, i.e.,
\begin{align} \label{I_phi_bounded_rewrite}
I_{\boldsymbol{\phi}}=\| q({\boldsymbol{\phi}}^{(1)}) - q({\boldsymbol{\phi}}^{(2)})\| \leq  \lambda_q \| {\boldsymbol{\phi}}^{(1)} - {\boldsymbol{\phi}}^{(2)}\|.
\end{align}

Now, using \eqref{P_lip_intro_split}, along with \eqref{I_z_bounded_rewrite} and \eqref{I_phi_bounded_rewrite}, we obtain
\begin{align} \label{P_lip_intro_split_rewrite}
\| P({\bz}^{(1)}, {\boldsymbol{\phi}^{(1)}}) - P({\bz}^{(2)}, {\boldsymbol{\phi}^{(2)}}) \| &\leq I_{\bz} +I_{\boldsymbol{\phi}}\leq    \lambda_f \lambda_g \| \bz^{(1)} - \bz^{(2)}\|+ \lambda_q \| {\boldsymbol{\phi}}^{(1)} - {\boldsymbol{\phi}}^{(2)}\|\nonumber\\
&\leq  \tilde{\lambda}_p \left( \| \bz^{(1)} - \bz^{(2)}\|+ \| {\boldsymbol{\phi}}^{(1)} - {\boldsymbol{\phi}}^{(2)}\|\right) \nonumber\\
&\leq  \tilde{\lambda}_p \left( \|\tilde{ \bz}^{(1)} - \tilde{\bz}^{(2)}\|+ \| \tilde{ \bz}^{(1)} - \tilde{ \bz}^{(2)}\|\right)\nonumber\\
&=  {\lambda}_p \left( \|\tilde{ \bz}^{(1)} - \tilde{\bz}^{(2)}\|\right),
\end{align}
where $\tilde{\lambda}_p \triangleq \max \left\{  \lambda_f \lambda_g,   \lambda_q \right\}$, and ${\lambda}_p \triangleq 2 \tilde{\lambda}_p$. This now completes the proof.
\end{proof}
\end{appendices}

\begin{multicols}{2}
\begin{IEEEbiography}[{\includegraphics[width=1in,height=1.25in,clip, keepaspectratio]{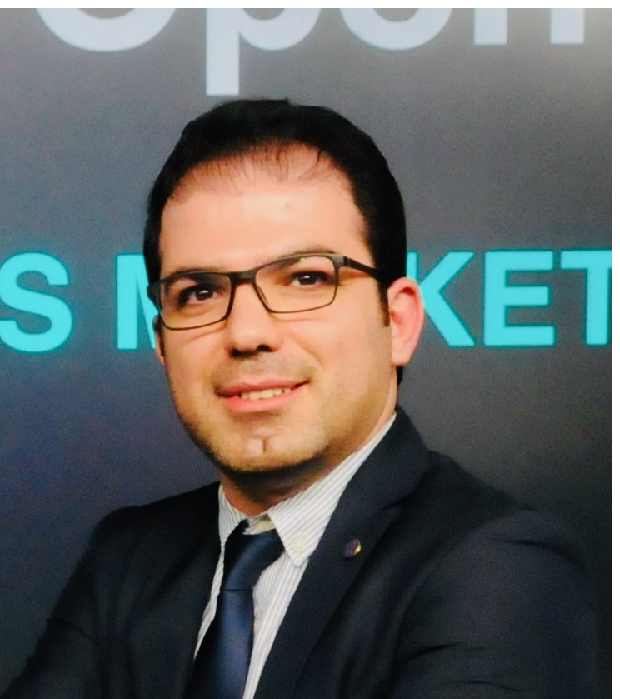}}]{Javad Mirzaei}
(M'18) was born in Tehran, Iran. He received the B.Sc. degree from Iran University of Science and Technology (IUST), Tehran, Iran, in 2010, the M.Sc. degree from Ontario Tech University, ON, Canada, in 2013, and the PhD degree from the University of Toronto, ON, Canada, in 2021, all in electrical and computer engineering. He is now a Post-Doctoral Fellow with the University of Toronto. From Feb. 2014 to Sept. 2015, he was with the Cable Shoppe Inc., Toronto, ON, Canada, where he was working in the area of broadband communication systems and IP-based 4G systems. His main research interests include wireless communication, signal processing, and machine learning.
\end{IEEEbiography}

\begin{IEEEbiography}[{\includegraphics[width=1in,height=1.25in,clip,keepaspectratio]{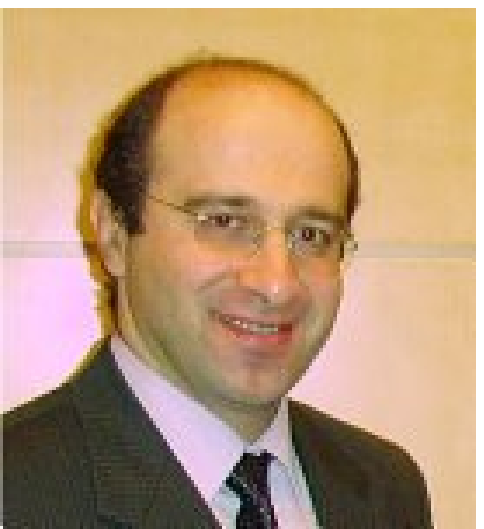}}]
{Shahram Shahbazpanahi} (M'02, SM'10) was born in Sanandaj, Kurdistan, Iran. He received the B.Sc., M.Sc., and Ph.D. degrees in electrical engineering from Sharif University of Technology, Tehran, Iran, in 1992, 1994, and 2001, respectively. From September 1994 to September 1996, he was an instructor   with the Department of Electrical Engineering, Razi University, Kermanshah, Iran. From July 2001 to March 2003, he was a Postdoctoral Fellow with the Department of Electrical and Computer Engineering, McMaster University, Hamilton, ON, Canada. From April 2003 to September 2004, he was a Visiting Researcher with the Department of Communication Systems, University of Duisburg-Essen, Duisburg, Germany. From September 2004 to April 2005, he was a Lecturer and Adjunct Professor with the Department of Electrical and Computer Engineering, McMaster University. In July 2005, he joined the Faculty of Engineering and Applied Science, University of Ontario Institute of Technology, Oshawa, ON, Canada, where he currently holds a  Professor position. His research interests include statistical and array signal processing; space-time adaptive processing; detection and estimation; multi-antenna, multi-user, and cooperative communications; spread spectrum techniques; DSP programming; and hardware/real-time software design for telecommunication systems. Dr. Shahbazpanahi has served as an Associate Editor for the IEEE TRANSACTIONS ON SIGNAL PROCESSING and the IEEE SIGNAL PROCESSING LETTERS. He has also served as a Senior Area Editor for the IEEE SIGNAL PROCESSING LETTERS. He was an elected  member of the Sensor Array and Multichannel (SAM) Technical Committee of the IEEE Signal Processing Society. He has received several awards, including the Early Researcher Award from Ontario's Ministry of Research and Innovation, the NSERC Discovery Grant (three awards), the Research Excellence Award from the Faculty of Engineering and Applied Science, the University of Ontario Institute of Technology, and the Research Excellence Award, Early Stage, from the University of Ontario Institute of Technology.
\end{IEEEbiography}

\begin{IEEEbiography}[{\includegraphics[width=1in,height=1.25in,clip,keepaspectratio]{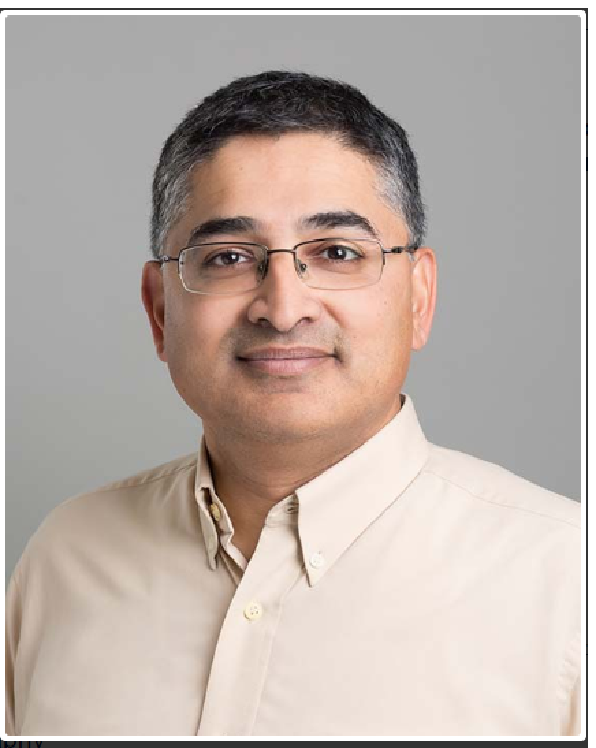}}]{Raviraj S. Adve} (S'88, M'97, SM'06, F'17) was born in Bombay, India. He received his B. Tech. in Electrical Engineering from IIT, Bombay, in 1990 and his Ph.D. from Syracuse University in 1996, His thesis received the Syracuse University Outstanding Dissertation Award. Between 1997 and August 2000, he worked for Research Associates for Defense Conversion Inc. on contract with the Air Force Research Laboratory at Rome, NY. He joined the faculty at the University of Toronto in August 2000 where he is currently a Professor. Dr. Adve’s research interests include analysis and design techniques for cooperative and heterogeneous networks, energy harvesting networks and in signal processing techniques for radar and sonar systems. He received the 2009 Fred Nathanson Young Radar Engineer of the Year award. Dr. Adve is a Fellow of the IEEE.
\end{IEEEbiography}
\begin{IEEEbiography}[{\includegraphics[width=1in,height=1.25in,clip,keepaspectratio]{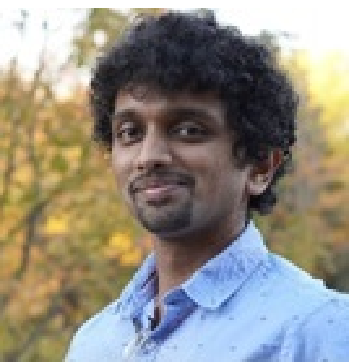}}]{Navaneetha Krishna Madan Gopal} was born in Bangalore, India. He received his Bachelor of Engineering (B.E.) in Electronics and Communication Engineering from M S Ramaiah Institute of Tecchnology, Bangalore, India in 2018 and his Master of Engineering (M. Eng.) from the University of Toronto, Canada in 2020. During his Masters’ degree, he worked as a Research Assistant under Prof. Raviraj S. Adve, developing an algorithm to achieve FDD massive MIMO with minimal feedback. He now works at TELUS Communications in Toronto as a part of the spectrum team and represents the company in 3GPP standards body meetings.
\end{IEEEbiography}

\end{multicols}

\end{document}